\documentclass[acmsmall,screen]{acmart}

\usepackage{comment}
\excludecomment{cameraready}
\includecomment{techreport}

\begin{techreport}
\end{techreport}
\begin{cameraready}
\setcopyright{rightsretained}
\acmPrice{}
\acmDOI{10.1145/3236769}
\acmYear{2018}
\copyrightyear{2018}
\acmJournal{PACMPL}
\acmVolume{2}
\acmNumber{ICFP}
\acmArticle{74}
\acmMonth{9}
\end{cameraready}
\setcopyright{rightsretained}
\copyrightyear{2018}

\setcitestyle{nosort}

\usepackage{amsmath, amssymb, amsthm}
\usepackage[mathscr]{euscript}
\usepackage{mathtools, stmaryrd}
\usepackage{graphicx}
\usepackage{booktabs}   
\usepackage{subcaption}
\usepackage{mathpartir}
\usepackage{listings}
\usepackage{tikz}
\usepackage{longtable}
\usepackage{thm-restate}
\usepackage[normalem]{ulem}

\usetikzlibrary{positioning}
\usepackage[scaled=0.75]{beramono} 
\usepackage{wrapfig}

\definecolor{mygreen}{rgb}{0,0.6,0}
\lstdefinelanguage{links}{
  language=C,
  basicstyle={\linespread{0.8}\ttfamily}, 
  tabsize=3,
  commentstyle=\color{mygreen},
  morekeywords={lens,lensselect,lensjoin,lensdrop,get,put,table,var,database,where,with,tablekeys,from,determined,by,on,insert,into,values},
}

\newcommand*{\Secref}[1]{\S\,\ref{sec:#1}}
\newcommand*{\SecrefTwo}[2]{\S\S\,\ref{sec:#1} and \ref{sec:#2}}

\newcommand*{\opname}[1]{\operatorname{\mathit{#1}}}

\newcommand*{\Genrecrevise}{\opname{recrevise}}

\newcommand*{\Relmerge}{\opname{merge}}

\newenvironment{salign}
   {\par\nobreak\small\noindent\csname align*\endcsname}
   {\csname endalign*\endcsname}

\newenvironment{sflalign}
   {\par\nobreak\small\setcounter{equation}{0}\noindent\csname flalign\endcsname}
   {\csname endflalign\endcsname}

\newcommand*{\newrecrevise}[3]{\Genrecrevise_{#2}(#1,#3)}
\newcommand*{\relrevise}[3]{\opname{revise}_{#2}(#1,#3)}

\newcommand*{\recupdate}[2]{{#1} \leftplusarrow {#2}}
\newcommand*{\recproj}[2]{{#1}[#2]}

\newcommand*{\relmerge}[3]{\opname{merge}_{#2}(#1,#3)}
\newcommand*{\proj}[2]{\pi_{#2}{(#1)}}

\newcommand*{\param}{\cdot}
\newcommand*{\dom}{\opname{dom}}
\newcommand*{\getOp}{\opname{get}}
\newcommand*{\putOp}{\opname{put}}
\newcommand*{\id}{\opname{id}}
\newcommand*{\sym}{\opname{sym}}
\newcommand*{\assoc}{\opname{assoc}}

\newcommand*{\defeq}{\stackrel{\smash{\text{\tiny def}}}{=}}
\newcommand*{\domsubtract}[2]{{#1}\backslash_{#2}}

\newcommand*{\select}[2]{\sigma_{#1}(#2)}
\newcommand*{\rename}[2]{\rho_{#1}(#2)}
\newcommand*{\set}[1]{\{#1\}}
\newcommand*{\record}[1]{\{#1\}}

\newcommand*{\positive}[1]{{#1}^{+}}
\newcommand*{\negative}[1]{{#1}^{-}}
\newcommand{\pDelta}[1]{\positive{\Delta #1}}
\newcommand{\nDelta}[1]{\negative{\Delta #1}}

\newcommand*{\letin}[2]{\texttt{let}\;{#1} = {#2}\;\texttt{in}}
\newcommand*{\letexpr}[3]{\letin{#1}{#2}\;{#3}}

\newcommand*{\Deltify}[1]{\dot{#1}}
\newcommand*{\deltify}[1]{{\delta\hspace{-0.15em}#1}} 
\newcommand*{\deltaOp}[1]{\deltify{\opname{#1}}}
\newcommand*{\deltaselect}[2]{\Deltify{\sigma}_{#1}(#2)}
\newcommand*{\deltarename}[2]{\Deltify{\rho}_{#1}(#2)}
\newcommand*{\deltasetminus}{\mathrel{\Deltify{\setminus}}}
\newcommand*{\deltaproj}[2]{\Deltify{\pi}_{#2}{(#1)}}
\newcommand*{\deltaJoin}{\mathrel{\Deltify{\Join}}}

\newcommand*{\deltarelrevise}[3]{\deltaOp{revise}_{#2}(#1,#3)}
\newcommand*{\deltarelmerge}[3]{\deltaOp{merge}_{#2}(#1,#3)}

\let\oldoplus\oplus
\let\oplus\undefined
\let\oldominus\ominus
\let\ominus\undefined

\newcommand*{\deltaminus}{\oldominus}
\newcommand*{\deltaplus}{\oldoplus}
\newcommand*{\deltaapp}{\oldoplus}

\newcommand{\hd}[1]{\textbf{#1}}
\newcommand{\tx}[1]{\texttt{`#1'}}
\newcommand{\ch}[1]{\textcolor{red}{\textbf{#1}}}

\newcommand{\vtype}[3]{\mathit{Rel}(#1,#2,#3)}
\newcommand{\reltype}[1]{\mathit{Rel}(#1)}

\newcommand*{\Right}{\opname{right}}
\newcommand*{\Left}{\opname{left}}
\newcommand*{\roots}{\opname{roots}}
\newcommand*{\outputs}{\opname{outputs}}
\newcommand*{\affected}{\opname{affected}}

\newcommand*{\SelectPrim}[1]{\texttt{select}_{#1}}
\newcommand*{\DropPrim}[3]{\texttt{drop}\;{#1}\;\texttt{determined\;by}\;(#2,#3)}
\newcommand*{\JoinDLPrim}{\texttt{join\_dl}}
\newcommand*{\JoinDRPrim}{\texttt{join\_dr}}
\newcommand*{\JoinBothPrim}{\texttt{join\_both}}

\newcommand*{\IdPrim}[1]{\texttt{id}_{#1}}
\newcommand*{\SymPrim}[1]{\texttt{sym}_{#1}}
\newcommand*{\AssocPrim}[1]{\texttt{assoc}_{#1}}
\newcommand*{\RenamePrim}[1]{\texttt{rename}_{#1}}

\newcommand*{\DeltaSelectPrim}[1]{\deltaOp{\texttt{select}}_{#1}}
\newcommand*{\DeltaDropPrim}[3]{\deltaOp{\texttt{drop}}\;{#1}\;\texttt{determined\;by}\;(#2,#3)}
\newcommand*{\DeltaJoinDLPrim}{\deltaOp{\texttt{join\_dl}}}
\newcommand*{\DeltaRenamePrim}[1]{\deltaOp{\texttt{rename}}_{#1}}

\newcommand*{\lensto}{\Leftrightarrow}

\newcommand*{\leftplusarrow}{\mathop{\leftarrow \negthickspace \negthickspace \negmedspace +}}
\renewcommand*{\emptyset}{\varnothing}
\renewcommand*{\setminus}{-}

\newenvironment{mathfig}{\begin{displaymath}}{\end{displaymath}}
\newenvironment{syntaxfig}{\begin{mathfig}\begin{array}{@{}l@{\quad}r@{~~}c@{\quad}ll}}{\end{array}\end{mathfig}}

\newcommand*{\tblname}[1]{\emph{#1}}
\newcommand*{\tblfd}[2]{
	\begin{array}{c}
		#1 \\ \text{with } \left\{ #2 \right\}
	\end{array}
}
\newcommand*{\colname}[1]{\emph{#1}}
\newcommand{\ms}[1]{$#1 ms$}

\newcommand*{\proofContext}[1]{\def\currentprefix{proof:#1}}
\newcommand*{\locallabel}[1]{\label{\currentprefix:#1}}
\newcommand*{\localref}[1]{\ref{\currentprefix:#1}}

\newcommand*{\lift}[1]{{#1}^\dagger}
\newcommand*{\DeltaOp}[1]{\lift{\opname{#1}}}
\newcommand*{\Deltaselect}[2]{\lift{\sigma}_{#1}(#2)}
\newcommand*{\Deltarename}[2]{\lift{\rho}_{#1}(#2)}
\newcommand*{\Deltasetminus}{\mathrel{\lift{\setminus}}}
\newcommand*{\Deltaproj}[2]{\lift{\pi}_{#2}{(#1)}}
\newcommand*{\DeltaJoin}{\mathrel{\lift{\Join}}}

\newcommand*{\Deltarelmerge}[3]{\lift{\opname{merge}}_{#2}(#1,#3)}

\newcommand*{\qedLocal}{\hookrightarrow}

\newcommand*{\interCaseZero}[1]{%
& \uline{\textbf{\emph{Case}}\ #1:}& %
\\[0.25em]%
}
\newcommand*{\interCaseOne}[1]{%
& \indent \uline{\textbf{\emph{Case}}\ #1:}& %
\\[0.25em]%
}
\newcommand*{\interCaseTwo}[1]{%
& \indent\indent \uline{\textbf{\emph{Case}}\ #1:}& %
\\[0.25em]%
}

\newcommand*{\interSubcaseZero}[1]{%
& \uline{#1}:& %
\\[0.25em]%
}
\newcommand*{\interSubcaseOne}[1]{%
& \indent \uline{#1}:& %
\\[0.25em]%
}

\DeclareFontFamily{U}{mathb}{\hyphenchar\font45}
\DeclareFontShape{U}{mathb}{m}{n}{
      <5> <6> <7> <8> <9> <10>
      <10.95> <12> <14.4> <17.28> <20.74> <24.88>
      mathb10
      }{}
\DeclareSymbolFont{mathb}{U}{mathb}{m}{n}

\DeclareMathSymbol{\sqbullet}{1}{mathb}{"0D}

\newcommand{\minimal}[2]{#1 \text{ minimal for~} #2}

\newcommand{\reflemma}[1]{lem.~\ref{lem:#1}}
\newcommand{\negDM}{\negative{\Delta M}}
\newcommand{\negDN}{\negative{\Delta N}}
\newcommand{\posDM}{\positive{\Delta M}}
\newcommand{\posDN}{\positive{\Delta N}}

\newcommand{\Derive}[1]{\delta(#1)}
\newcommand{\Dagger}[1]{(#1)^\dagger}
\newcommand{\opmap}[2]{\opname{map}_{#1}(#2)}
\newcommand{\dopmap}[2]{\deltaOp{map}_{#1}(#2)}

\allowdisplaybreaks

\begin{document}

\title{Incremental Relational Lenses}

\author{Rudi Horn}
\orcid{0000-0001-6164-2208}
\affiliation{
  \institution{University of Edinburgh}            
  \country{United Kingdom}                    
}
\email{r.horn@ed.ac.uk}

\author{Roly Perera}
\orcid{0000-0001-9249-9862}
\affiliation{
  \institution{University of Edinburgh and University of Glasgow}            
  \country{United Kingdom}                    
}
\email{roly.perera@ed.ac.uk}

\author{James Cheney}
\orcid{0000-0002-1307-9286}
\affiliation{
  \institution{University of Edinburgh}            
  \country{United Kingdom}                   
}
\email{jcheney@inf.ed.ac.uk}

\begin{abstract}

  Lenses are a popular approach to bidirectional transformations, a
  generalisation of the \emph{view update} problem in databases, in
  which we wish to make changes to \emph{source} tables to effect a
  desired change on a \emph{view}.  However, perhaps surprisingly,
  lenses have seldom actually been used to implement updatable views
  in databases. Bohannon, Pierce and Vaughan proposed an approach to
  updatable views called \emph{relational lenses}, but to the best of
  our knowledge this proposal has not been implemented or evaluated to
  date.  We propose \emph{incremental relational lenses}, that equip
  relational lenses with change-propagating semantics that map small
  changes to the view to (potentially) small changes to the source
  tables.  We also present a language-integrated implementation of
  relational lenses and a detailed experimental evaluation, showing
  orders of magnitude improvement over the non-incremental
  approach. Our work shows that relational lenses can be used to
  support expressive and efficient view updates at the language level,
  without relying on updatable view support from the underlying
  database.

\end{abstract}

\ccsdesc[500]{Software and its engineering~Domain specific languages}.
\ccsdesc[500]{Information systems~Database views}
\ccsdesc[300]{Information systems~Relational database query languages}

\keywords{lenses, bidirectional transformations, relational calculus, incremental computation}  

\maketitle

\section{Introduction}

A typical web application is based on a three-tier architecture with user
interaction on the client (web browser), application logic on the server, and
data stored in a (typically relational) database. Relational databases employ
SQL query and update expressions, including \emph{projections},
\emph{selections} and \emph{joins}, which correspond closely to familiar list
comprehension operations in functional languages.  Relational databases offer
data persistence and high performance and are flexible enough for a range of
applications. However, the \emph{impedance mismatch} between database queries
and conventional programming makes even simple programming tasks
challenging~\cite{copeland84smalltalk}.  Languages such as
C\#~\cite{meijer2006sigmod}, F\#~\cite{syme2006linq,cheney2013linq},
Links~\cite{cooper2006links}, and Ur/Web~\cite{chlipala2015urweb}, and libraries
such as Database-Supported Haskell~\cite{ulrich2015dsh}, have partly overcome
this challenge using \emph{language-integrated query}, in which query
expressions are integrated into the host language and type system.

However, language support for database programming is still
incomplete.  For example, \emph{views} are a fundamental concept in
databases~\cite{ramakrishnan2003database} that are typically not
supported in language-integrated query.  A view is a relation defined
over the database tables using a relational query.  Views have many
applications:
\begin{enumerate}
\item A \emph{materialised} view can \emph{precompute} query answers, avoiding
expensive recomputation;
\item a \emph{security} view shows just the information a user needs, while
hiding confidential data;
\item finally, views can be used to define the \emph{external schema} of a
database, presenting the data in a form convenient for a particular application,
or for \emph{integrating} data from different databases.
\end{enumerate}

Most databases allow specifying views using a variant of table
creation syntax, and views can then be used in much the same way as
regular database tables.  It is therefore natural to wish to
\emph{update} a view; for example, making a security view updatable
would make sense if the user is intended to have write access to the
view data but should not have write access to the underlying table.
Unfortunately, it is nontrivial to update views.  Some view updates
may correspond to multiple possible updates to the source tables,
while others may not be translatable at
all~\cite{dayal1982correct,bancilhon1981update}.  Most relational
database systems only allow selection and projection operations in
updatable views, so updating views defined using joins is not allowed.

\emph{Lenses} were introduced by~\citet{foster2007combinators} as a
generalisation of updatable views to arbitrary data structures.  A lens is a
pair of functions $\getOp : S \to V$ and $\putOp : S\times V \to S$, subject to
the following \emph{round-tripping} or \emph{well-behavedness} laws:
\[\getOp(\putOp(s,v)) = v \qquad \putOp(s,\getOp(s)) = s\]
Lenses are particularly well-suited to programming tasks where it is necessary
to maintain consistency between `the same' data stored in different places, as
often arises in web programming.  A great deal of research on
\emph{bidirectional programming} has considered this problem, especially in the
functional programming
community~\cite{stevens2007landscape,bohannon2006relational,hidaka2010unql,foster2007combinators,foster10ssgip,diskin2011delta,hofmann2011symmetric,hofmann2012edit,wang2011incremental}.

Perhaps surprisingly, relatively little of this work has considered view updates
in databases. The most important exception is the work of
\citet{bohannon2006relational}, who proposed lens combinators for projections,
selections and joins on relational data, and proved their well-behavedness.
These \emph{relational lenses} are defined using $\putOp$ functions which map
the source database and an updated view to the updated source database.
\citeauthor{bohannon2006relational}'s work showed that it is possible in
principle for databases to support updatable views including joins, provided the
type system tracks \emph{integrity constraints} on the data, such as functional
dependencies.

However, to the best of our knowledge, the practicality of relational lenses has
not been demonstrated.  The proposed definitions of $\getOp$ and $\putOp$ are
\emph{state-based} --- they showed how to compute the view from the base table
state, and how to compute the entire new state of the base tables from the
updated view and the old table state.  These definitions suggest an obvious, if
naive, implementation strategy: computing the new source table contents and
replacing the old contents.  This is simply impractical for any realistic
database, and is usually hugely wasteful, in the common case when updates affect
relatively few records.  Replacing source tables would also necessitate locking
access to the affected tables for long periods, destroying any hope of
concurrent access.

Luckily, replacing entire source tables with their new contents is
seldom necessary.  The reason is that updates to tables (and views)
are often small: for example, a row might be inserted or deleted, or a
single field value modified.  Indeed, there is a large literature on
the problem of \emph{incremental view
  maintenance}~\cite{gupta1995maintenance,koch2010incremental,koch2016incremental} addressing the problem of
how to modify a materialised view to keep it consistent with changes
to the source tables.
\begin{figure}
	\includegraphics[width=0.6\textwidth]{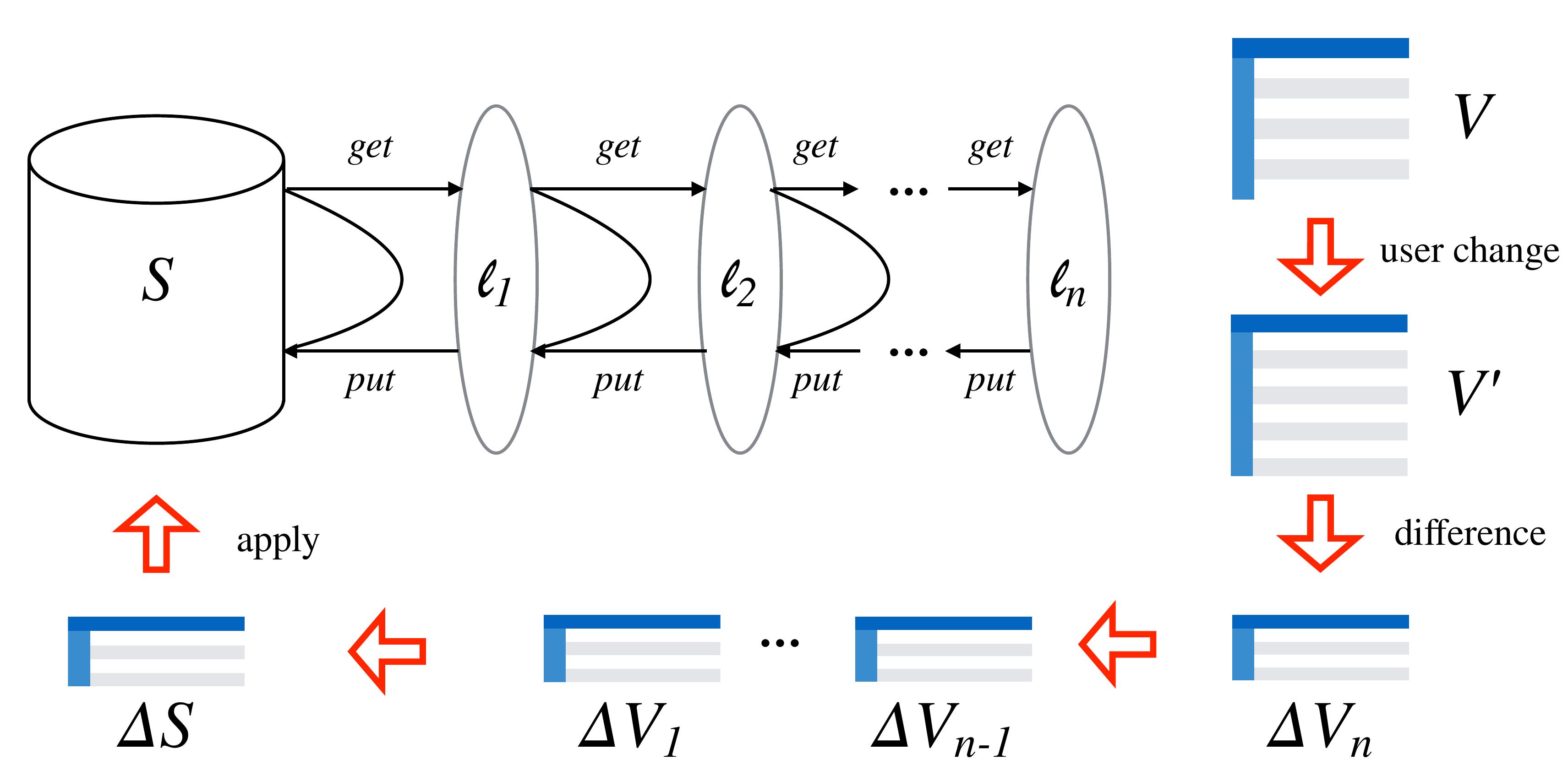}
	\caption{Propagating changes through lenses from view to
          source.}
	\label{fig:incremental_get_put}
\end{figure}
The benefits of incremental evaluation are not
confined to databases, either: witness the growing literature on
\emph{adaptive} and \emph{incremental} functional
programming~\cite{acar2006toplas,hammer2014pldi,cai2014pldi}.  Indeed,
the foundations of change-oriented bidirectional transformations have even been
investigated in the form of \emph{edit lenses}~\cite{hofmann2012edit}
and other formalisms.  It  is natural to ask whether
incrementalisation can be used to make relational lenses practical.

In this paper, we propose \emph{incremental relational lenses}.
Figure~\ref{fig:incremental_get_put} illustrates our approach.  Given a lens
$\ell$ (defined by composing several primitive lenses $\ell_i$), and the initial
view value $V = \getOp_\ell(S)$, suppose $V$ is updated to $V'$.  We begin by
calculating a \emph{view delta} (i.e. change set) $\Delta V = V' \deltaminus V$.
Here $\deltaminus$ is the operation that calculates a delta mapping one value to
another.  Then, for each step $\ell_i$ of the definition of $\ell$, we translate
the view delta $\Delta V_i$ of $V_i$ to a source delta $\Delta V_{i-1}$.  We do
this by defining an incremental version of the $\putOp$ operation,
$\deltaOp{\putOp}$, which takes $S$ and $\Delta V$ as arguments, and which
satisfies the following law:
\[\putOp(S,\getOp(S) \deltaapp \Delta V) = S \deltaapp \deltaOp{\putOp}(S,\Delta V)
\]
where $\deltaapp$ denotes the application of a delta to a value.
Finally, once we have calculated the source delta $\Delta S = \Delta V_0$, we translate it to a
sequence of SQL \texttt{INSERT}, \texttt{UPDATE} and \texttt{DELETE}
commands.

Our approach avoids recomputing and replacing entire
tables.  Moreover, it can often translate small view
deltas to small source deltas.  Working with small deltas reduces the amount of
computation and data movement incurred.  On the other hand, incremental
relational lenses still may need to access the source tables to compute
correct deltas.  We show that this can be done efficiently by issuing
auxiliary queries during delta propagation.

We have implemented incremental relational lenses in
Links~\cite{cooper2006links}, a web programming language with comprehensive
support for language-integrated queries.  Our experiments show that incremental
evaluation offers dramatic performance benefits over the naive state-based
approach, just as one would hope or even expect.  Perhaps more importantly, we
prove the correctness of our approach.  The state-based relational lens
definitions have a number of subtleties, and proving the correctness of their
incremental versions is a nontrivial challenge. Since relational lenses use
set-based rather than multiset-based semantics, recent work by
\citet{koch2010incremental} and \citet{cai2014pldi} on incrementalising multiset
operations does not apply; instead, we build on classical work on incremental
view updates~\cite{qian91tkde,griffin1997incremental}.  Incremental relational
lenses are also related to edit lenses \cite{hofmann2012edit} and some other
frameworks; we discuss the relationship in detail in Section~\ref{sec:related}.

\paragraph{Outline and contributions}
In the rest of this paper, we present necessary background on
relational lenses and incrementalisation, define and prove the
correctness of incremental relational lenses, and empirically validate
our implementation to establish practicality.
\begin{itemize}
\item \Secref{overview} shows how relational lenses are integrated into Links, and illustrates our incremental semantics for relational lenses via examples.
\item \Secref{background} presents background from
  \citet{bohannon2006relational}, including auxiliary concepts and
  their state-based definitions of relational lenses.
\item \Secref{framework} introduces our framework for incremental relational queries.
\item \Secref{incremental} presents our main contribution,
  incremental relational lenses, along with proofs of correctness of
  optimised forms of their $\deltaOp{put}$ operations.
\item \Secref{evaluation} presents our experimental results.
\item \SecrefTwo{related}{concl} discuss related work and present our conclusions.
\end{itemize}

\section{Relational lenses by example}\label{sec:overview}

\lstset{language=links}

In this section we illustrate the use of relational lenses as integrated into
Links, a web and database programming language~\cite{cooper2006links}.
We first
illustrate the naive implementation of relational lenses and then show the
incremental approach.  We use an example from \citet{bohannon2006relational}
involving a small database of albums and tracks, and an updatable view that can
be defined over it using relational lenses.  In Links, it is straightforward to
define a web-based user interface for editing such a view; we elide those
details.  We also suppress the technical details of relational lenses and
incrementalisation until later sections.

In Links, we can initialise a database connection using the
\lstinline{database}\footnote{Information such as the hostname, port, username
and password needed to access the database also needs to be provided; this can
be included in the \lstinline{database} expression or, as here, in a separate
configuration file.} expression, e.g.:

\begin{lstlisting}
var db = database "music_database";
\end{lstlisting}

\noindent
We can then define table references that correspond to actual database
tables, using the \texttt{table} expression; for example,
we can define references to the tables ``albums'' and
``tracks'' as follows:

\begin{lstlisting}
var albumsTable = table "albums"
                  with (album: String, quantity: Int)
                  tablekeys [["album"]] from db;

var tracksTable = table "tracks"
                  with (track: String, date: Int, rating: Int, album: String)
                  tablekeys [["track", "album"]] from db;
\end{lstlisting}
In a \lstinline{table} expression, the table name and names of fields
and their types are given, and the \lstinline{tablekeys} clause
provides a list of lists of field names.  Each element of this list
should be a \emph{key} in the sense that the combined field values are
unique in the table; for example, in \lstinline{albumsTable} there is
at most one row with a given \lstinline{album} field, whereas in
\lstinline{tracksTable} the track and album fields together uniquely
identify a row.  Key constraints are a special kind of
\emph{functional dependency}, discussed further below.

Links supports language integrated queries and updates over base tables.  We
have extended Links with updatable views based on relational lenses. Relational
lens definitions require the specification of \emph{functional dependencies} for
base tables, constraints that indicate which attributes determine the value of
other attributes.  As an example, a functional dependency $album \to quantity$
says that if two rows have the same \lstinline{album} attribute, they must also
have the same \lstinline{quantity} value.

The \lstinline{lens} keyword makes a base table into a basic lens, and
allows specifying additional functional dependencies beyond key
constraints~\footnote{Such constraints are not currently checked in
  our implementation, but this is straightforward and an orthogonal
  concern.}. We define basic lenses for each of our tables: the
\emph{albums} table has a functional dependency $album \to quantity$,
and the \emph{tracks} table has a functional dependency
$track \to date\ rating$, which says that $date$ and $rating$ depend
on $track$.  This implies that any track may appear in different
albums, but should have the same $date$ and $rating$.  (These
functional dependencies are as specified by
\citet{bohannon2006relational}, but result in a database with some
redundancy; however, we keep the example as is for ease of comparison
with prior work.)

\begin{lstlisting}
var albumsLens = lens albumsTable with album -> quantity;
var tracksLens = lens tracksTable with track -> date rating;
\end{lstlisting}

A common workflow within a web application is to extract some view of
the data from the database, associate it with a form, and map form
responses to updated versions of the associated data.  If the mapping
from the database to the form data is defined by a lens, then the view
can be fetched from the lens using the $\getOp$ operation.  When the
user submits the form, the $\putOp$ operation of the lens should
allow us to propagate the changes to the underlying database.  So,
we can add a row to $albumsTable$ as follows:
\begin{lstlisting}
var albums = get albumsLens;
var newAlbums = albums ++ [(album='Disintegration', quantity=1)]
put albumsLens with newAlbums;
\end{lstlisting}
Because \lstinline{albumsLens} is simply a basic lens wrapping
\lstinline{albumsTable}, this is equivalent to just using a
SQL-style update, written in Links as follows:
\begin{lstlisting}
insert into albumsTable values [(album='Disintegration', quantity=1)]
\end{lstlisting}

In practice, views are more useful for selecting subsets of data or
combining tables.  For example, the web application might show a form
allowing updates to a single album at a time, such as \tx{Galore}.
The updatable form data can be extracted using a select lens on
\emph{tracksLens} which requires rows to have $album = \tx{Galore}$.
We then call $\getOp$ on the select lens, make any desired changes and
then use $\putOp$ on the select lens and the new view to update the
database.

\begin{lstlisting}
/* create a lens which selects only the tracks from 'Galore' */
var selectLens = lensselect from tracksLens where album = 'Galore';
/* get the smaller subset of tracks */
var tracks = get selectLens;
/* ... newTracks = updates to tracks ... */
/* update database */
put selectLens with newTracks;
\end{lstlisting}

Suppose \lstinline{tracksTable} contains the entries specified in Figure
\ref{fig:select_lens_get_example} on the left. Calling $\getOp$ on the lens
produces the view on the right, containing only the records having
$album=\tx{Galore}$. If the user changes the rating of the track \tx{Lullaby} to
4, then submits the form, an updated view is generated as shown on the right in
Figure \ref{fig:select_lens_put_example}.  The application can then call
$\putOp$ on the view, which will cause the underlying \emph{tracksTable} table
to be updated with the changes to the \tx{Lullaby} tracks.  Notice that we must
change both tracks because of the functional dependency $track \to date~rating$.
This will produce the updated table as shown on the left in Figure
\ref{fig:select_lens_put_example}.

\begin{figure}[tb]
	\begin{equation*}
		\begin{array}{c|cccc}
			& \hd{track} & \hd{date} & \hd{rating} & \hd{album} \\
			\hline
			& \tx{Lullaby} & 1989 & 3 & \tx{Galore} \\
			& \tx{Lullaby} & 1989 & 3 & \tx{Show} \\
			& \tx{Lovesong} & 1989 & 5 & \tx{Galore} \\
			& \tx{Lovesong} & 1989 & 5 & \tx{Paris} \\
			& \tx{Trust} & 1992 & 4 & \tx{Wish} \\
		\end{array}
		\stackrel{\smash{\text{\tiny{$\getOp$}}}}{\Rightarrow}
		\begin{array}{c|cccc}
			& \hd{track} & \hd{date} & \hd{rating} & \hd{album} \\
			\hline
			& \tx{Lullaby} & 1989 & 3 & \tx{Galore} \\
			& \tx{Lovesong} & 1989 & 5 & \tx{Galore} \\
		\end{array}
	\end{equation*}
	\caption{Select lens example: computing the view (right) from the
          source (left) using $\getOp$ }
	\label{fig:select_lens_get_example}
	\begin{equation*}
		\begin{array}{c|cccc}
			& \hd{track} & \hd{date} & \hd{rating} & \hd{album} \\
			\hline
			& \tx{Lullaby} & 1989 & \ch{4} & \tx{Galore} \\
			& \tx{Lullaby} & 1989 & \ch{4} & \tx{Show} \\
			& \tx{Lovesong} & 1989 & 5 & \tx{Galore} \\
			& \tx{Lovesong} & 1989 & 5 & \tx{Paris} \\
			& \tx{Trust} & 1992 & 4 & \tx{Wish} \\
		\end{array}
		\stackrel{\smash{\text{\tiny{$\putOp$}}}}{\Leftarrow}
		\begin{array}{c|cccc}
			& \hd{track} & \hd{date} & \hd{rating} & \hd{album} \\
			\hline
			& \tx{Lullaby} & 1989 & \ch{4} & \tx{Galore} \\
			& \tx{Lovesong} & 1989 & 5 & \tx{Galore} \\
		\end{array}	\end{equation*}
              \caption{Select lens example: computing the new source
                (left) using $\putOp$ on the new view
                (right) and old source (from
                Fig.~\ref{fig:select_lens_get_example}, left).  The change
                to the view results in two changes to the
                source.}
	\label{fig:select_lens_put_example}
\end{figure}

\paragraph{Type and Integrity Constraints}
Both views and base tables can be associated with integrity
constraints, and updated views need to respect these constraints.
There are three kinds of constraints:

\begin{enumerate}
\item The updated view should be well-typed in the usual sense.  Views
  that have rows with extra or missing fields, or field values of the wrong
  types, are ruled out statically by the type system.
\item The updated view should satisfy the functional dependencies
  associated with the view. Thus, the functional dependency
  $track \to  date\ rating$ from our example implies that we cannot
  change the rating or date of \tx{Lullaby} in one row without
  changing all the others to match.
\item The updated view may also need to satisfy a predicate on the
  rows.  Views defined by lenses may have selection conditions, such
  as $album = \tx{Galore}$.  Inserting rows with other $album$ values,
  or changing this field in existing rows, is not allowed.
\end{enumerate}
These constraints all originate in the definition of \emph{schemas}
for relational lenses introduced by \citet{bohannon2006relational}.  The correctness
properties of relational lenses rely on these integrity constraints,
and if the updated view
satisfies its constraints then the updated underlying table will
also satisfy its own constraints.

\paragraph{Incremental View Updates}

\citet{bohannon2006relational} define the $\getOp$ and $\putOp$
directions of relational lenses as set-theoretic expressions showing
how to compute the new source given the old source and the updated
view.  The most obvious approach to implementing the $\putOp$ behavior
of a relational lens is to use these definitions to calculate the
new source table `from scratch' and replace the old one with the new
one.  For example, in Figure \ref{fig:select_lens_put_example} this
would mean deleting all five tuples of the old tracks table and then
re-inserting the three unchanged tuples and the new versions of the
two modified ones.  We could accomplish the desired effect in
Figure~\ref{fig:select_lens_put_example} using SQL UPDATE operations
to change just the ratings of the \tx{Lullaby} tracks to 4. This would
typically be more efficient (especially if there were many more
unaffected rows).

Therefore, we adopt an incremental approach, as outlined in the introduction.
Instead of working with the entire tables, we first compute a delta for the
modified view, that is, sets of rows to be inserted and deleted. We illustrate
deltas as tables with rows annotated with `+' (for insertion) or `-' (for
deletion).  An example delta for the update shown in
Figure~\ref{fig:select_lens_put_example} is shown on the left of
Figure~\ref{fig:select_lens_put_delta_example}.  This delta is then used to
calculate a delta for the source table, as shown on the right side of
Figure~\ref{fig:select_lens_put_delta_example}.

\begin{figure}[tb]
	\begin{equation*}
		\begin{array}{c|cccc}
			& \hd{track} & \hd{date} & \hd{rating} & \hd{album} \\
			\hline
			-& \tx{Lullaby} & 1989 & 3 & \tx{Galore} \\
			+& \tx{Lullaby} & 1989 & 4 & \tx{Galore} \\
			-& \tx{Lullaby} & 1989 & 3 & \tx{Show} \\
			+& \tx{Lullaby} & 1989 & 4 & \tx{Show} \\
		\end{array}
		\stackrel{\smash{\text{\tiny{$\deltaOp{\putOp}$}}}}{\Leftarrow}
		\begin{array}{c|cccc}
			& \hd{track} & \hd{date} & \hd{rating} & \hd{album} \\
			\hline
			-& \tx{Lullaby} & 1989 & 3 & \tx{Galore} \\
			+& \tx{Lullaby} & 1989 & 4 & \tx{Galore} \\
		\end{array}
	\end{equation*}
	\caption{Select lens example: using $\deltaOp{\putOp}$, we
          compute the source delta (left-hand side) from the view
          delta (right-hand side) and the original source
          (Fig.~\ref{fig:select_lens_get_example}, left).}
	\label{fig:select_lens_put_delta_example}
\end{figure}

Once we have computed the change set for the underlying tables from
Figure~\ref{fig:select_lens_put_delta_example}, we can use the delta and other
available information (such as \lstinline{tablekeys} declarations) to produce
SQL update commands that perform the desired update.  For our example, we can
perform the needed updates using two SQL \lstinline{update} operations, as
follows:

\begin{lstlisting}[language=SQL]
update tracks set date = 1989 rating = 4 where track = 'Lullaby' and album = 'Galore';
update tracks set date = 1989 rating = 4 where track = 'Lullaby' and album = 'Show';
\end{lstlisting}

\paragraph{Composition}

So far we have discussed only one relational lens primitive, namely
selection.  Updatable views can also be defined using relational
lenses for dropping attributes (projection) or combining data from
several tables (joining).  We can combine these primitive relational
lenses using the general definition of composition for
lenses~\cite{foster2007combinators}.

The join operation raises a subtle issue: joining a table with itself
can lead to copying data, which lacks clean bidirectional
semantics~\cite{foster2007combinators}.
\citet{bohannon2006relational} implicitly used a linear type discipline for
relational lenses, forbidding repeated use of the same base table in
different parts of a view.  We do not currently check this constraint
in Links, but there is no fundamental obstacle to doing so, and our
formalisation also enforces it.

\begin{figure}[p]
\begin{lstlisting}
var joinLens = lensjoin albumsLens with tracksLens on album;
var dropLens = lensdrop date determined by track default 2018 from joinLens;
var selectLens = lensselect from dropLens where quantity > 2;
\end{lstlisting}
\caption{A view \lstinline{selectLens} defined by composing join,
  drop and select operators.}
\label{fig:lens_composition_example}
	\includegraphics[width=1\linewidth]{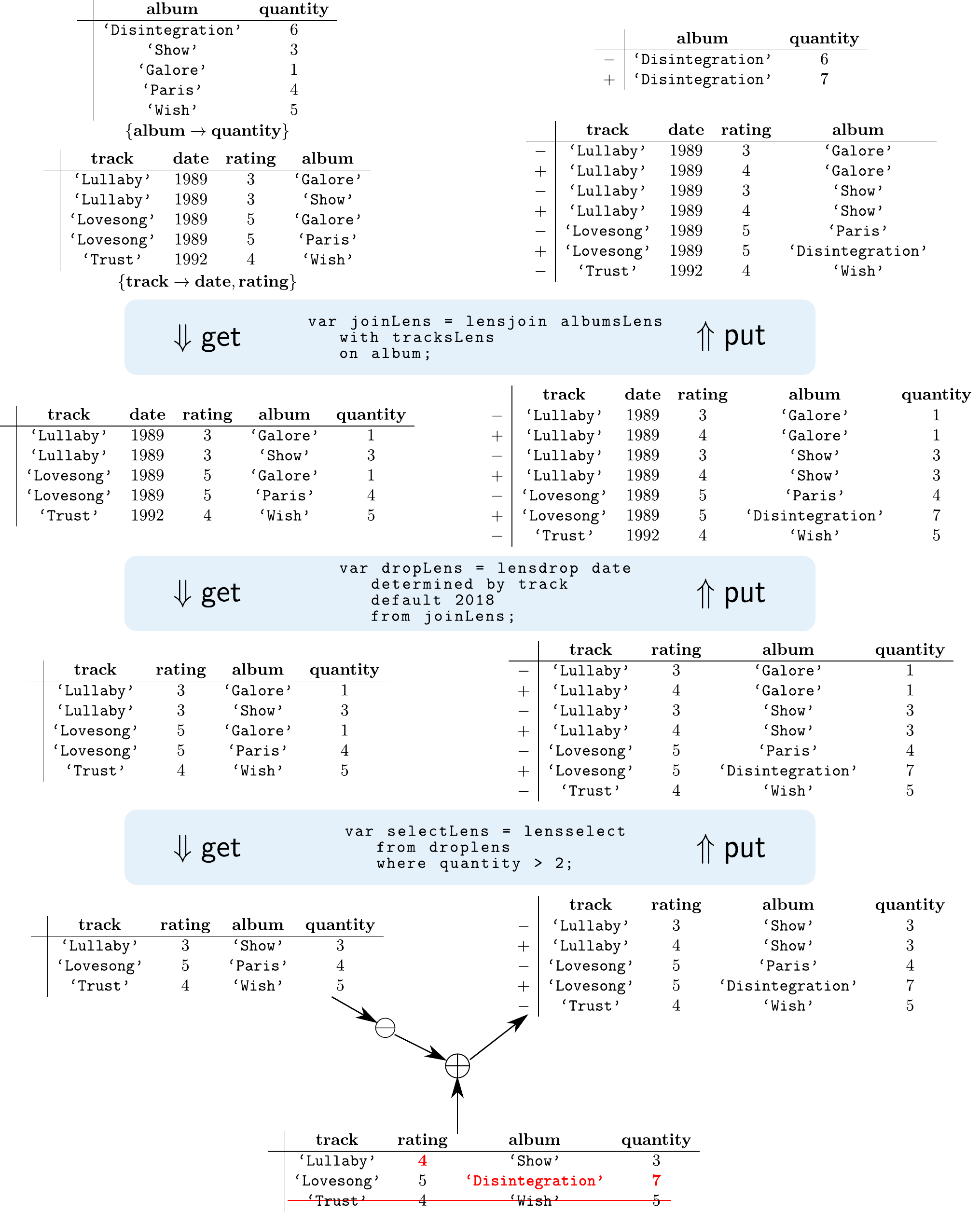}
	\caption{An example of how an update propagates through
          \lstinline{selectLens}. Changes to the view are shown at the bottom in red.}
	\label{fig:lens_composition_example_propagation}
\end{figure}

We extend our track example as shown in
Figure~\ref{fig:lens_composition_example} by first joining the two
tables.
This gives us a view \lstinline{joinLens} containing all
tracks and their corresponding albums and album quantities.  We may
then decide to discard the \emph{date} attribute using a projection
lens, yielding view \lstinline{dropLens}.  (The \lstinline{lensdrop}
combinator includes a \lstinline{default} value giving a value to use
when new data is inserted into the view.)
Finally we use selection to define a view \lstinline{selectLens} retaining rows with quantity
greater than 2.
Figure~\ref{fig:lens_composition_example_propagation} shows each of
the lenses in blue, and along the left shows how the composite lens's
$\getOp$ produces the table in the bottom left with the three tracks
\texttt{'Lullaby'}, \texttt{'Lovesong'} and \texttt{'Trust'}.  We show
the intermediate views in the $\getOp$ direction for completeness, but
it is not necessary to compute them explicitly; we can compose the
$\getOp$ directions and extract a single SQL query to produce the
final output.  The query for the example in Figure
\ref{fig:lens_composition_example} would be:

\begin{lstlisting}[language=SQL]
	select t1.track, t1.rating, t1.album, t2.quantity
		from tracks as t1
		join albums as t2 on t1.album = t2.album
		where t2.quantity > 2;
\end{lstlisting}

Suppose a user then makes the changes shown in red at the bottom of
Figure \ref{fig:lens_composition_example_propagation}.  Performing the update with
composed lenses works similarly to the case for single lenses: for a
composite lens $\ell_1;\ell_2$ we first propagate the view delta backwards
through $\ell_2$ to obtain a source delta, then treat that as a view
delta for $\ell_1$.  We calculate an initial delta by comparing the
updated view with the original view for the last lens.  This is shown at
the bottom of Figure~\ref{fig:lens_composition_example_propagation}: comparing the
original view with the updated table yields the change set shown at
the bottom right.

All intermediate change sets are calculated using the previous change
set and by querying the database.  Since the (non-incremental)
$\putOp$ function is defined in terms of the previous source and
updated view, sometimes we need to know parts of the values of the old
source or old view to calculate the incremental behaviour. Therefore,
for some relational lens steps we need to run one or more queries
against the database during change propagation.  The select lens is an
example: in order to ensure that the source update preserves the
functional dependency $track \to date~ rating$, we need to query the
database to find out what other album/track rows might need to have
their ratings updated.  The drop lens step also illustrates the need
for auxiliary querying, in this case to find out the dropped dates of
rows that are being updated. Finally, the join lens splits the changes
of the joined view into changes for the individual tables; this too
may require querying the underlying data.  This produces the deltas
shown in the top right corner of Figure
\ref{fig:lens_composition_example_propagation}.

Finally, we convert the source deltas into
SQL update commands to update the underlying tables.  Again we make use of
table key information to generate concise updates, as follows:

\begin{lstlisting}[language=SQL]
update albums set quantity = 7 where album = 'Disintegration';
update tracks set date = 1989 rating = 4 where track = 'Lullaby' and album = 'Galore';
update tracks set date = 1989 rating = 4 where track = 'Lullaby' and album = 'Show';
delete from tracks where track = 'Lovesong' and album = 'Paris';
insert into tracks (track, date, rating, album)
values ('Lovesong', 1989, 5, 'Disintegration');
delete from tracks where track = 'Trust' and album = 'Wish';
\end{lstlisting}

\section{State-based relational lenses}\label{sec:background}

In this section we recapitulate background concepts from database
theory~\cite{abiteboul1995foundations} and then review the definitions of relational
lenses~\cite{bohannon2006relational}.  We use different notation from
that paper in some cases, and explain the differences as necessary.

\subsection{Database preliminaries}

\subsubsection{Attributes and records}

\emph{Attribute names}, or simply \emph{attributes}, are ranged over by $A$,
$B$, $C$ and attribute values by $a$, $b$, $c$.  Records $m$, $n$ are partial
functions from attributes to attribute values.  For simplicity, we assume a
single (unwritten) type for attribute values; of course, in our implementation
we support the usual integers, strings, booleans, etc.  Records are written
$\record{A = a, B = b, \ldots}$.  Identifiers $U$, $V$ range over sets of
attributes considered as record domains; we use $X,Y,Z$ for arbitrary sets. We
write $m : U$ to indicate that $\dom(m) = U$. Basic operations on records
include:
\begin{itemize}
	\item[-] record projection $\recproj{m}{V}$, which means record $m$ domain-restricted to $\dom(m)
\cap V$;
 	\item[-] domain antirestriction $\domsubtract{m}{V}$, which means $m$
domain-restricted to $\dom(m) \setminus V$;
   \item[-] record update $\recupdate{m}{n} : U \cup V$ where $m : U, n : V$, which defines
$(\recupdate{m}{n})(A)$ as $n(A)$ if $A \in V$ and $m(A) $ otherwise.
\end{itemize}

\noindent Given attribute $A \in U$ and $B \notin U$, we write $U[A/B]$ for $(U
\setminus \set{A}) \cup \set{B}$, and similarly if $m : U$ then $m[A/B] :
U[A/B]$ is the tuple resulting from renaming attribute $A$ in $M$ to $B$.
Renaming is definable as
$\recupdate{(\domsubtract{m}{\set{A}})}{\record{B=m(A)}}$.

\subsubsection{Relations}

\begin{figure}
	\small
\begin{syntaxfig}
   \mbox{Queries}
   &
   q,q'
   &
   ::=
   &
   M \mid R \mid 
   \letexpr{R}{q}{q'}
& \text{Relations, names and let binding}
\\
&&
\mid&
   q \setminus q'
\mid
   q \cup q'
\mid
   q \cap q'
& \text{Set operations}
\\
 &&\mid&
  \select{P}{q}
\mid
   \proj{q}{U}
\mid
   q \Join q'
\mid
   \rename{A/B}{q}
& \text{Relational algebra}\\
\mbox{Predicates} 
& 
P,Q
& 
::= 
& 
\top \mid \neg P \mid P \wedge Q \mid P \vee Q & \text{Logical connectives}\\
&&\mid&A = B \mid A = a \mid X\in{q} & \text{Tuple predicates}\\
&&\mid&\pi_U(P) \mid P \Join Q \mid
   \rename{A/B}{q} 
& \text{Relational algebra}
\end{syntaxfig}
\caption{Syntax of relational expressions and predicates}
\label{fig:syntax:expr}
\end{figure}

\emph{Relations} $M$, $N$, $O$ are (finite) sets of records with the same
domain. $M$ has domain $U$, or equivalently $M$ is a relation of type
$U$, written $M: U$, if $m: U$ for all $m \in M$. Relations are closed
under the standard operations of relational algebra;
Figure~\ref{fig:syntax:expr} defines the syntax of relational
expressions $q$. This includes relation constants $M$, relation names
$R$, and a $\tt{let}$ construct we include for convenience. The
operations $\setminus$, $\cup$ and $\cap$ have their usual
set-theoretic interpretation, subject to the constraint that the
arguments $q, q'$ have the same type $U$, i.e. $q, q': U$.  We explain
the remaining relational algebra operations in this section.

\emph{Relational projection} is record projection extended to relations:
\[\proj{M}{U} \defeq \set{\recproj{m}{U} \mid m \in M}\]

\noindent Given $M: U$ and $N: V$, their \emph{natural join} is defined by
\[M \Join N = \set{m: U \cup V \mid \recproj{m}{U} \in M\text{ and }\recproj{m}{V} \in N}\]

$P$ ranges over predicates, which we can interpret as (possibly
infinite) sets of records, or equivalently as functions from records
to Booleans.  Predicates are required for specifying selection filters
in relational selection, as well as for the specification of filter
conditions in lens definitions.  We write $P:U$ to indicate a
predicate over records with domain $U$.  The predicate $A = B$ holds
for records $m$ satisfying $m(A) = m(B)$, while $A =a$ holds when
$m(A) = a$, and $X \in q$ holds when $m[X]$ is in the result of query
$q$.  The predicates $\top$ (truth), $\neg P$ (negation), and
$P \wedge Q$ (conjunction) are interpreted as usual.  For convenience
we include predicates $\proj{P}{U}$, $P \Join Q$, and
$\rename{A/B}{P}$ which behave analogously to the relational
operations, if we view predicates as sets of records.  For example,
$\proj{P}{U}$ holds for records $u$ such that $t[U] = u$ for some $t$
satisfying $P$.

Given a predicate $P$ and relation $M$, the selection $\select{P}{M}$ is defined
as follows:
\[\select{P}{M} \defeq \{m \in M \mid P(m) \}\]
We will be interested in cases where predicates are insensitive to the values of
certain attributes; we write ``$P$ ignores $U$'' when $P(m)$ can be determined
without considering any of the values that $m$ assigns to attributes in $U$ ---
i.e.~when for all $m$ and $n$, if $\domsubtract{m}{U} = \domsubtract{n}{U}$ then
$P(m) \iff P(n)$.

We define the relational renaming operation
$\rename{A/B}{M}$ as
\[\rename{A/B}{M} \defeq \{m[A/B] \mid m \in M\}\]
which makes it possible to join tables with differing column names. As mentioned
above, we also write $\rename{A/B}{P}$ for the result of renaming attribute $A$
in predicate $P$ to $B$.

\subsubsection{Functional dependencies}

A \emph{functional dependency} is a pair of sets of attributes, written $X \to
Y$. We say $X \to Y$ is a functional dependency \emph{over $U$}, written $X \to
Y : U$, iff $X \cup Y \subseteq U$. If $X \to Y$ is a functional dependency over
$U$ and $M: U$, then $M$ \emph{satisfies} $X \to Y$, written $M \vDash X \to Y$,
iff $\recproj{m}{X} = \recproj{n}{X}$ implies $\recproj{m}{Y} = \recproj{n}{Y}$
for all $m, n \in M$. We write $m, M \vDash X \to Y$ as a shorthand for $\set{m}
\cup M \vDash X \to Y$. It is conventional in database theory to write sets of
attributes such as $\{A,B,C\}$ as $A\;B\;C$, and $A \to B\;C$ to mean the
functional dependency $\set{A} \to \set{B, C}$.

Typically we work with sets $F, G$ of functional dependencies over a fixed $U$
and write $F: U$ iff $X \to Y: U$ for every $X \to Y \in F$. The notation $M
\vDash F$ means that $M \vDash X \to Y$ for all $X \to Y \in F$.  Likewise, $F
\vDash G$ means that $M \vDash F$ implies $M \vDash G$ for any $M$, and $F
\equiv G$ means that $F \vDash G$ and $G \vDash F$. We write $F[A/B]$ for the
result of renaming attribute $A$ to $B$ in all functional dependencies in $F$,
i.e. $F[A/B] \defeq \{X[A/B] \to Y[A/B] \mid X \to Y \in F\}$.

\subsubsection{Relation types and database schemas}
\citet{bohannon2006relational} employ \emph{relation types} $\vtype{U}{P}{F}$,
where $U$ is an attribute set, $P$ is a predicate over $U$, and $F$ is a set of
functional dependencies over $U$.  A value $M$ of type $\vtype{U}{P}{F}$ is a
relation $M : U$ such that $P(m) = \top$ for each $m\in M$ and $M \models F$.
(\citet{bohannon2006relational} wrote $(U,P,F)$ instead of $\vtype{U}{P}{F}$; we
prefer the more descriptive notation.)

\citet{bohannon2006relational} also defined lenses between
\emph{relational schemas} $\Sigma$ mapping relation names to relation
types $\vtype{U}{P}{F}$, and relational lens combinators referred to
the relations in the source and target schemas by their names. We
adopt an alternative, but equivalent approach. Schemas have the following
syntax:
\[\Sigma ::= \vtype{U}{P}{F} \mid \Sigma \otimes \Sigma'\]
Thus a database schema is a product of one or more relation types, written
$\vtype{U_1}{P_1}{F_1}\otimes \cdots \otimes \vtype{U_n}{P_n}{F_n}$. Values of
schema type $\Sigma \otimes \Sigma'$ are pairs of values of type $\Sigma$ and
$\Sigma'$.  We use tensor product notation $\otimes$ instead of $\times$ to
indicate that schema types obey a linear typing discipline. When we work with
relational lenses, the initial schema will consist of the tensor product of all
tables that contribute to the view we wish to define; we assume that a table
appears at most once in a value of such a schema.

\subsection{Lenses} \label{sec:background:lenses}

A \emph{lens}~\cite{foster2007combinators} $\ell: S \lensto V$ is a bidirectional
transformation between two sets $S$ and $V$, where $S$ is a set of
source values and $V$ is a set of possible views,
determined by two functions, $\getOp_\ell$ and $\putOp_\ell$, with the following
signatures:
\begin{align*}
	\getOp_\ell &: S \to V
&
	\putOp_\ell &: S \times V \to S
\end{align*}

\noindent A lens is \emph{well-behaved} if it satisfies two round-tripping
properties relating $\getOp_\ell$ and $\putOp_\ell$. The property PutGet ensures that
whatever data we put into a lens is returned unchanged if we get it again. The
property GetPut ensures that if we put the view value back into the lens
unchanged, the underlying source value is also unchanged.
\begin{align*}
	\tag{PutGet}
	& \getOp_\ell(\putOp_\ell(s, v)) = v \\
	\tag{GetPut}
	& \putOp_\ell(s, \getOp_\ell(s)) = s
\end{align*}

\noindent From now on the notation $\ell: S \lensto V$ means that $\ell$ is a
well-behaved lens from $S$ to $V$.

Lenses form a category with identity and composition constructions.
The identity lens $\id_X: X \lensto X$ is given by the functions
$\getOp$ and $\putOp$ defined as:
\begin{align*}
	\getOp_{\id}(x) &= x
&
	\putOp_{\id}(x, x') &= x'
\end{align*}

\noindent
We omit the subscript on $\id$ when clear from context. The identity
lens is trivially well-behaved.  Diagram-order composition
$\ell_1; \ell_2: X \lensto Z$ of the lenses $\ell_1: X \lensto Y$ and
$\ell_2: Y \lensto Z$ is given by the functions $\getOp$ and $\putOp$ defined
as:
\begin{align*}
	\getOp_{\ell_1; \ell_2}(x) &= \getOp_{\ell_2}(\getOp_{\ell_1}(x))
	&
	\putOp_{\ell_1; \ell_2}(x, z) &= \putOp_{\ell_1}(x, \putOp_{\ell_2}(\getOp_{\ell_1}(x), z))
\end{align*}

As discussed in previous work~\cite{hofmann2011symmetric,hofmann2012edit}, the
category of lenses also has symmetric monoidal products; that is, there is a
construction $\otimes$ on its objects such that $X \otimes Y$ is the set of
pairs $\{(x,y) \mid x \in X, y \in Y\}$, and which satisfies symmetry and
associativity laws:
\begin{align*}
\tag{Sym}
X \otimes Y & \equiv Y \otimes X \\
\tag{Assoc}
X \otimes (Y \otimes Z) &\equiv (X \otimes
Y) \otimes Z
\end{align*}
These laws
are witnessed by (invertible) lenses $\sym_{X,Y}$ and $\assoc_{X,Y,Z}$, defined as follows:
\begin{align*}
	\getOp_{\sym}(x,y) &= (y,x) & \getOp_{\assoc}(x,(y,z)) &= ((x,y),z)
	\\
	\putOp_{\sym}(\_,(y,x)) &= (x,y)
                                    &
\putOp_{\assoc}(\_,((x,y),z)) &= (x,(y,z))
\end{align*}
In addition, we have the following combinator for combining two lenses
`side-by-side':
\begin{align*}
	\getOp_{\ell_1\otimes \ell_2}(x_1,x_2) &= (\getOp_{\ell_1}(x_1),\getOp_{\ell_2}(x_2))
	\\
	\putOp_{\ell_1\otimes \ell_2}((x_1,x_2), (y_1,y_2)) &= (\putOp_{\ell_1}(x_1,y_1),\putOp_{\ell_2}(x_2,y_2))
\end{align*}
\begin{cameraready}
  so that if $\ell_1: X_1 \lensto Y_1$ and $\ell_2: X_2 \lensto Y_2 $
  then
  $ \ell_1 \otimes \ell_2 : X_1 \otimes X_2 \lensto Y_1 \otimes Y_2$.
\end{cameraready}
\begin{techreport}These lens constructs preserve well-behavedness as characterised by the following
inference rules:
\begin{mathpar}
	\inferrule*[right={T-id}]
	{
		\strut
	}
	{
		\id_X: X \lensto X
	}
	\and
	\inferrule*[right={T-Compose}]
	{
		\ell_1: X \lensto Y
		\\
		\ell_2: Y \lensto Z
	}
	{
		\ell_1;\ell_2: X \lensto Z
	}
\and
	\inferrule*[right={T-sym}]
	{
	}
	{
		\sym_{X,Y}: X \otimes Y \lensto Y \otimes X
	}
	\and
	\inferrule*[right={T-assoc}]
	{
   }
	{
		\assoc_{X,Y,Z}: X \otimes (Y \otimes Z) \lensto (X \otimes Y) \otimes Z
	}
\and
	\inferrule*[right={T-Product}]
	{
		\ell_1: X_1 \lensto Y_1
		\\
		\ell_2: X_2 \lensto Y_2
	}
	{
          \ell_1 \otimes \ell_2 : X_1 \otimes X_2 \lensto Y_1 \otimes Y_2
        }
\end{mathpar}
\end{techreport}
\subsection{Relational revision}

A key relational lens concept introduced by \citet{bohannon2006relational} is
\emph{relational revision}.  Given a set of functional dependencies $F:U$ and
relations $M,N:U$ such that $N \models F$, relational revision modifies $M$ to
$M'$ so that $M' \cup N \models F$.  For example, given $F = \{A \to B\}$ and $M =
\set{\record{A=1,B=2},\record{A=2,B=3}}$ and $N = \set{\record{A=1,B=42}}$, the
result of revising $M$ to be consistent with $N$ and $F$ is
$\set{\record{A=1,B=42},\record{A=2,B=3}}$.

\subsubsection{Functional dependencies in tree form}

In general revision may not be uniquely defined, for example if there are
cycles among functional dependencies.
 \citet{bohannon2006relational} avoid this problem by requiring that
sets of functional dependencies be in a special form called \emph{tree
  form}.  We briefly restate the definition for concreteness.
\begin{definition}
  Given functional dependencies $F$, define
\[V_F = \{X \mid X\to Y \in  F\} \cup \{Y \mid X \to Y \in F\} \qquad
E_F = \{(X,Y) \mid X \to Y \in F\}\]
Then we say $F$ is in \emph{tree form} if the graph $T_F = (V_F,E_F)$ is a forest and
$V_F$ partitions $\bigcup V_F$.
\end{definition}
If $F$ is in tree form, then each attribute set of $F$ corresponds to
a node in a tree (or forest) where the edges correspond to elements of
$F$.  Moreover, no distinct nodes of $T_F$ have common atttibutes. For
example, $\{A \to B~C, B \to D\}$ is not in tree form, but is
equivalent to $\{A \to B, A \to C, B \to D\}$ which is in tree form.
However, $\{A \to B~C, C \to A~D\}$ has no equivalent tree form
representation.


Figure~\ref{fig:FD-operations} defines various functions on sets of functional
dependencies. The sets $\Left(F)$ and $\Right(F)$ consist of all attributes
appearing on the left or right side of any functional dependency $X \to Y \in
F$. The set $\outputs(F)$ consists of all attributes that are actually
constrained in $F$ by other attributes. Finally, $\roots(F)$ is the set of all
nodes of $T_F$ that have indegree zero.
From now on, we assume sets of functional dependencies $F$ appearing
in relation types $\vtype{U}{P}{F}$ are in tree form, as a
prerequisite to this type being well-formed.

\subsubsection{Revision and merge operations}

Relational revision is expressed in terms of a \emph{record revision} operation
$\newrecrevise{m}{F}{N}$ which takes a set of functional dependencies $F:U$ in
tree form, a record $m:U$, and a set of records $N:U$ such that $N \models F$,
and is defined by recursion over the tree structure of $F$.  If $F$ is empty,
record revision simply returns $m$. Otherwise, there must be at least one
functional dependency $X \to Y$ in $F$ such that $X$ is a root. If $m$ and some
$n \in N$ have the same values for $X$, we return
$\recupdate{m}{\recproj{n}{Y}}$, that is, a copy of $m$ whose $Y$ attributes
have been updated with those from $\recproj{n}{Y}$; otherwise we return $m$
unchanged. We then recursively process the remaining functional dependencies.

Traversing $F$ starting from the roots is nondeterministic, but provided $F$ is
in tree form, the end result of record revision is uniquely defined, because
each attribute in $\Right(F)$ is modified at most once and no attribute can be
modified until all other attributes it depends on have been modified.

\begin{figure}[tb]
\begin{salign}
	\Left(F) &= {\textstyle\bigcup\,\set{X \mid X \to Y \in F}}
   \\
	\Right(F) &= {\textstyle\bigcup\,\set{Y \mid X \to Y \in F}}
   \\
   \outputs(F) &= {\textstyle\bigcup\,\set{Y \mid \exists X.F \vDash X \to Y\text{ and }X \cap Y = \varnothing}}
   \\
   \roots(F) &= {\textstyle\set{X \mid \exists Y. X \to Y \in F \text{ and } X \cap \Right(F) = \varnothing}}
   \\[2mm]
   F = {\set{X \to Y}} \cdot F'
   &\iff
   F = \set{X \to Y} \uplus F'\text{ and }X \in \roots(F)
   \\[2mm]
   \Genrecrevise_F
   &: U \to \vtype{U}{F}{P} \to U
   \\
   \newrecrevise{m}{\varnothing}{N}
   &=
   m
   \\
  \newrecrevise{m}{ \set{X \to Y} \cdot F }{N}
   &=
   \begin{cases}
      \newrecrevise{\recupdate{m}{\recproj{n}{Y}}}{F}{N}
      & \text{if }\exists n \in N.\recproj{m}{X} = \recproj{n}{X}
      \\
      \newrecrevise{m}{F}{N}
      & \text{otherwise}
   \end{cases}
   \\[2mm]
   \relrevise{M}{F}{N}
   &=
   \set{\newrecrevise{m}{F}{N} \mid m \in M}
   \\
   \relmerge{M}{F}{N}
   &=
   \relrevise{M}{F}{N} \cup N
\end{salign}
\caption{Operations on functional dependencies and relational
  revision}
\label{fig:FD-operations}
\end{figure}

\begin{definition}[Relational revision]
\label{def:rel_revise}
Figure~\ref{fig:FD-operations} defines the \emph{relational revision} operation
$\relrevise{M}{F}{N}$ that takes two sets of records $M: U$ and $N: U$ where $N
\vDash F$, and applies record revision to every record $m \in M$ using
the given functional dependencies $F$.
\end{definition}

\begin{definition}[Relational merge]
\label{def:rel_merge}
Figure~\ref{fig:FD-operations} also defines the \emph{relational
  merge} operation $\relmerge{M}{F}{N}$, where $N \vDash F$, which revises $M$ according
to $F$ and $N$ and then unions the result with $N$.
\end{definition}

\subsection{Relational lens primitives}
\label{sec:background:relational-lenses}

In this section we recapitulate the primitive relational lenses introduced by
\citet{bohannon2006relational}: a \emph{selection} lens that corresponds to
selection, a \emph{drop} lens that corresponds to projection, and a \emph{join}
lens that corresponds to relational join. (``Corresponds'' means that the
$\getOp$ direction coincides with the relational operation.) We also introduce a
trivial \emph{rename} lens corresponding to the relational renaming operator.
The syntax of the relational lenses, including the generic operations from
\Secref{background:lenses}, is given in Figure~\ref{fig:syntax:relational-lens}.

Our presentation differs from that of \citet{bohannon2006relational} in that we
use generic lens combinators arising from the symmetric monoidal product
structure to deal with linearity.  This makes it possible for each primitive to
mention only the affected source and target data and not the rest of the
database instance.

\begin{figure}[tb]
\begin{syntaxfig}
   \mbox{Relational lenses}
   &
   \ell,\ell'
   &
   ::=
   &
   \SelectPrim{P}
\mid
   \DropPrim{A}{X}{a}
\mid
   \JoinDLPrim
\mid
   \RenamePrim{A/B}
   \\
   &&\mid&
   \IdPrim{}
\mid   \ell_1;\ell_2
\mid
   \SymPrim{}
\mid
   \AssocPrim{}
\mid
   \ell_1 \otimes \ell_2
 \end{syntaxfig}
 \caption{Syntax of relational lens expressions}
 \label{fig:syntax:relational-lens}
\end{figure}

Relational lenses are lenses between schema types $\Sigma$, that is, tensor
products of relation types $\vtype{U}{P}{F}$.  Relational lens expressions are
subject to a typing judgement given in Figure~\ref{fig:typing-relational-lens};
well-typed lenses are guaranteed to be well-behaved.  The preconditions in these
rules are those given by \citet{bohannon2006relational}, to which we refer the
reader for further explanation.

\begin{figure}
\begin{mathpar}
   \inferrule*[right={T-Select}]
	{
		\label{eq:t_select}
		Q\text{ ignores }\outputs(F)
	}
	{
		\SelectPrim{P}: \vtype{U}{Q}{F} \lensto \vtype{U}{P \wedge Q}{F}
	}
   \and
   \inferrule*[right={T-Drop}]
	{
		\label{eq:t_drop}
		P = \proj{P}{U \setminus A} \Join \proj{P}{A}
		\\
		\record{A=a} \in \proj{P}{A}
	}
	{
		\DropPrim{A}{X}{a}:
      \vtype{U}{P}{F \uplus \set{X \to A}} \lensto \vtype{U \setminus A}{\proj{P}{U \setminus A}}{F}
	}
   \and
   \inferrule*[right={T-JoinDL}]
	{
		\label{eq:t_joinDL}
		G \vDash U \cap V \to V
		\\
		P\text{ ignores }\outputs(F)
		\\
		Q\text{ ignores }\outputs(G)
	}
	{
		\JoinDLPrim:
      \vtype{U}{P}{F} \otimes \vtype{V}{Q}{G} \lensto \vtype{U \cup V}{P \Join Q}{F \cup G}
	}
\and
   \inferrule*[right={T-Rename}]
	{
		\label{eq:t_rename}
		A \in U \\ B \notin U
	}
	{
		\RenamePrim{A/B}: \vtype{U}{P}{F} \lensto \vtype{U[A/B]}{\rename{A/B}{P}}{F[A/B]}
	}
\end{mathpar}
\caption{Typing rules for relational lens primitives}
\label{fig:typing-relational-lens}
\end{figure}

\subsubsection{Select lens}

The lens $\SelectPrim{P}: \vtype{U}{Q}{F} \lensto \vtype{U}{P \wedge Q}{F}$ is
defined as follows by the functions $\getOp$ and $\putOp$:
\[\begin{array}{lrcl}
&	\getOp(M) &=& \select{P}{M}
	\\
&	\putOp(M, N) &=& \letexpr{M_0}{\relmerge{\select{\neg
                         P}{M}}{F}{N}}{}\\
&&& \letexpr{N_\#}{\select{P}{M_0} \setminus N}{}
\\
&&&M_0 \setminus N_\#
\end{array}\]

\noindent The $\putOp$ operation first calculates $M_0$, the set of records
$\select{\neg P}{M}$ excluded from the original view, revised to be consistent
with the functional dependencies witnessed by the updated view $N$
together with $N$ itself. The set $N_\#$ collects records matching $P$
but not in $N$, which are removed from $M_0$ in order to satisfy PutGet.

\subsubsection{Project lens}
\label{sec:background:relational-lenses:project}

The lens $\DropPrim{A}{X}{a}:  \vtype{U}{ P}{ F \uplus \set{X \to A}}\lensto
\vtype{U \setminus A}{ \proj{P}{U \setminus A}}{F }$ is defined as follows
\begin{techreport}
	\footnote{This is slightly simpler than, but equivalent to, the definition given
	by \citet{bohannon2006relational}. The proof that the two definitions are
	equivalent is provided in Appendix~\ref{appendix:background}.}
\end{techreport}
\begin{cameraready}
	\footnote{This is slightly simpler than, but equivalent to, the definition given
	by \citet{bohannon2006relational}. The proof that the two definitions are
	equivalent is provided in the full version of the paper~\cite{irl-full}.}
\end{cameraready}
by the functions
$\getOp$ and $\putOp$:
\[\begin{array}{lrcl}
&	\getOp(M) &=& \proj{M}{U \setminus A}
	\\
&	\putOp(M, N) &= &
\letexpr{M'}{N \Join \set{\{ A = a \}}}{}
\\
&&& \relrevise{M'}{X \to A}{M}
\end{array}\]

\noindent 
For $\putOp$, each row in $N$ is initially given the default value $a$
for $A$. $M$ is then used to override the default value in $M'$ using relational
revision, so that if there is an entry $m \in M'$ with the same value for the
determining column $X$ the corresponding $A$ value from $M$ is used instead.

\subsubsection{Join lens} \citet{bohannon2006relational} described several
variants of lenses for join operations.  All three perform the natural join of
their two input relations in the $\getOp$ direction, but differ in how deletions
are handled in the $\putOp$ direction.  A view tuple deletion could be
translated to a deletion in the left, right, or both source relations, and so
there are three combinators $\JoinDLPrim$, $\JoinDRPrim$, and $\JoinBothPrim$
expressing these three alternatives.  In their extended report,
\citet{bohannon2006tr} showed how to define all three combinators as a special
case of a generic template.  To keep the presentation simple, we present just
the $\JoinDLPrim$ combinator here.

The `join/delete left' lens $\JoinDLPrim: \vtype{U}{P}{F} \otimes \vtype{V}{Q}{G}
\lensto \vtype{U \cup V}{P \Join Q}{ F \cup G}$ is given by the functions $\getOp$
and $\putOp$ defined as follows:
\[\begin{array}{lrcl}
&	\getOp(M,N) &=& M \Join N
	\\
&	\putOp((M,N), O) &=& \letexpr{M_0}{\relmerge{M}{F}{\proj{O}{U}}}{}
\\
&&&
\letexpr{N'}{\relmerge{N}{G}{\proj{O}{V}}}{}
\\
&&&
\letexpr{L}{ (M_0 \Join N') \setminus O}{}
\\
&&&
\letexpr{M'}{M_0 \setminus \proj{L}{U}}{}
\\
&&&(M',N')
\end{array}\]

\noindent The intuition for the $\putOp$ direction is as follows.  We first
compute $M_0$ by merging the projection $\proj{O}{U}$ into source table $M$, and
likewise $N'$, merging the projection $\proj{O}{V}$ into the source table $N$.
We next identify those tuples $L$ which are in the join of $M_0$ and $N'$ but
which are not present in the updated view $O$. To satisfy PutGet, we must make
sure these tuples do not appear in the join after updating the source relations.
It is sufficient to delete each of those records from one of the two source
tables; since this lens deletes uniformly from the left table, we compute $M'$
by subtracting the projection $\proj{L}{U}$ from $M_0$.  Finally, $M'$ and $N'$
are the new values for the source tables.

\subsubsection{Renaming lens}

Renaming is a trivial but important operation in relational algebra, since
otherwise there is no way to join the $A$ field of one table with the $B \neq A$
field of another. We introduce a renaming lens $\RenamePrim{A/B} :
\vtype{U}{P}{F} \lensto \vtype{U[A/B]}{\rename{A/B}{P}}{F[A/B]}$, provided $A
\in U $ and $B \notin U$, with its $\getOp$ and $\putOp$ operations defined as
follows: \[\begin{array}{lrcl} &	\getOp(M) &=& \rename{A/B}{M} \\ &	\putOp(\_,N) &=&
\rename{B/A}{N} \end{array}\] \citet{bohannon2006relational} did not define such
a lens, but its well-behavedness is obvious.

\section{Incremental framework}\label{sec:framework}

To describe the incremental behaviour of relational lenses, we need to
represent changes to query results in a simple, compositional way. We
adopt an approach similar to \citet{griffin1997incremental}, who model
``delta relations'' as disjoint pairs of relations specifying tuples
to be added and removed from a relation of the same type. We use notation similar to
\citet{cai2014pldi} and our relations and delta relations form a \emph{change
  structure} in their sense, though we do not spell out the details formally.

\subsection{Delta relations}

\begin{definition}[Delta relation]
A \emph{delta relation over $U$} is a pair $\Delta M = (\pDelta{M},
\nDelta{M})$ of disjoint relations $\pDelta{M}: U$
and $\nDelta{M}: U$. The empty delta relation $(\emptyset, \emptyset)$
is written $\emptyset$.  We write $\Delta M: \Delta U$ to indicate that $\Delta
M$ is a delta relation over $U$.
\end{definition}

A delta specifies a modification to a relation: for example, if $M = \set{2, 3,
4}$ and $\Delta M = (\set{3, 5}, \set{4, 9})$ then $\pDelta{M}$ specifies that
$\set{3, 5}$ are to be added to $M$, resulting in the set $\set{1, 3, 5}$. Note
that the redundant insertion of $3$ specified by $\pDelta{M}$ and the redundant
deletion of $9$ specified by $\nDelta{M}$ are both permitted. However,
\citet{griffin1997incremental} define a delta $\Delta M: \Delta U$ to be
\emph{minimal} for $M: U$ if it contains no redundant insertions or deletions of
that sort; for example, $(\{5\},\{4\})$ is the minimal delta relative to $M$
equivalent to $\Delta M$ above.

\begin{definition}[Minimal delta]
\label{def:minimality}
   $\Delta M: \Delta U$ is \emph{minimal} for $M: U$ iff $\pDelta{M}
   \cap M = \emptyset$ and $\nDelta{M} \subseteq M$.
\end{definition}

\begin{definition}[Implicit coercion to delta-relation]
Any relation $M: U$ can be implicitly coerced to a delta-relation $M: \Delta U
\defeq (M, \emptyset)$ which is minimal for $\emptyset: U$.
\end{definition}

Deltas of the same type can be combined by a minimality-preserving \emph{merge}
operation $\deltaplus$.

\begin{definition}[Delta merge]
For any $\Delta M, \Delta N: \Delta U$, define
\[
   (\Delta M \deltaplus \Delta N): \Delta U
   \defeq
   ((\pDelta{M} \setminus \nDelta{N}) \cup (\pDelta{N} \setminus \nDelta{M}),
    (\nDelta{M} \setminus \pDelta{N}) \cup (\nDelta{N} \setminus \pDelta{M}))
\]
\end{definition}

Implicit coercion of $M$ to the delta-relation $(M, \emptyset)$, combined with
delta merge $\deltaplus$, gives rise to a notion of \emph{delta application} $M
\deltaapp \Delta M$. If $\Delta M$ is minimal then the resulting delta has an
empty negative component and can be coerced back to a relation.

\begin{restatable}{lemma}{restatelemexpandminimal}
	If $\Delta M$ is minimal for $M$ then $M \deltaplus \Delta M = (M \setminus
	\nDelta{M}) \cup \pDelta{M} = (M \cup \pDelta{M}) \setminus \nDelta{M}$.
	\label{lem:expand_minimal}
\end{restatable}

\begin{definition}[Delta negate]
   For any $\Delta M: \Delta U$, define $\deltaminus\Delta M: \Delta U \defeq (\nDelta{M}, \pDelta{M})$.
\end{definition}

\begin{definition}[Delta difference]
   For any $\Delta M, \Delta N: \Delta U$, define $(M \deltaminus N): \Delta U \defeq M \deltaplus (\deltaminus N)$.
\end{definition}

\noindent The implicit coercion to delta-relations gives rise to a notion of
relational difference $(M \deltaminus N): \Delta U$, not to be confused with $(M
\setminus N): U$, which is the set difference and only removes elements $N$ from $M$. $M \deltaminus
N$ can be used, for example, to calculate the difference between two
views, such that $N \deltaplus (M \deltaminus N) = M$.

\begin{restatable}{lemma}{restatelemdiff}
\label{lem:diff}
   Suppose $M: U$ and $N: U$. Then $(M \deltaminus N): \Delta U = (M \setminus
   N, N \setminus M)$. Moreover $M \deltaminus N$ is minimal for $N$.
\end{restatable}

The following are some useful straightforward properties of deltas:

\begin{restatable}{lemma}{restatelemdeltabasicspositive}
	\label{lem:delta-basics-positive}
	Suppose $\Delta M$ minimal for $M$. Then $(M \deltaapp \Delta M) \setminus M = \pDelta M$.
\end{restatable}

\begin{restatable}{lemma}{restatelemdeltabasicsnegative}
	\label{lem:delta-basics-negative}
	Suppose $\Delta M$ minimal for $M$. Then $(M \cap (M \deltaapp
        \Delta M)) \deltaminus M = \deltaminus \nDelta M$, hence $M \setminus (M \deltaapp \Delta M) =  \nDelta M$.
\end{restatable}

\begin{corollary}
  \label{cor:minimal-diff}
  If $\minimal{\Delta M}{M}$ then $(M \deltaplus
\Delta M) \deltaminus M = \Delta M$.
\end{corollary}
\begin{proof}
By the previous two lemmas, $(M \deltaplus
\Delta M) \deltaminus M = ((M \deltaplus \Delta M)\setminus
M, M \setminus (M \deltaplus \Delta M)) = (\pDelta M,\nDelta M) =
\Delta M$.
\end{proof}

\begin{corollary}
  \label{cor:minimal-unique}
If $\Delta M$ and $\Delta M'$ are minimal for $M$ and $M \deltaplus \Delta M =
M \deltaplus \Delta M'$ then $\Delta M = \Delta M'$.
\end{corollary}
\begin{proof}
By Corollary~\ref{cor:minimal-diff}, $\Delta M = (M \deltaplus\Delta M)
\deltaminus M = (M \deltaplus\Delta M')
\deltaminus M = \Delta M'$.
\end{proof}
The property $(M \deltaapp \Delta M) \deltaminus M = \Delta M$ is
mentioned in \citet{cai2014pldi} but not required by their definition of
change structures.  It is very helpful in our setting because it
implies that query expressions incrementalise in a unique,
compositional way, as we show next.

\subsection{Delta-relational operations}\label{sec:delta-relational}

We now consider how to incrementalise relational operations.  For each
relational operator, such as $\select{P}{M}$ or $\relmerge{M}{F}{N}$, we would
like to define an operation that translates deltas to the arguments to a delta
to the result. Incremental operations with symbolic names are written with a
dot, for example $\deltaselect{P}{M}$, while alphabetic names have their
incremental counterpart written with a preceding $\delta$, for example
$\deltarelmerge{M}{F}{N}$.

The notion of \emph{delta-correctness} characterises when a function
$\deltaOp{op}$ with a suitable signature which operates on deltas can be
considered to be a valid ``incrementalisation'' of a non-incremental operation
$\opname{op}$.  As observed by \citet{griffin1997incremental}, composing
incremental relational operations is easier if they also preserve minimality, so
we build this property into our definition.

\begin{definition}[Delta-correctness]
  For any operation $\opname{op}: X \to Y$, a delta operation
  $\deltaOp{op}: X \times \Delta X \to \Delta Y$ is
  \emph{delta-correct} for $\opname{op}$ if for any  $\Delta x$
  minimal for $x$, we have:
  \begin{enumerate}
  \item
    $\deltaOp{op}(x,\Delta x)$ is minimal for $\opname{op}(x)$.
  \item $\opname{op}(x \deltaapp \Delta x) = \opname{op}(x) \deltaapp \deltaOp{op}(x,\Delta x)$
  \end{enumerate}
\end{definition}

We generalise the above definition to binary operations in the obvious
way.  Delta-correct operations are uniquely determined by the
minimality condition:

\begin{lemma}\label{lem:delta-correct-unique}
  If $\deltaOp{op}$ is delta-correct then $\deltaOp{op}(x,\Delta x) =
  \opname{op}(x \deltaplus \Delta x) \deltaminus \opname{op}(x)$ provided $\Delta x$
  is minimal for $x$.  In particular,
  $\deltaOp{op}(x,\emptyset) = \emptyset$.
\end{lemma}
\begin{proof}
By lemma~\ref{lem:diff}, $\opname{op}(x \deltaplus \Delta x) \deltaminus \opname{op}(x)$ is
minimal for $\opname{op}(x)$, and by the definition of
delta-correctness and lemma~\ref{lem:diff} we have
\[
\opname{op}(x) \deltaapp \deltaOp{op}(x,\Delta x) =
\opname{op}(x \deltaplus \Delta x) =
\opname{op}(x) \deltaplus (\opname{op}(x \deltaplus \Delta x) \deltaminus \opname{op}(x))
\]
By Corollary~\ref{cor:minimal-unique} we can conclude $
\deltaOp{op}(x,\Delta x) = \opname{op}(x \deltaplus \Delta x) \deltaminus \opname{op}(x)$.
\end{proof}

For ease of composition, we define $\DeltaOp{op}(x,\Delta x) =
(\opname{op}(x),\deltaOp{op}(x,\Delta x))$ as the function that returns both the
updated result and the delta.
\begin{definition}
   For any $\opname{op}: X \to Y$, define      $ \DeltaOp{op} : X \times \Delta X \to Y \times \Delta Y$ as
   \begin{align*}
      \DeltaOp{op}(x, \Delta x) &= (\opname{op}(x), \deltaOp{op}(x, \Delta x))
   \end{align*}
We say $\DeltaOp{op}$ is delta-correct (with respect to $\opname{op}$) when $\deltaOp{op}$ is.
\end{definition}
If $\deltaOp{op}$ is delta-correct then whenever $\Delta x$ is minimal
for $x$, so is $\deltaOp{op}(x,\Delta x)$ for $\opname{op}(x)$.  This implies composability in the following sense:
\begin{lemma}
\label{lem:delta-correct-composition}
  If $\deltaOp{op_1} : X \times \Delta X \to \Delta Y,\deltaOp{op_2} :
  Y \times \Delta Y \to Z$ are delta-correct then $\deltaOp{op_2}
  \circ \DeltaOp{op_1}$ and $\DeltaOp{op_2} \circ \DeltaOp{op_1}$ are
  delta-correct (both with respect to $\opname{op_2} \circ
  \opname{op_1}$).
\end{lemma}
Furthermore, this implies we may incrementalise any function built up
out of incrementalisable relational operations, by replacing ordinary
operators with their incremental counterparts, largely as described by
\citet{cai2014pldi}.  Given a query $q(R_1,\ldots,R_n)$, we can transform it to a
delta-correct (but not necessarily efficient) incremental version
by taking $\Derive{q} = \letexpr{(R,\Delta R)}{\Dagger{q}}{\Delta R}$,
  where the transformation $\Dagger{\cdot}$ is defined as follows:
\[\begin{array}{rclcrcl}
\Dagger{M} &=& (M,\emptyset) &&
\Dagger{op(q_1,\ldots,q_n)} &=&
                               (\DeltaOp{op}(\Dagger{q_1},\ldots,\Dagger{q_n}))
\\
\Dagger{R} &=& (R,\Delta R) & &
\Dagger{   \letexpr{R}{q}{q'}} &=& \letexpr{(R,\Delta
                                   R)}{\Dagger{q}}{\Dagger{q'}}
\end{array}
\]
Essentially $\Dagger{q}$ traverses the query, replacing relation variables with
pairs of variables and deltas, replacing constant relations with pairs
$(M,\emptyset)$ and dealing with individual operations and
let-bindings compositionally.  We abuse notation slightly by adding
syntax for pairs.
\begin{restatable}{theorem}{restatethmcompositional}
\label{thm:compositional}
If $q : \reltype{U_1}\times \cdots \times \reltype{U_n} \to
\reltype{U}$ then $\Derive{q}$ and
$\Dagger{q}$ are delta-correct with respect to $q$.
\end{restatable}

\subsection{Optimisation rules for delta operations}
\label{sec:properties-delta-operations}

To sum up, we have established that for any query there is a an extensionally
unique incrementalisation, obtained by computing the difference between the
updated query result and the original result. Of course, this is far from an
efficient implementation strategy. In this section, we present a number of
optimisation rules for incremental relational operations, as well as relational
revision and merge.

Most of the following characterisations of incremental relational
operations are presented in prior work such as
\citet{griffin1997incremental}, but without detailed proofs;
\begin{cameraready}
  we include detailed proofs in the full version~\cite{irl-full}.
\end{cameraready}
\begin{techreport}
  we include detailed proofs in the appendix.
\end{techreport}

\begin{restatable}{lemma}{restatelemcorrectnessdeltaops}[Valid optimisations]
   \label{lem:correctness-delta-ops}
Assume $\Delta M$, $\Delta N$ are minimal for $M,N$ respectively. Then:
   \begin{enumerate}
      \item $\deltaselect{P}{M, \Delta M} =
   (\select{P}{\positive{\Delta M}}, \select{P}{\negative{\Delta M}})$
      \item $ \deltaproj{M, \Delta M}{U}
  =
   (\proj{\positive{\Delta M}}{U} \setminus \proj{M}{U},
    \proj{\negative{\Delta M}}{U} \setminus \proj{M \deltaapp \Delta M}{U}) $
      \item $(M, \Delta M) \deltaJoin (N, \Delta N)
   =
   (((M \deltaapp \Delta M) \Join \positive{\Delta N}) \cup (\positive{\Delta M} \Join (N \deltaapp \Delta N)),
    (\negative{\Delta M} \Join N) \cup (M \Join \negative{\Delta N})) $
      \item $\deltarename{A/B}{M, \Delta M}
   =
   (\rename{A/B}{\Delta M^+}, \rename{A/B}{\Delta M^-})$
      \item If $N \subseteq M$ and
     $N \deltaplus \Delta N \subseteq M \deltaplus \Delta M$ then $(M, \Delta M) \deltasetminus (N, \Delta N)
   =
   \Delta M \deltaminus \Delta N$
   \end{enumerate}
\end{restatable}

Relational revision is only used directly for drop lenses, where only the first
argument $M$ may change. The following lemma provides an optimisation for this
case:

\begin{restatable}{lemma}{restatelemdeltarelrev}
   \label{lem:delta_rel_rev}
Suppose $M \models X \to A$ and $M \deltaplus \Delta M \models X \to
A$.  Then
   $\deltarelrevise{(M, \Delta M)}{X \to A}{(N, \emptyset)}
    =
	(\relrevise{\pDelta{M}}{X \to A}{N},\relrevise{\nDelta{M}}{X
          \to A}{N})$.
\end{restatable}

For relational merge, the join lens makes use of the following special case:

\begin{restatable}{lemma}{restatelemdeltarelmerge}
   \label{lem:delta_rel_merge}
   If $\relmerge{M}{F}{N} = M$ then
   $\deltarelmerge{(M, \emptyset)}{F}{(N, \Delta N)} =
   \relmerge{M}{F}{\pDelta{N}} \deltaminus M$.
\end{restatable}

In select and join lenses, we will avoid explicitly recomputing
$\relmerge{M}{F}{N}$ by showing that it is sufficient to consider only a subset
of possibly-affected rows in $M$. We define a function called $\affected_F$
which returns a predicate selecting a (hopefully small) superset of the rows
that may be changed by relational merge according to $F$ and a set of view
records $N$. The returned predicate is the necessary condition for any changes
implied by $F$ and $N$.

\begin{definition}
	$\affected_F(N) \defeq \bigvee_{X \to Y \in F} X \in \proj{N}{X}.$
\end{definition}

\noindent It is then possible to replace the target relation $M$ with only those
rows in $M$ which are likely to be updated, allowing fewer rows to be queried
from the database:

\begin{restatable}{lemma}{restaterelmergeefficient}
	\label{lem:relmerge_efficient}
	If $P=\affected_F(\pDelta{N})$ then $
		\relmerge{M}{F}{\Delta N^+} \deltaminus M = \relmerge{\select{P}{M}}{F}{\Delta N^+} \deltaminus \select{P}{M}$.
\end{restatable}

\section{Incrementalising Relational Lenses}\label{sec:incremental}

\subsection{Incremental lenses}

Assume $S,V$ are sets of relations equipped with sets of deltas $\Delta S, \Delta V$ and
corresponding operations $\deltaplus, \deltaminus$.  A (well-behaved)
\emph{incremental lens} $\ell: S \lensto V$ is a well-behaved lens equipped with
additional operations $\deltaOp{get}_\ell: S \times \Delta S \to \Delta V$ and
$\deltaOp{put}_\ell: S \times \Delta V \to \Delta S$ satisfying
\begin{align*}
\getOp_\ell(s \deltaapp \Delta s) &= \getOp_\ell(s) \deltaapp \deltaOp{get}_\ell(s, \Delta s) \\
\tag{$\Delta$PutGet}
\putOp_\ell(s, \getOp_\ell(s) \deltaapp \Delta v) &= s \deltaapp \deltaOp{put}_\ell(s, \Delta v)
\end{align*}

\noindent and such that if $\Delta s$ is minimal for $s$ then
$\deltaOp{get}_\ell(s,\Delta s)$ is minimal for $\getOp_\ell(s)$, and likewise
if $\Delta v$ is minimal for $\getOp_\ell(s)$ then $\deltaOp{put}_\ell(s,\Delta
v)$ is minimal for $s$.

The $\deltaOp{get}_\ell$ direction simply performs incremental view maintenance,
which is not our main concern here; we include it to show how it fits together
with $\deltaOp{put}_\ell$ but do not discuss it further. The first equation and
minimality condition is simply delta-correctness of $\deltaOp{get}_\ell$
relative to $\getOp_\ell$.

Our focus here is the $\deltaOp{put}_\ell$ operation. In this direction, it
would be redundant to supply an argument holding the previous value of the view,
since it can be obtained via $\getOp_\ell$. The $\Delta$PutGet rule and
associated minimality condition is a special case of the delta-correctness rule,
where we only consider changes to $V$, not $S$:
\begin{align*}
\putOp_\ell(s, \getOp_\ell(s) \deltaapp \Delta v)
&=
\putOp_\ell(s, \getOp_\ell(s)) \deltaapp \deltaOp{put}_\ell(s, \Delta v)
\end{align*}

\noindent and the term $\putOp_\ell(s, \getOp_\ell(s))$ has been simplified to $s$ by
the GetPut rule.

We can equip the generic lens combinators from \Secref{background:lenses}
with suitable delta-correct $\deltaOp{put}$ operations as follows:
\begin{align*}
  \deltaOp{put}_{\IdPrim{}}(\_,\Delta x) & = \Delta x\\
  \deltaOp{put}_{\SymPrim{}}(\_,(\Delta y,\Delta x)) & = (\Delta
                                                       x,\Delta y)\\
  \deltaOp{put}_{\AssocPrim{}}(\_,((\Delta x, \Delta y), \Delta z)) &
                                                                      = (\Delta x, (\Delta y, \Delta z)) \\
  \deltaOp{put}_{\ell_1;\ell_2}(x,\Delta z) &=
                                      \deltaOp{put}_{\ell_1}(x,\deltaOp{put}_{\ell_2}(\getOp_{\ell_1}(x),\Delta
                                      z))\\
  \deltaOp{put}_{\ell_1 \otimes \ell_2}((x_1,x_2),(\Delta y_1, \Delta y_2)) &=
                                                                    (\deltaOp{put}_{\ell_1}(x_1,\Delta
                                                                    y_1),
                                                                    \deltaOp{put}_{\ell_2}(x_2,\Delta y_2))
\end{align*}
\noindent It is straightforward to show that the resulting
incremental lenses are well-behaved.

For each relational lens primitive $\ell$ described in
\Secref{background:relational-lenses}, $\SelectPrim{P}$, $\DropPrim{A}{X}{a}$,
$\JoinDLPrim$ and $\RenamePrim{A/B}$, we will define an incremental
$\deltaOp{put}_\ell$ operation as follows. First, we incrementalise the
corresponding $\putOp_\ell$ definition from
\Secref{background:relational-lenses}, obtaining a function $\deltaOp{Put}_\ell :
(S \times \Delta S) \times (V \times \Delta V) \to \Delta S$ that is
delta-correct with respect to $\putOp_\ell$.  Since we are only interested in
the case where $S$ does not change and $v = \getOp_\ell(s)$, we then specialize
this operation to obtain $\deltaOp{put}_\ell(s,\Delta v) =
\deltaOp{Put}_\ell((s,\emptyset),(\getOp_\ell(s),\Delta v))$, which yields a
well-behaved lens.  We then apply further optimisations to simplify this
expression to a form that can be evaluated efficiently.

The well-behavedness of the generic lens combinators and the relational lens
primitives imply the well-behavedness of any well-typed lens expression.

\subsection{Select lens}

The incremental lens $\ell = \DeltaSelectPrim{P}: \vtype{U}{Q}{F} \lensto
\vtype{U}{P \wedge Q}{F}$ is the lens $\SelectPrim{P}$ of the same type,
equipped with $\deltaOp{put}_\ell$ defined as follows:
\[\begin{array}{lrcl}
& \deltaOp{put}_\ell
  &:&
  \vtype{U}{Q}{F} \times \Delta\vtype{U}{P \wedge Q}{F}
  \to
  \Delta\vtype{U}{Q}{F}
  \\
& \deltaOp{put}_\ell(M, \Delta N) &=& \letin{N}{\select{P}{M}} \\
&&& \letin{(M_0, \Delta M_0)}{\Deltarelmerge{\Deltaselect{\neg P}{M, \emptyset}}{F}{(N, \Delta N)}} \\
&&& \letin{(N_\#, \Delta N_\#)}{\Deltaselect{P}{M_0, \Delta M_0} \Deltasetminus (N, \Delta N)} \\
&&& (M_0, \Delta M_0) \deltasetminus (N_\#, \Delta N_\#)
\end{array}\]

\begin{lemma}
   \label{lem:delta-putget-select}
The incremental select lens $\DeltaSelectPrim{P}$ is well-behaved.
\end{lemma}

\begin{definition}
\label{def:delta-select-opt}

Define an optimised incremental $\DeltaSelectPrim{P}$ lens $\ell'$ with
$\deltaOp{put}_{\ell'}$ defined as follows:
\[\begin{array}{lrcl}
& \deltaOp{put}_{\ell'}(M, \Delta N)
&=& \letin{Q}{\affected_F(\positive{\Delta N})} \\
&&& \letin{\Delta M_0}{(\relmerge{\select{Q \wedge \neg P}{M}}{F}{\positive{\Delta N}} \deltaminus \select{Q \wedge \neg P}{M}) \deltaminus \negative{\Delta N}} \\
&&& \letin{\Delta N_\#}{(\select{P}{\pDelta{M_0}},\select{P}{\nDelta{M_0}}) \deltaminus \Delta N} \\
&&& \Delta M_0 \deltaminus \Delta N_\#
\end{array}\]
\end{definition}

The optimised version works as follows. $\Delta M_0$ can be calculated by
querying the database for $\select{Q \wedge \neg P}{M}$ and then performing
relational merge using $\Delta N^+$.  The remaining computations involve only
deltas and can be performed in-memory. $\Delta M_0$ contains all changes to the
underlying table including any removed rows, but does not account for rows which
previously didn't satisfy $P$, but do after the updates. These rows, which would
violate lens well-behavedness, are found in $\Delta N_\#$. We calculate $\Delta
N_\#$ just using the delta difference operator $\deltaminus$ because $N_\#
\deltaplus \Delta N_\#$ is always a subset of $M_0 \deltaplus \Delta M_0$. The
final update consists of the changes to the table $M_0$ merged with the changes
to remove all rows in $\Delta N_\#$.

\begin{restatable}{theorem}{restateselectoptimised}[Correctness of optimised select lens]
   \label{thm:select:optimised}
   Suppose $N = \select{P}{M}$ where $M :\vtype{U}{Q}{F}$.  Suppose also that $\Delta N$ is
   minimal with respect to $N$ and that
   $N \deltaapp \Delta N : \vtype{U}{P \wedge Q}{F}$.  Then
   $\deltaOp{put}_\ell(M,\Delta N) = \deltaOp{put}_{\ell'}(M,\Delta
   N)$.
\end{restatable}

\subsection{Project lens}

The incremental lens $\ell = \DeltaDropPrim{A}{X}{a}: \vtype{U}{P}{F}\lensto
\vtype{U \setminus A}{\proj{P}{U \setminus A}}{F'}$, where $F \equiv F' \uplus
\set{X \to A}$, is the lens $\DropPrim{A}{X}{a}$ of the same type, equipped with
$\deltaOp{put}_\ell$ defined as follows:
\[\begin{array}{lrcl}
& \deltaOp{put}_\ell
  &:&
  \vtype{U}{P}{F} \times \Delta\vtype{U \setminus A}{\proj{P}{U \setminus A}}{F'}
  \to
  \Delta\vtype{U}{P}{F}
  \\
& \deltaOp{put}_\ell(M, \Delta N)
&=& \letin{N}{\proj{M}{U \setminus A}} \\
&&& \letin{(M', \Delta M')}{(N, \Delta N) \DeltaJoin (\set{\set{A = a}}, \emptyset)}\\
&&& \deltarelrevise{(M', \Delta M')}{X \to A}{(M, \emptyset)}
\end{array}\]

\begin{lemma}
   \label{lem:delta-putget-project}
The incremental projection lens $\DeltaDropPrim{A}{X}{a}$ is well-behaved.
\end{lemma}

\begin{definition}
\label{def:delta-project-opt}
Define an optimised incremental $\DeltaDropPrim{A}{X}{a}$ lens $\ell'$ with
$\deltaOp{put}_{\ell'}$ defined as follows:
\[\begin{array}{lrcl}
& \deltaOp{put}_{\ell'}(M, \Delta N) &=& \letin{\Delta M'}{(\posDN \Join \set{\set{A = a}}, \negDN \Join \set{\set{A = a}})} \\
&&& (\relrevise{\pDelta{M'}}{X \to A}{M}, \relrevise{\nDelta{M'}}{X \to A}{M})
\end{array}\]
\end{definition}

\noindent $\Delta M'$ extends $\Delta N'$ with the extra attribute $A$ set to
the default value $a$, to match the domain of the underlying table. The final
step optimises the use of $\deltarelrevise{\cdot}{X \to A}{\cdot}$ in
$\deltaOp{put}_\ell$ using Lemma~\ref{lem:delta_rel_rev}.

\begin{restatable}{theorem}{restateprojectoptimised}[Correctness of optimised project lens]
   \label{thm:project:optimised}
   Suppose $M :\vtype{U}{P}{F}$ and  $N = \proj{M}{U \setminus A}$.  Suppose also that $\Delta N$ is
   minimal with respect to $N$ and that
   $N \deltaapp \Delta N : \vtype{U\setminus A}{\proj{P}{U \setminus A}}{F'}$, where $F \equiv
   F' \uplus \{X \to A\}$.  Then
   $\deltaOp{put}_\ell(M,\Delta N) = \deltaOp{put}_{\ell'}(M,\Delta
   N)$.
\end{restatable}

\subsection{Join lens}

The incremental lens $\ell = \DeltaJoinDLPrim: \vtype{U}{P}{F} \otimes
\vtype{V}{Q}{G} \lensto \vtype{U \cup V}{P \Join Q}{F \cup G}$ is the lens
$\JoinDLPrim$ of the same type, equipped with $\deltaOp{put}_\ell$ defined as
follows:
\[\begin{array}{lrcl}
&	\deltaOp{put}_\ell &:&
   \vtype{U}{P}{F} \times \vtype{V}{Q}{G}
   \times \Delta\vtype{U \cup V}{P \Join Q}{F \cup G}\\
&&&
   \to \Delta\vtype{U}{P}{F} \times \Delta\vtype{V}{Q}{G}
	\\
&	\deltaOp{put}_\ell((M,N), \Delta O)
&=&
\letin{O}{M \Join N}
\\
&&&
\letin{(M_0, \Delta M_0)}{\Deltarelmerge{(M, \emptyset)}{F}{\Deltaproj{O, \Delta O}{U}}}
\\
&&&
\letin{(N', \Delta N')}{\Deltarelmerge{(N, \emptyset)}{G}{\Deltaproj{O, \Delta O}{V}}}
\\
&&&
\letin{(L, \Delta L)}{((M_0, \Delta M_0) \DeltaJoin (N', \Delta N')) \Deltasetminus (O, \Delta O)}
\\
&&&
\letin{\Delta M'}{(M_0, \Delta M_0) \deltasetminus \Deltaproj{L, \Delta L}{U}}
\\
&&&
(\Delta M', \Delta N')
\end{array}\]

\begin{lemma}
   \label{lem:delta-putget-join}
   The incremental join lens $\DeltaJoinDLPrim$ is well-behaved.
\end{lemma}

\begin{definition}
\label{def:delta-join-opt}
Define an optimised incremental $\DeltaJoinDLPrim$ lens $\ell'$ with $\deltaOp{put}_{\ell'}$
defined as follows:
\[\begin{array}{lrcl}
&	\deltaOp{put}_{\ell'}((M,N), \Delta O)
&=& \letin{P_M}{\affected_F(\proj{\positive{\Delta O}}{U})} \\
&&& \letin{P_N}{\affected_G(\proj{\positive{\Delta O}}{V})} \\
&&&
\letin{\Delta M_0}{\relmerge{\select{P_M}{M}}{F}{\proj{\positive{\Delta O}}{U}} \deltaminus \select{P_M}{M}}
\\
&&&
\letin{\Delta N'}{\relmerge{\select{P_N}{N}}{G}{\proj{\positive{\Delta O}}{V}} \deltaminus \select{P_N}{N}}
\\
&&&
\texttt{let}\;\Delta L =
(((M \deltaplus \Delta M_0) \Join \positive{\Delta N'})
 \cup (\positive{\Delta M_0} \Join (N \deltaplus \Delta N')), \\
&&&
\phantom{\texttt{let}\;\Delta L}\mathrel{\phantom{=}}
\;\,(\negative{\Delta M_0} \Join N) \cup (M \Join \negative{\Delta N'})) \deltaminus \Delta O\;\texttt{in}
\\
&&&
\letin{\Delta M'}{\Delta M_0 \deltaminus \proj{\Delta L}{U}}
\\
&&&
(\Delta M', \Delta N')
\end{array}\]
\end{definition}

In the optimised join lens, $\Delta M_0$ and $\Delta N'$ can be calculated by
first querying $\select{P_M}{M}$ and $\select{P_N}{N}$, where $P_M$ and $P_N$
include all rows potentially affected by merging functional dependencies, and
then performing the appropriate relational merges using $\proj{\pDelta{O}}{U}$
and $\proj{\pDelta{O}}{V}$. Since this step may result in additional rows being
generated or deleted rows not being removed, any excess rows $\Delta L$ are
determined by calculating which rows would have been changed in the joined view
after updates to the underlying tables $\Delta M_0$ and $\Delta N'$, and then
comparing those to the desired changes $\Delta O$. We can calculate $\Delta L$
efficiently by querying the underlying $M$ and $N$ tables only for records
having identical join keys to records in $\Delta M_0$ and $\Delta N'$.

Finally, the updated left table can be calculated as the changes to the left
table $\Delta M_0$ minus all records that need to be removed to ensure the lens
is well behaved. The changes $\Delta N'$ are used for the right table.

\begin{restatable}{theorem}{restatejoinoptimised}[Correctness of
  optimised join lens]
   \label{thm:join:optimised}
   Suppose $ M : \vtype{U}{P}{F}$ and $N :
\vtype{V}{Q}{G}$ and $O = M \Join N$.  Suppose also that $\Delta O$ is minimal with
respect to $O$, and $O \deltaplus \Delta O : \vtype{U \cup V}{P \Join Q}{F \cup G}$.  Then $
   \deltaOp{put}_\ell((M, N),\Delta O) =
   \deltaOp{put}_{\ell'}((M,N),\Delta O)$.
\end{restatable}

\subsection{Rename lens}

The incremental lens $\ell = \DeltaRenamePrim{A/B} : \vtype{U}{P}{F} \lensto
\vtype{U[A/B]}{\rename{A/B}{P}}{F[A/B]}$, where $A \in U$ and $B \notin U$, is
the lens $\RenamePrim{A/B}$ with the additional function $\deltaOp{put}_\ell$
defined as follows:
\[\begin{array}{lrcl}
& \deltaOp{put}_\ell
  &:&
  \vtype{U}{P}{F} \times \Delta\vtype{U[A/B]}{\rename{A/B}{P}}{F[A/B]}
  \to
  \Delta\vtype{U}{P}{F}
  \\
& \deltaOp{put}_\ell(M, \Delta N)
&=& \letin{N}{\rename{B/A}{M}} \\
&&& \letin{(M', \Delta M')}{\Deltarename{B/A}{N,\Delta N}} \\
&&& \Delta M'
\end{array}
\]

\begin{definition}
\label{def:delta-rename-opt}
Define an optimised incremental $\DeltaRenamePrim{A/B}$ lens $\ell'$ with
$\Delta\putOp_{\ell'}$ defined as follows:
\[\begin{array}{lrcl}
		& \deltaOp{put}_{\ell'}(M, \Delta N) &=& (\rename{B/A}{\positive{\Delta N}}, \rename{B/A}{\negative{\Delta
    N}})
\end{array}
\]
\end{definition}

The optimised rename lens performs the inverse rename operation on both
components of the delta relation. No database queries are required.

\begin{restatable}{theorem}{restaterenameoptimised}[Correctness of rename lens]
   \label{thm:rename:optimised}
Suppose $M : \vtype{U}{P}{F}$ and $N= \rename{A/B}{M}$, and that $\Delta
N$ is minimal with respect to $N$ and satisfies $N \deltaapp \Delta N
: \vtype{U \cup V}{P \Join Q}{F \cup G}$.  Then $\deltaOp{put}_\ell(M,\Delta N) =
\deltaOp{put}_{\ell'}(M,\Delta N)$.
\end{restatable}

\section{Evaluation}\label{sec:evaluation} We have implemented both naive and
incremental relational lenses in Links. Given the lack of an existing
implementation of relational lenses, for the naive version we implemented the
lenses described by \citet{bohannon2006relational}. A benefit is a fairer
performance comparison, as the non-incremental relational lenses are implemented
using the same set operation implementation as the incremental version.  We
evaluate the performance of the optimised $\deltaOp{put}$ operations defined
earlier.

Each performance experiment follows a similar pattern and are all executed on an
Intel(R) Core(TM) i5-6500 with 16GB of RAM and a mechanical hard disk.  The
computer was running Ubuntu 17.10 and PostgreSQL 9.6 was used to host the
database.  Our customised Links version was compiled using OCaml 4.05.0.  All
generated tables contained a primary key index and no other indices.

\subsection{Microbenchmarks}
\subsubsection{Lens primitives}
\label{sec:lens-primitives}

We first evaluate lens change-propagation performance as a set of
microbenchmarks over the lens primitives, using two metrics: total time to
compute the source delta for a single lens given a view delta as input, or in
the case of the naive lenses, to calculate new source tables, referred to as
\emph{total execution time}. We also measure the amount of the total execution
time which can be attributed to query execution.

The following steps are taken for each benchmark:

\begin{enumerate}
	\item Generate the required tables with a specified set of columns and fill the tables with random data.
	For test purposes, the random data can either be sequential, a bounded random number or a random number. The microbenchmarks use the following tables:

	\begin{itemize}
		\item Table $t_1$ with domain $A$, $B$ and $C$ and functional dependency $A \to B\ C$.
			Populated with $n$ rows, with $A$ calculated as a sequential value, $B$ a random number up to $n / 10$ and $C$ a random number up to $100$.
		\item Table $t_2$ with domain $B$ and $D$ and functional dependency $B \to D$.
			Populated with $n/10$ rows with $B$ being a sequential value and $D$ a random number up to $n / 10$.
	\end{itemize}
	\item Generate lenses for the underlying tables and compose the required lenses on top of these.
		Lenses that use both $t_1$ and $t_2$ always start by joining the two tables on column $B$.
	\item Fetch the output view of the lens using $\getOp$ and then make changes in a systematic fashion as described for each setup.
	The changes are designed to affect a small portion of the database.
	Most changes are of the form: ``take all rows with attribute $A$ having some value, and update attribute $B$''.
	\item Apply $\putOp$ for the first lens using the updated view.
	For the incremental lenses this means first calculating the
  delta from the view (not timed) and then timing $\deltaOp{put}$ which calculates the source delta.
	For naive lenses we measure the time required to
        recompute the source table using non-incremental $\opname{put}$, but not the time needed to update
        the database.
	This process is repeated multiple times and the median value taken.
\end{enumerate}

We repeat this process for each of the join, select and projection lenses.  We
exclude the time required to calculate the view delta or update the database
because these operations only need to be performed once per lens, regardless of
the number of intermediate steps defining the lens. Instead, we measure the
performance of these one-time costs in
\SecrefTwo{change_set_calc_perf}{change_set_app_perf}, and present a complete
example involving two lens primitives in \Secref{dblp_perf}.

\begin{figure}
	\centering
	\begin{subfigure}{0.33\textwidth}
		\includegraphics[width=\textwidth]{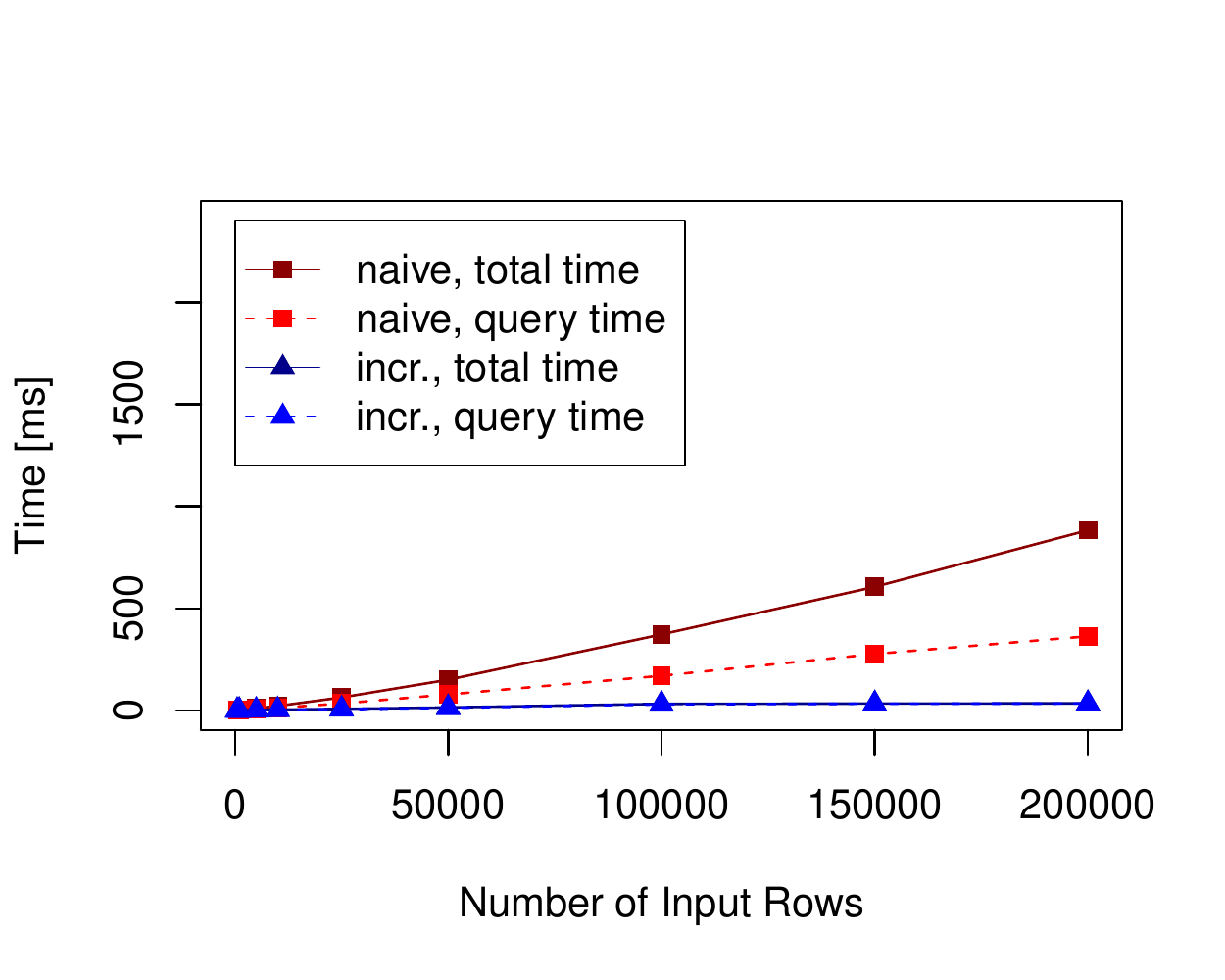}
		\caption{select lens}
		\label{fig:benchmark_select}
	\end{subfigure}
	~
	\begin{subfigure}{0.33\textwidth}
		\includegraphics[width=\textwidth]{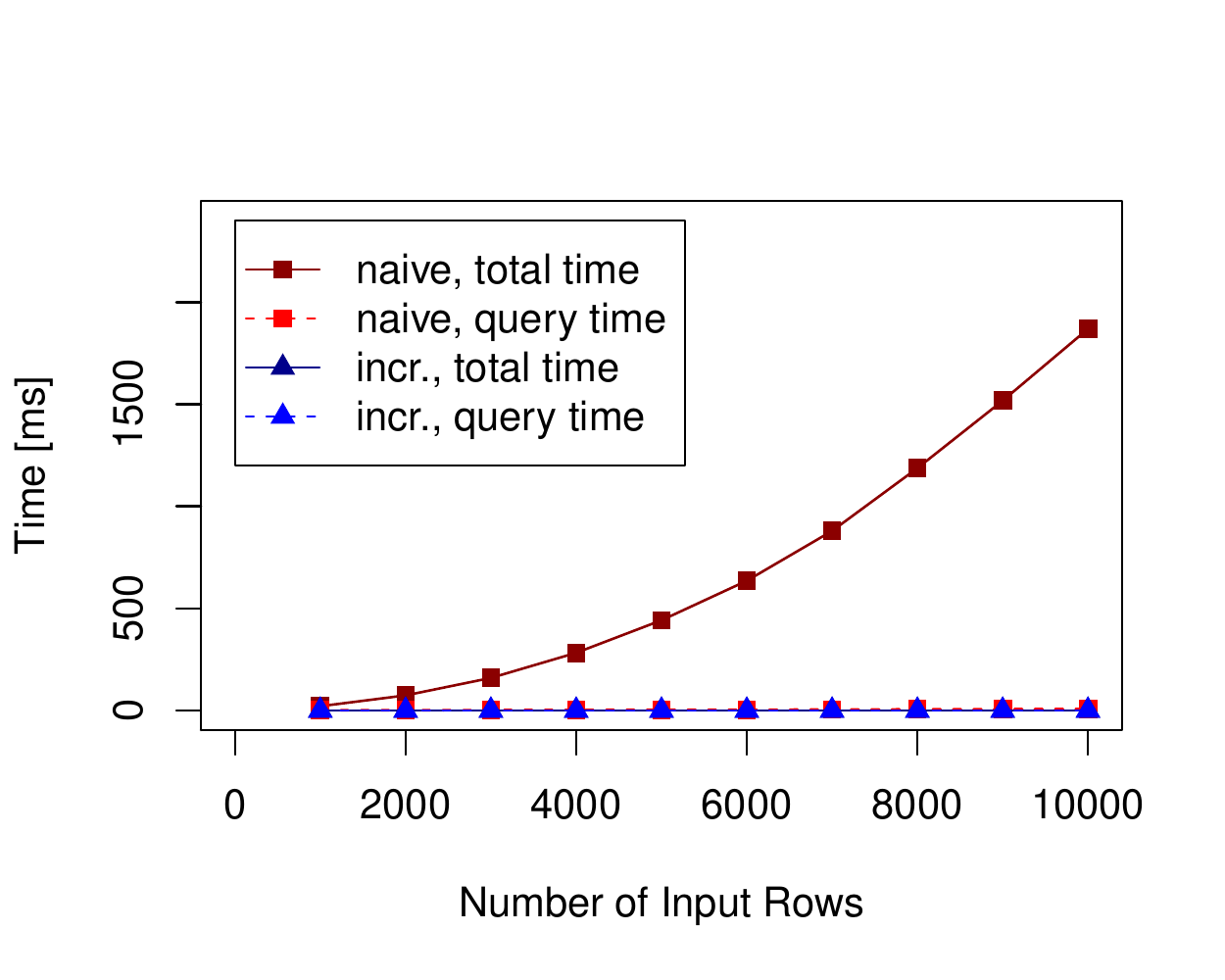}
		\caption{projection lens}
		\label{fig:benchmark_drop}
	\end{subfigure}
	~
	\begin{subfigure}{0.33\textwidth}
		\includegraphics[width=\textwidth]{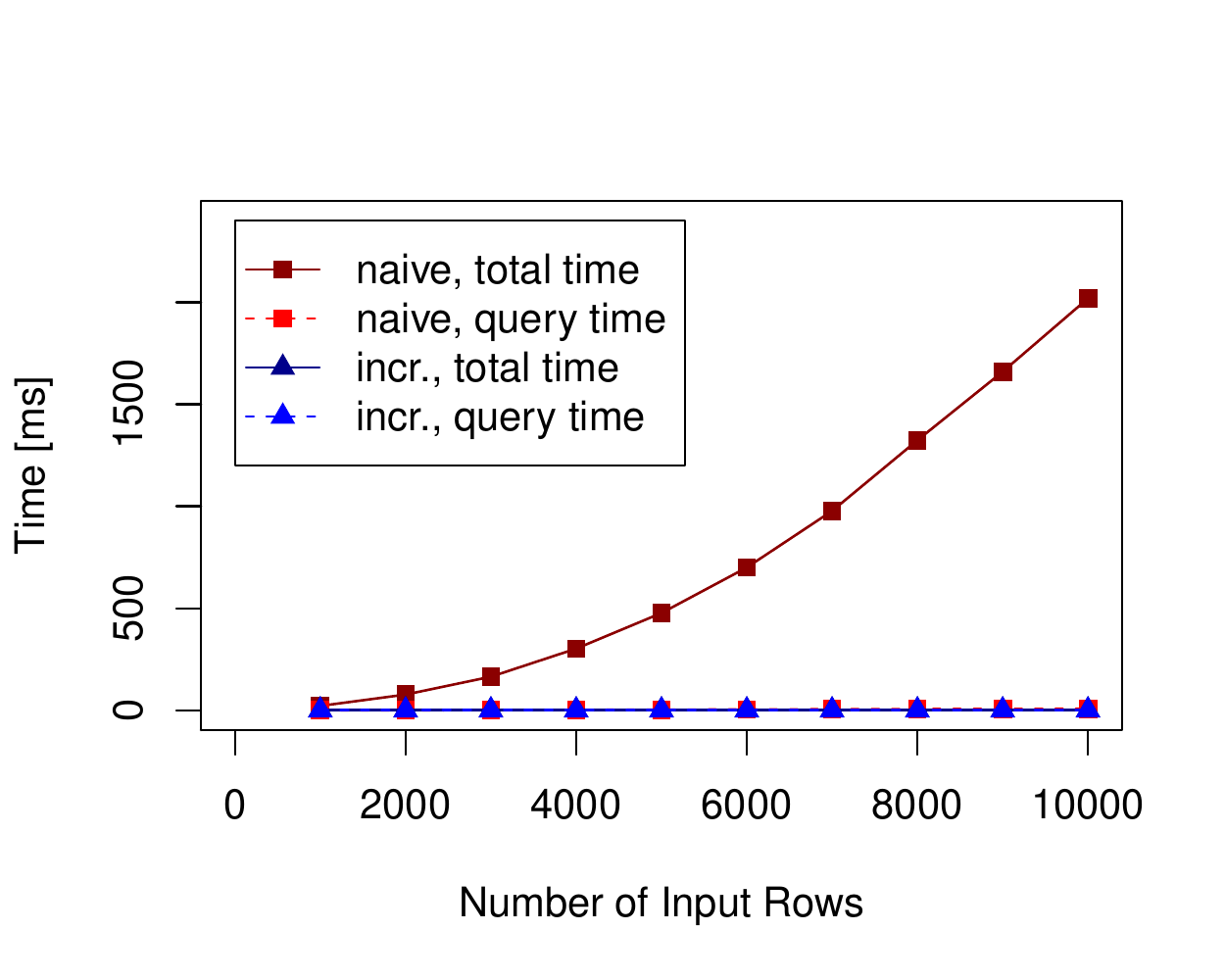}
		\caption{join lens}
		\label{fig:benchmark_join}
	\end{subfigure}
	\caption{Total and query execution time required by individual
          lens primitives vs. underlying table sizes}
\end{figure}

\paragraph{Select example}
\label{example:eval_select}
The select benchmark uses both $t_1$ and $t_2$.
We compose a select lens on top of the join lens with the predicate $C
= 3$, which produces a view with an average size of $n/100$ entries.
The view is modified so that all records with $0 \le B \le 100$ have their $D$ value set to 5.
This approach produces deltas of length 20 on average, containing one row removal and one row addition.

The naive $\opname{put}$ operation makes use of a single query and requires a
total computation time of between \ms{1} and \ms{867} depending on the row
count. While the performance is acceptable for small tables, it is still too
slow for most applications as the tables become larger. It also shows how an
unavoidable bottleneck is introduced, as the query time reaches up to \ms{363}.

Depending on the row count, the incremental version only needs between \ms{1}
and \ms{40} total computation time. Of this time, between \ms{1} and \ms{39} are
used to perform the required queries, accounting for the majority of the total
computation time. The execution time and query time scale proportionally to the
data size. The incremental performance reflects the fact that the view is much
smaller than the entire table, which needs to be recomputed for the
non-incremental version. It may be possible to improve performance by
configuring the database to index $C$, since this may reduce the query execution
time. Indexing would not affect naive performance, since we always fetch the
entire source tables.

\paragraph{Projection performance}
We define a drop lens over the table $t_1$ removing attribute $C$, which is defined by attribute $A$ with a default value of $1$.
The view is modified by setting $B$ to $5$ for all records where $60 < A < 80$.
This process modifies $20$ of the $n$ records in the view.

The performance of the lens is shown in Figure \ref{fig:benchmark_drop}. As in
the case of the join lens, the naive projection lens implementation quickly
becomes infeasible, requiring a total execution time of over \ms{2000} as it
processes $10000$ records. The naive version spends up to \ms{8} querying the
server. The incremental version is able to perform the $\deltaOp{put}$ operation
in \ms{1} and less for the given row counts. This time includes the time
required to query the database server for additional information.

\paragraph{Join example}
\label{example:eval_join}
The join benchmark uses the join lens defined over the two tables $t_1$ and $t_2$.
We fetch the resulting view, which will contain $n$ rows.
After that we modify all records containing a value for $B$ between $40$ and $50$ and set their $C$ value to $5$.

We benchmark the described setup with $n$ values ranging from $1000$ to $10000$
in increments of $1000$, timing the lens $\opname{put}$ duration for each $n$ as
specified. The performance results are shown in Figure \ref{fig:benchmark_join}.
The $\opname{put}$ operation for the naive join requires two queries but quickly
becomes impractical. At $10000$ rows it already requires over \ms{2000} and
continues to rise quadratically. Of the computation time, approximately \ms{1}
to \ms{9} depending on the table size is required for querying the database.
While the query time taken by the naive approach is relatively low, this is due
to the fact that the tables are relatively small and the time increases to
hundreds of milliseconds as the table size grows to hundreds of thousands of
rows.

In comparison, the incremental approach can scale to hundreds of thousands of
rows and requires only \ms{1} to \ms{2} of both computation and query time for
the given views. It requires 5 queries which are all simple to compute and
return small views.

\paragraph{Summary}

 The above experiments show that incremental evaluation outperforms
naive evaluation of relational lenses for data sizes up to 10K rows.
We have also measured the performance of incremental evaluation for
200K rows to confirm that incremental performance continues to scale.
Table~\ref{tab:largedata} shows the number of queries, query evaluation time
and total evaluation time for all three microbenchmarks discussed
above.

\begin{table}[h]
\caption{Query counts and times for large data sizes}\label{tab:largedata}
\begin{equation*}
	\begin{array}{c|ccc}
		& \text{select} & \text{project} & \text{join} \\
		\hline
		\text{query count}	& 1			& 1			& 5		\\
		\hline
		\text{query } n=200k	& 37ms		& <1ms		& 1ms	\\
		\text{total } n=200k  & 39ms		& < 1ms		& 2.5ms	\\
	\end{array}
\end{equation*}
\end{table}

\begin{figure}
	\centering
	\begin{subfigure}{0.33\textwidth}
		\includegraphics[width=\textwidth]{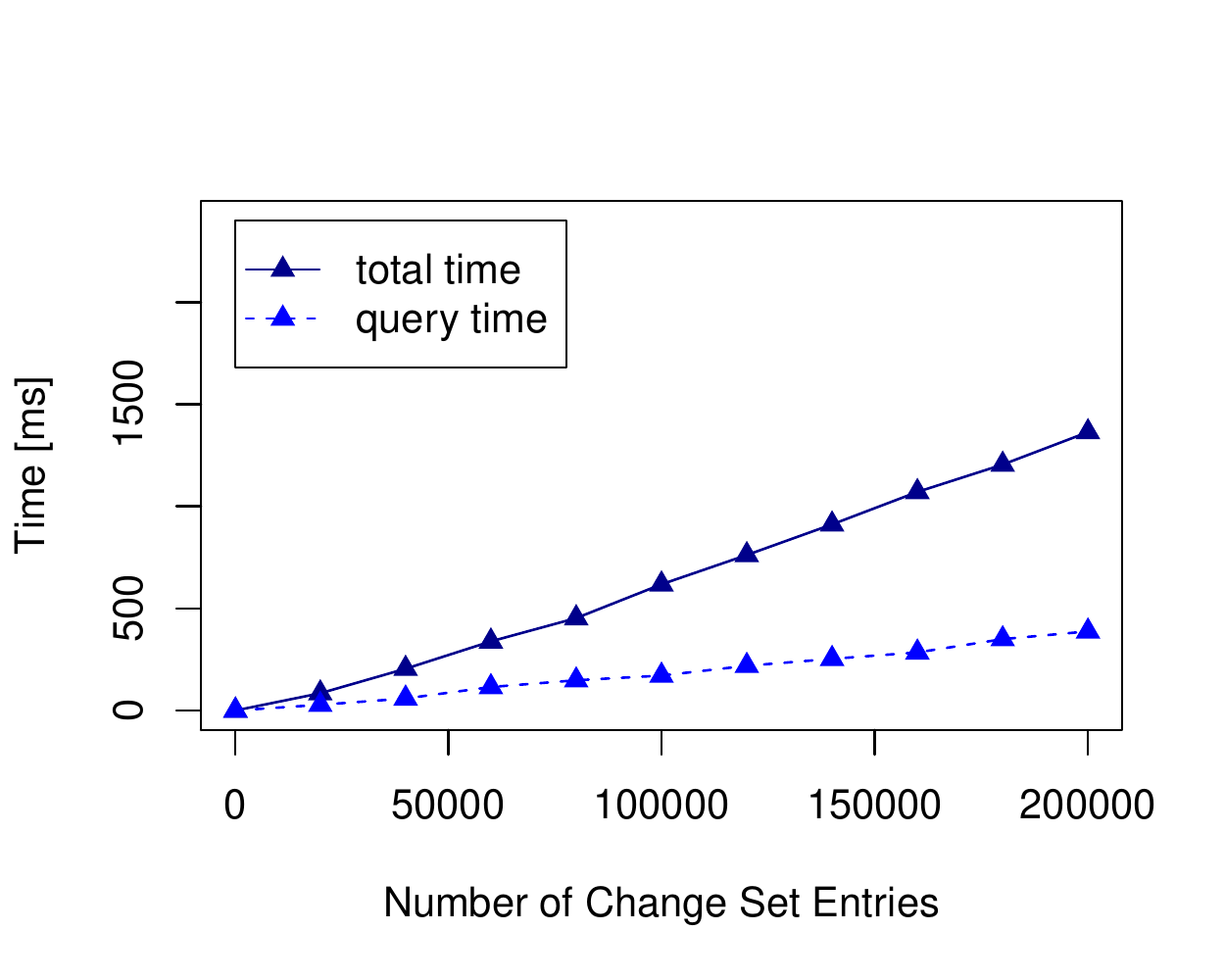}
		\caption{Time needed to calculate\\ delta between old
                  and new view}
		\label{fig:benchmark_delta_get}
	\end{subfigure}
	~
	\begin{subfigure}{0.33\textwidth}
		\includegraphics[width=\textwidth]{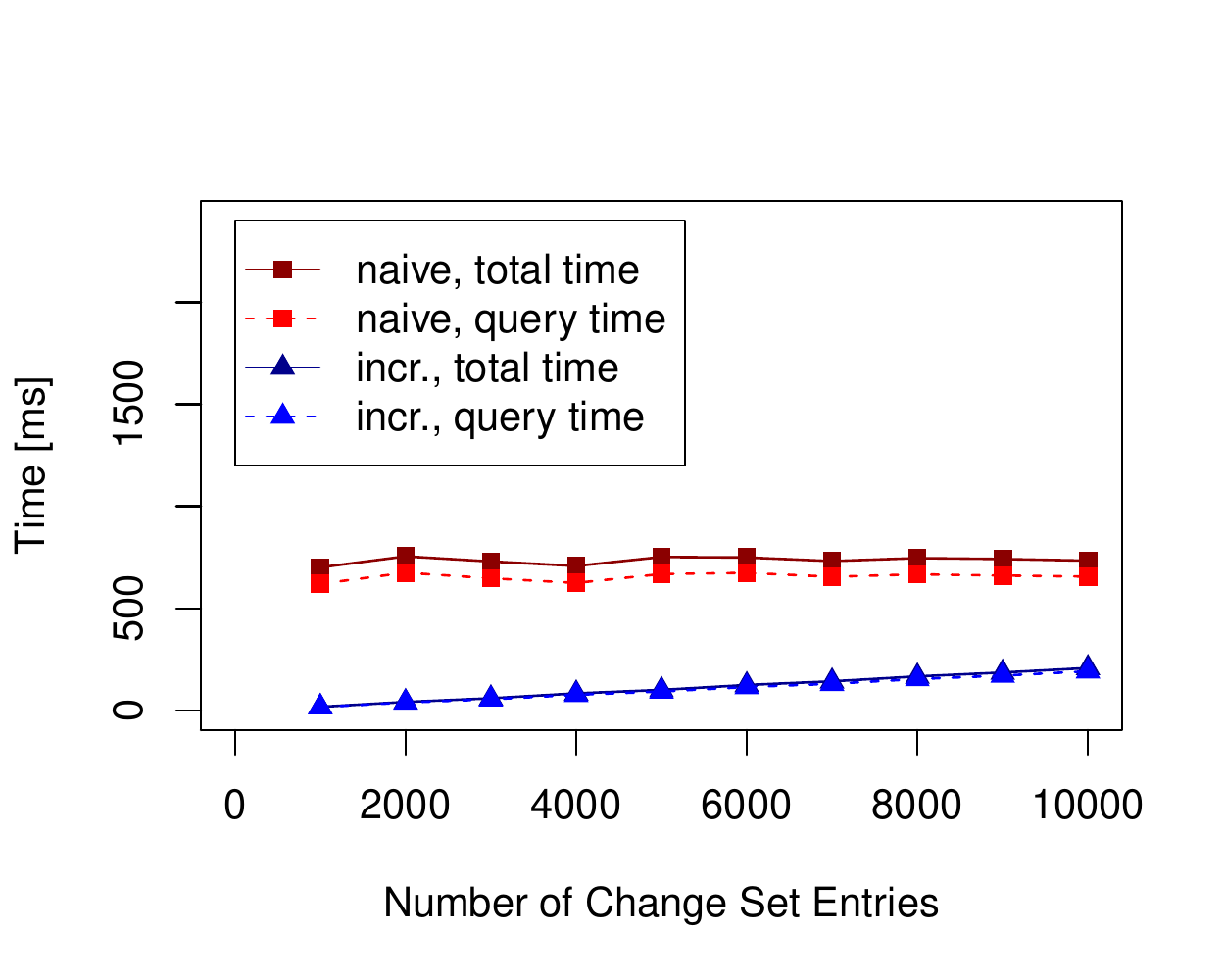}
		\caption{Time needed to apply\\source delta to the database}
		\label{fig:benchmark_delta_put}
	\end{subfigure}
	~
	\begin{subfigure}{0.33\textwidth}
		\includegraphics[width=\textwidth]{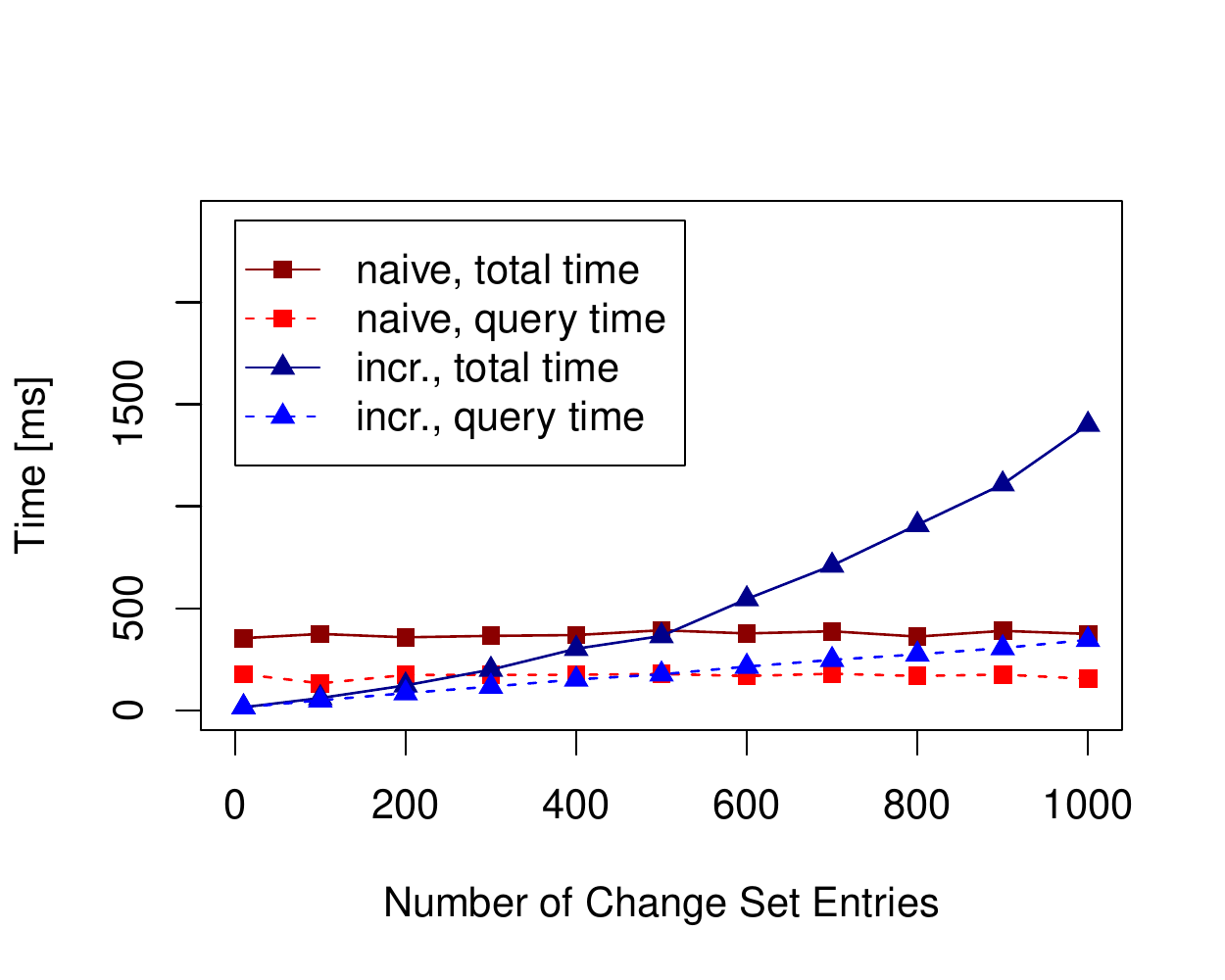}
		\caption{Delta propagation time as a function of view
                  update size (n=100000)}
		\label{fig:benchmark_ds_select}
	\end{subfigure}
	\caption{Evaluation of delta calculation, delta application
          and delta propagation time as a function of view update size.}
	\label{fig:plot_delta}
\end{figure}

\subsubsection{Delta calculation performance}
\label{sec:change_set_calc_perf}

While microbenchmarks on lens primitives give us some insight into the
performance of the lenses, they do not account for the time required to
calculate the initial delta, which is only required for incremental lenses. We
modify the view of the join lens defined over $t_1$ and $t_2$ by fetching the
view using $\getOp$, and by then performing changes as done in the other
experiments. Specifically we set $B$ to $5$ for all records where $0 < D < 10$.
Given that this example does not have any selection lenses, the size of the view
will always be $n$.

We measure the time taken to fetch the unchanged view from the database and then
subtract it from a modified view.  As in the previous examples we measure both
the time required to query the database server as well as the total execution
time on the client. We measure the time required for $n$ values ranging between
$100$ and $200000$.  The results are shown in Figure
\ref{fig:benchmark_delta_get}.

Both the query and execution time are approximately linear with respect to the
number of input rows. We require between \ms{<1} and \ms{1360} to compute the
delta, of which just under half (\ms{396} for $200000$ rows) is spent querying
the database.

\subsubsection{Delta application performance}
\label{sec:change_set_app_perf}

We also measure the time it takes to apply a delta to a table. This process
requires the generation of insert, update and delete SQL commands which must
then be executed on the server.

We consider delta application for a single table. We use the table $t_1$ from
\Secref{lens-primitives} and populate it with $n=10000$. We generate a delta
containing $m$ entries, where a quarter of the entries produce $m / 4$
insertions, another quarter produce $m / 4$ deletions and the remaining half
produce $m / 2$ updates.

Given such a delta, we time how long it takes to produce the SQL commands from
the already calculated delta with varying size $m$. The SQL update commands are
concatenated and sent to the database together as a single transaction. As in
the other cases we time both the total and query execution times, which are
shown in Figure \ref{fig:benchmark_delta_put}. For the naive version, we
generate an update command that deletes the current contents of the table and
inserts the new contents.  For the incremental version, we generate updates that
insert, delete, or replace only affected records.

The figure shows that the naive version's performance is independent of the
number of changes, requiring around \ms{718}, most of which is spent querying
the database. The incremental version, on the other hand, requires less time for
the given change sizes and scales linearly, requiring between \ms{18} and
\ms{207} depending on the delta size.

\subsubsection{Varying delta size}

In addition to varying the size of the underlying database tables we also
consider how the size of the delta may affect the performance of an update. To
do this we use the two tables $t_1$ and $t_2$ with $n=100000$ and define a
select lens on top of the join lens with the predicate $C = 3$. We then
determine a $b'$, starting from $0$ and in steps of $100$, so that modifying all
records where $0 < B < b'$ by setting $D = 5$ produces a delta of size greater
than $m$.

As in the other microbenchmarks we measure the total and query execution time
taken to perform the $\deltaOp{put}$ of an already calculated view delta or, in
the naive case, the time to recalculate the full source tables using
$\opname{put}$. We repeat this experiment for varying $m$ values, ranging from
$10$ to $1000$. The resulting execution times are plotted in Figure
\ref{fig:benchmark_ds_select}.

As would be expected, the naive lens is relatively constant regardless of the
number of entries in the delta.  While the incremental version starts slows down
as the delta becomes larger, eventually becoming even slower than the naive
version, it does show that when the changes are sufficiently small the
incremental method is much more efficient.  When the size of the deltas exceeds
around $500$, however, it starts to become more efficient to recalculate the
tables.

\subsection{DBLP example}
\label{sec:dblp_perf}

In addition to the microbenchmarks we also perform some experiments on a
real-world example involving the DBLP Computer Science Bibliography
\cite{ley2009dblp}, a comprehensive collection of bibliographic information
about computer science publications. It is published as a large and freely
available XML file with millions of records containing publications,
conferences, journals, authors, websites and more.

Our example uses a table containing a collection of conference publications
called \tblname{inproceedings} as well as a table of their respective authors
\tblname{inproceedings\_author}. This first table contains the \colname{title}
of the paper, the \colname{year} it was published as well as the
\colname{proceedings} it is in, while the \tblname{inproceedings\_author} table
contains an entry for each author on every paper, allowing a single publication
to have multiple authors.
\begin{cameraready}
  More information as well as example data is listed in the full
  version of this paper~\cite{irl-full}.
\end{cameraready}
\begin{techreport}
  More information as well as example data is listed in Appendix
  \ref{appendix:dblp_example}.
\end{techreport}

In order to determine how the application scales with varying database sizes, we
generate the underlying tables by selecting a set of entries so that the join of
\tblname{inproceedings} and \tblname{inproceedings\_author} contains $n$ rows.
Given that set of entries, we select all entries in the complete
\tblname{inproceedings\_author} table, which have a corresponding entry in the
subset of entries chosen.

\begin{figure}
	\centering
	\begin{subfigure}{0.33\textwidth}
		\includegraphics[width=\textwidth]{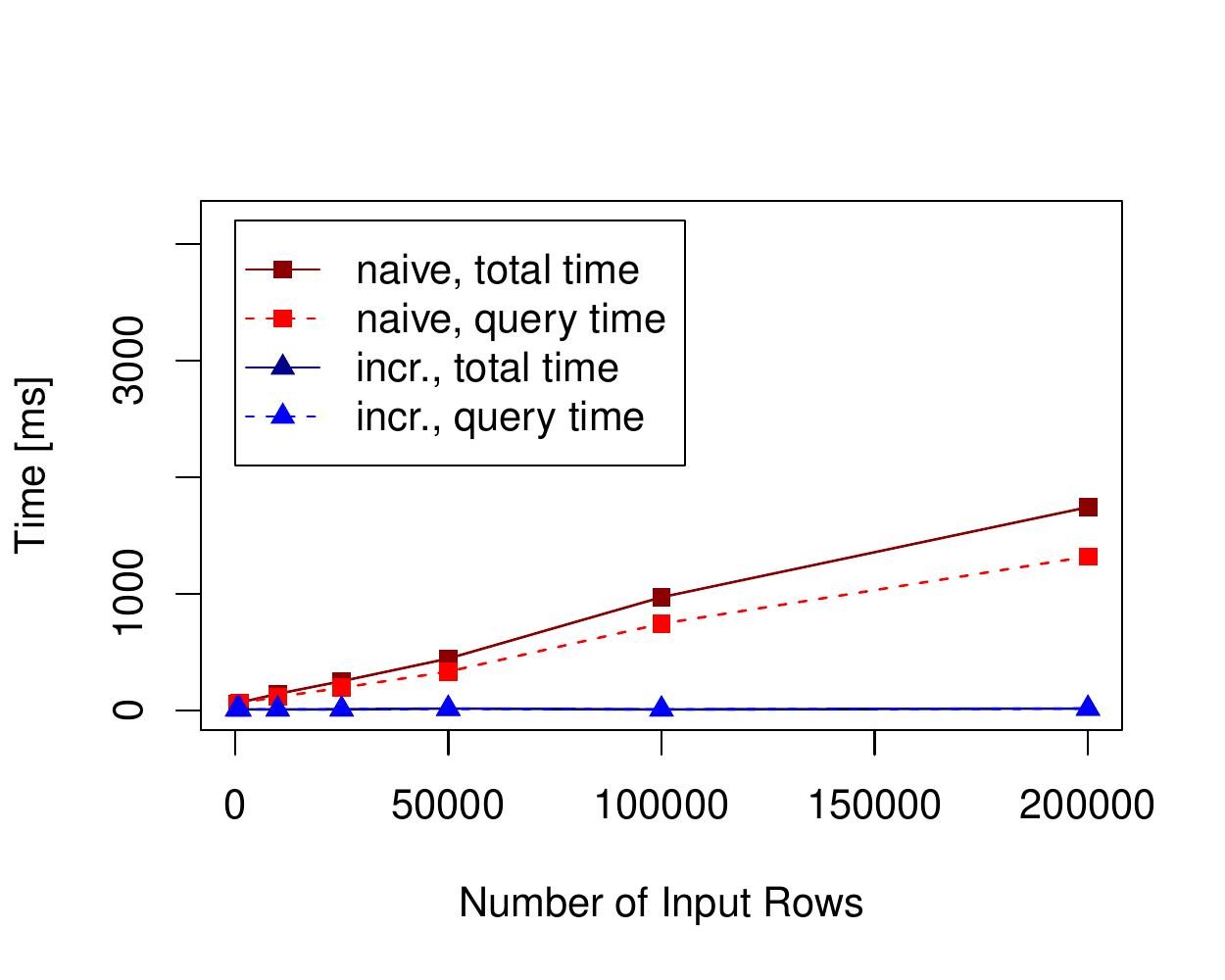}
		\caption{Local database server}
		\label{fig:benchmark_dblp_local}
	\end{subfigure}
	~
	\begin{subfigure}{0.33\textwidth}
		\includegraphics[width=\textwidth]{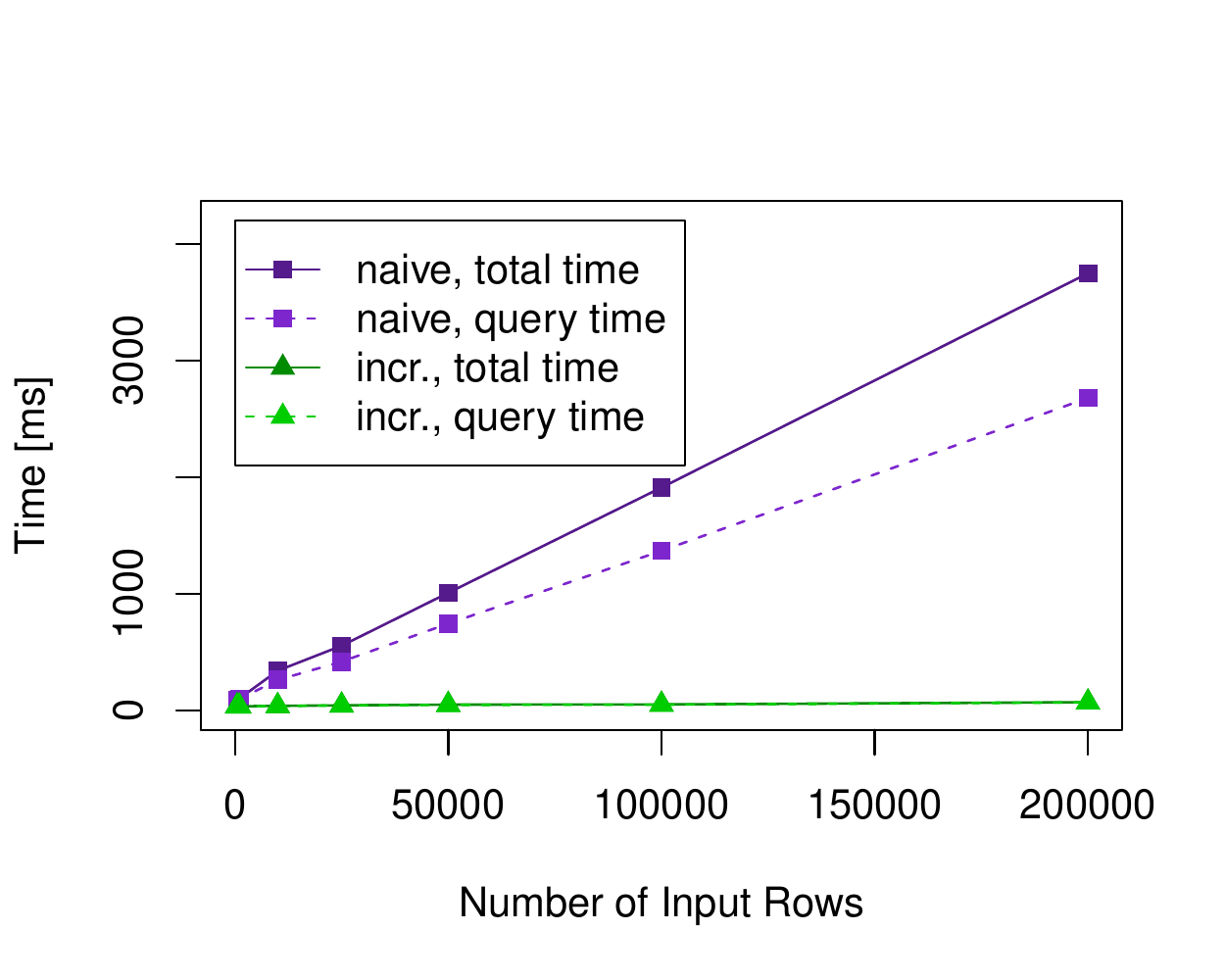}
		\caption{Remote database server}
		\label{fig:benchmark_dblp_remote}
	\end{subfigure}
	~
	\begin{subfigure}{0.33\textwidth}
		\includegraphics[width=\textwidth]{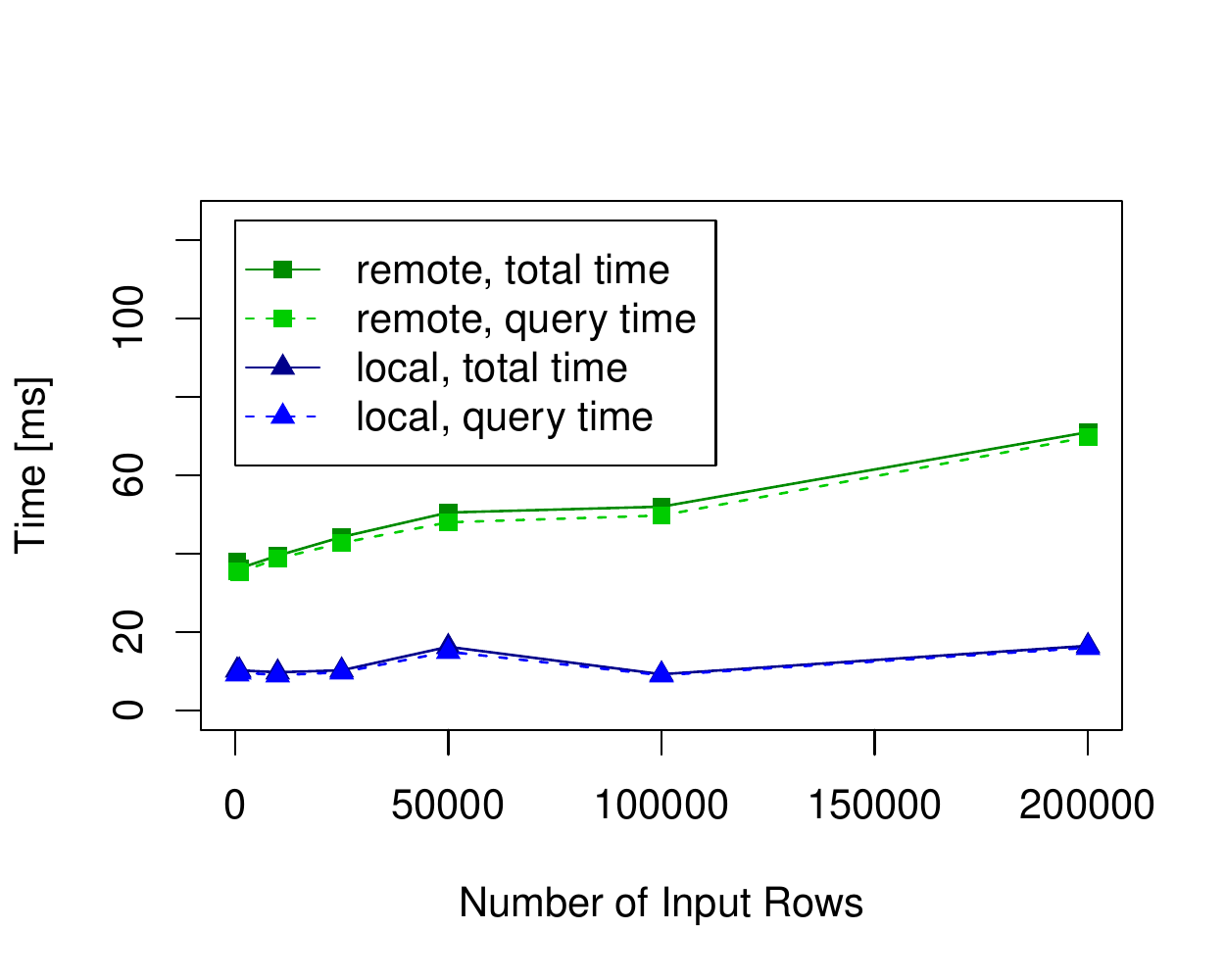}
		\caption{Local / remote performance}
		\label{fig:benchmark_dblp_cmp}
	\end{subfigure}
	\caption{The total/query execution time for the $\putOp$ operation applied to our DBLP database.}
\end{figure}

Using these tables we join the two tables on the
\colname{inproceedings} attribute and then  select all  entries from PODS 2006.
The Links code used to generate the lenses is shown below.

\begin{lstlisting}
	var joinL = lensjoin inproceedings_authorL with inproceedingsL on inproceedings;
	var selectL = lensselect from joinL where proceedings == "conf/pods/2006";
\end{lstlisting}

As in the other examples we fetch the output of the select lens using $\getOp$
and then make a small modification. We then perform the $\putOp$ operation to
apply those changes to the database. During the $\putOp$ we time the entire
process of generating a delta from the view, calculating the delta for the
underlying tables and updating the database for both the naive and incremental
lenses. While timing we keep track of how much time was spent querying the
database and the total time spent performing the operation. We also perform
$\putOp$ using a database located on a remote server and compare it to a
database located on the same machine.

Our performance results with the database hosted on the local machine are shown
in Figure~\ref{fig:benchmark_dblp_local}. Similar to our earlier benchmarks, the
incremental lenses perform favourably in comparison to the naive lenses. The
naive lenses require linear time as the data grows and need up to \ms{1743} to
update the database when the data grows to $200000$ rows. A majority of the
naive execution time (up to \ms{1320}) is used to query the entire database, and
so optimising the local algorithms will have little effect on the overall
performance.

The incremental lenses on the other hand perform much faster even as the data
grows to hundreds of thousands of rows, requiring only between \ms{9.5} and
\ms{17.75} total execution time depending on the size of the underlying tables.
Of that time between \ms{9} and \ms{17} are used to query the database, making
database performance the limiting factor for small datasets.

Figure \ref{fig:benchmark_dblp_remote} shows how the same application ran when
using a database server situated on another computer. This affects the query
time because the bandwidth is limited and the latency becomes higher. Less
bandwidth means that the data loads more slowly, while higher latency imposes a
delay per query. The total execution time of the naive version nearly doubles to
up to \ms{3748}, of which \ms{2679} is used to query the database server, while
the incremental version remains much faster requiring only up to \ms{62}, of
which \ms{59} are used for querying the database.

Figure \ref{fig:benchmark_dblp_cmp} directly compares the performance of
incremental relational lenses when the database is located on the same or a
different machine.  Given that the bandwidth required is relatively small, the
main additional overhead associated with using a remote database server is a
roughly constant increase of about \ms{30}.  Note that the data in
Figure~\ref{fig:benchmark_dblp_cmp} is the same as that for incremental
evaluation in Figures~\ref{fig:benchmark_dblp_local}
and~\ref{fig:benchmark_dblp_remote}, but with the scale of the y-axis adjusted
to allow easier comparison of the two incremental versions.

\section{Related Work}\label{sec:related}

\paragraph{Language-integrated query and web programming}

Our approach is implemented in Links~\cite{cooper2006links}, but should also be
applicable to other functional languages with support for language-integrated
query, such as Ur/Web~\cite{chlipala2015urweb} and F\#~\cite{syme2006linq}, or
the Database-Supported Haskell library~\cite{ulrich2015dsh}.  Our incremental
relational lens definitions could also be implemented inside a database system.
It would be interesting to explore what extensibility features could accommodate
relational lenses in other languages.

\paragraph{Incremental computation}

Incremental view maintenance is a well-studied topic in
databases~\cite{gupta1995maintenance}.  We employ standard incrementalisation
translations for relations (sets of
tuples)~\cite{qian91tkde,griffin1997incremental}.  More recently,
\citet{koch2010incremental} developed an elegant framework for incremental query
evaluation for bags (multisets of tuples), and \citet{koch2016incremental}
extended this approach to \emph{nested} relational queries. We think it would be
very interesting to investigate (incremental) lenses over nested collections or
multisets.

Incremental recomputation also has a large literature, including work on
adaptive functional programming and self-adjusting
computation~\cite{acar2006toplas,hammer2014pldi}. While closely related in
spirit, this work focuses on a different class of problems, namely recomputing
computationally expensive results when small changes are made to the inputs.  In
this setting, recording a large trace caching intermediate results can yield
significant savings if the small changes to the input only lead to small changes
in the trace.  It is unclear that such an approach would be effective in our
setting.  In any case, to the best of our knowledge, this approach has not been
used for database queries or view updates.

Our approach to incrementalisation does draw inspiration from the
\emph{incremental lambda calculus} of \citet{cai2014pldi}.  They used
Koch's incremental multiset operations in examples, but our set-valued
relations and deltas also fit into their framework.  Another relevant
system, SQLCache~\cite{scully2017sqlcache}, shows the value of
language support for caching: in SQLCache, query results (and derived
data) are cached and when the database is updated, dynamic checking is
used to avoid recomputing results if the query did not depend on the
changed data.  However, otherwise SQLCache recomputes the results from
scratch.  Language support for incremental query evaluation could be
used to improve performance in this case.

\paragraph{Updatable views and lenses} Updatable views have been studied
extensively in the database literature, and are supported (in very limited
forms) in recent SQL standards and systems.  We refer to
\citet{bohannon2006relational} for discussion of earlier work on view updates
and how relational lenses improve on it. Although updatable views (and their
limitations) are well-understood, they are still finding applications in current
research, for example for annotation propagation~\cite{buneman02pods} or to
``explain'' missing answers, via updates to the source data that would cause a
missing answer to be produced~\cite{missinganswers}.  \citet{date12view}
discusses current practice and proposes pragmatic approaches to view update.  To
the best of our knowledge, the work that comes closest to implementing
relational lenses is Brul~\cite{zan2016brul}, which builds on top of
BiGUL~\cite{ko2016bigul}, a put-oriented language for programming bidirectional
transformations.  Brul includes the core relational lens primitives and these
can be combined with other bidirectional transformations written in BiGUL.
However, \citet{zan2016brul} implement the state-based definitions of relational
lenses over Haskell lists and do not evaluate their performance over large
databases or consider efficient (incremental) techniques.  They also do not
consider functional dependencies or predicate constraints, so it is up to the
programmer to ensure that these constraints are checked or maintained.
\citet{ko2018axiomatic} recently proposed a Hoare-style logic for reasoning
about BiGUL programs in Agda, which could perhaps be extended to reason about
relational lenses.

Object-relational mapping (ORM) is a popular technique for accessing and
updating relational data from an object-oriented language.  ORM can impose
performance overhead but \citet{bernstein13incremental} show that incremental
query compilation is effective in this setting.  We would like to investigate
whether incremental relational lenses could be composed with more conventional
(edit) lenses to provide ORM-like capabilities for functional languages.

\citet{wang2011incremental} considered incremental updates for efficient
bidirectional programming over tree-shaped data structures, but not relations.
There are several approaches to lenses that are based on translating changes,
including \emph{edit lenses}~\cite{hofmann2012edit}, \emph{delta
lenses}~\cite{diskin2011delta}, \emph{c-lenses}~\cite{johnson2013delta} and
\emph{update lenses}~\cite{ahman2014coalgebraic}.  None of these approaches has
been applied to relational lenses as far as we know.  Rather than utilize (and
recapitulate) the needed technical background for these approaches, we have
opted for concrete approach based on incrementalisation in the style of
\citep{cai2014pldi}, but it would be interesting to understand the precise
relationships among these various formalisms.

Bidirectional approaches to query languages for XML or graph data models have
also been proposed~\cite{hidaka2010unql,liu2007xquery}. However, to the best of
our knowledge these approaches are not incremental and have not been evaluated
on large amounts of data. There is also work on translating updates to \emph{XML
views} over relational data, for example \citet{fegaras2010lineage}; however,
this work does not allow joins in the underlying relations.

\section{Conclusions}
\label{sec:concl}

View update is a classical problem in databases, with applications to database
programming, security, and data synchronisation.  Updatable views seem
particularly valuable in web programming settings, for bridging gaps between a
normalised relational representation of application data and the representation
the programmer actually wants to work with.  Updatable views were an important
source of inspiration for work on lenses in the functional programming languages
community.  There has been a great deal of research on lenses for functional
programming since the influential work of \citet{foster2007combinators}, but
relatively little of this work has found application to the classical view
update problem.  The main exception to this has been
\citet{bohannon2006relational}, who defined well-behaved relational lenses based
on a type system that tracks functional dependencies and predicate constraints
in addition to the usual type constraints.  Unlike updatable views in mainstream
relational databases, relational lenses support complex view definitions
(including joins) and offer strong guarantees of correct round-tripping
behavior.  However, to the best of our knowledge relational lenses have never
been implemented efficiently over actual databases.

In this paper we developed the first practical implementation of relational
lenses, based on incrementalisation.  Here, we again draw on parallel
developments in the database and functional programming communities: incremental
view maintenance is a classical topic in databases, and there has been a great
deal of work in the programming language community on \emph{adaptive} or
\emph{incremental} functional programming.  We show how to embed relational
lenses (and their associated type and constraint system) into Links and prove
the correctness of incremental versions of the select, drop, and join relational
lenses and their compositions.  We also presented an implementation and
evaluated its efficiency. In particular, we showed that the naive approach of
shipping the whole source database to a client program, evaluating the put
operation in-memory, and replacing the old source tables with their new versions
is realistic only for trivial data sizes. We demonstrate scalability to
databases with hundreds of thousands of rows; for reasonable view and delta
sizes, our implementation takes milliseconds whereas the naive approach takes
seconds.

Our work establishes for the first time the feasibility of relational lenses for
solving classical view update problems in databases. Nevertheless, there may be
room for improvement in various directions. We found a pragmatic solution that
uses a small number of simple queries, but other strategies for calculating
minimal deltas are possible.  Developing additional incremental relational lens
primitives or combinators, and combining relational lenses with conventional
lenses, are two other possible future directions.

\begin{acks}
  Effort sponsored by the Air Force Office of Scientific Research, Air
  Force Material Command, USAF, under grant number
  \grantnum{AFOSR}{FA8655-13-1-3006}. The U.S. Government and University of Edinburgh
  are authorised to reproduce and distribute reprints for their
  purposes notwithstanding any copyright notation thereon. Perera was
  also supported by UK EPSRC project \grantnum{EPSRC}{EP/K034413/1}, and Horn and
  Cheney were supported by ERC Consolidator Grant Skye (grant number
  \grantnum{ERC}{682315}).  We are grateful to Jeremy Gibbons, James McKinna and the
  anonymous reviewers for comments.
\end{acks}


\begin{techreport}
\pagebreak
\appendix

\section{Basic properties of sets and relational operators}

\subsection{Well-known facts about sets}

\begin{center}
\begin{tabular}{ll}
	additive inverse
	&
	$M \setminus M = \emptyset$
	\\
	$\cdot \setminus N$ decreasing
	&
	$M \setminus N \subseteq M$
	\\
	$\cdot \setminus N$ monotone
	&
	if $M \subseteq M'$ then $M \setminus N \subseteq M' \setminus N$
	\\
	M $\setminus \cdot$ antitone
	&
	if $N \subseteq N'$ then $M \setminus N \supseteq M \setminus N'$
	\\
	$\emptyset$ unit for $\cup$
	&
	$M \cup \emptyset = M$
	\\
	$M \cup \cdot$ increasing
	&
	$M \subseteq M \cup N$
	\\
	$\cup$ commutative
	&
	$M \cup N = N \cup M$
	\\
	$\cup$ monotone
	&
	if $M \subseteq M'$ and $N \subseteq N'$ then $M \cup N \subseteq M' \cup N'$
	\\
	$\cup$ least upper bound
	&
	$M \cup N \subseteq O$ iff $M \subseteq O$ and $N \subseteq O$
	\\
	$M \cap \cdot$ decreasing
	&
	$M \cap N \subseteq N$
	\\
	$\cap$ commutative
	&
	$M \cap N = N \cap M$
	\\
	$\cap$ monotone
	&
	if $M \subseteq M'$ and $N \subseteq N'$ then $M \cap N \subseteq M' \cap N'$
	\\
	$\cap$ greatest lower bound
	&
	$M \cap N \supseteq P$ iff $M \supseteq P$ and $N \supseteq P$
	\\
	$\cap$ in terms of $\setminus$
	&
	$M \setminus (M \setminus N) = M \cap N$
	\\
	$\setminus$ and $\cap$ complementary
	&
	$M = (M \setminus N) \cup (M \cap N) = (M \cup N) \setminus (N - M)$
	\\
	$\emptyset$ least
	&
	$\emptyset \subseteq M$
	\\
	$\emptyset$ annihilator for $\cap$
	&
	$M \cap \emptyset = \emptyset$
	\\
	$\cap$ and $\cup$ induce $\subseteq$
	&
	$M \subseteq N$ iff $M \cap N = M$ iff $M \cup N = N$
	\\
	$\cup$ distributes over $\cap$
	&
	$(M \cup N) \cap P = (M \cap P) \cup (N \cap P)$
	\\
	$\cap$ distributes over $\cup$
	&
	$(M \cap N) \cup P = (M \cup P) \cap (N \cup P)$
	\\
	$\setminus$ distributes over $\cup$
	&
	$(M \cup N) \setminus O = (M \setminus O) \cup  (N \setminus O)$
\end{tabular}
\end{center}

Proving $M \subseteq \emptyset$ suffices to establish $M = \emptyset$, by
minimality of $\emptyset$ and antisymmetry.

\subsection{Other basic properties of sets}

\begin{lemma}
	$(M \setminus N) \setminus O = M \setminus (N \cup O)$.
	\label{lem:minus_minus_cup}
\end{lemma}

\begin{lemma}
	$(M \setminus M') \setminus (M \setminus N') = (M \cap N') \setminus M'$
	\label{lem:minus_minus}
\end{lemma}

\begin{lemma}
	If $M \cap N = M \cap O$, then $M \setminus N = M \setminus O$.
	\label{lem:minus_similar}
\end{lemma}

\begin{lemma}
	$(N \setminus M) \cap M = \emptyset$.
	\label{lem:setminus_cap_empty}
\end{lemma}

\begin{lemma}
	$M \subseteq N \cup O$ iff $M \setminus O \subseteq N$.
	\label{lem:set_cup_to_minus}
\end{lemma}
\begin{proof}
	\begin{salign}
		\intertext{First show that $M \setminus O \subseteq N$ implies $M \subseteq N \cup O$.}
		M \setminus O &\subseteq N
			& \\
		(M \setminus O) \cup (M \cap O) &\subseteq N \cup (M \cap O)
			& \text{$\cup$ monotone} \\
		M &\subseteq N \cup (M \cap O)
			& \text{$\setminus$ and $\cap$ complementary} \\
		M &\subseteq N \cup O
			& \text{$M \cap \cdot$ decreasing; $\cup$ monotone} \\
			\\
			\intertext{Now show that $M \subseteq N \cup O$ implies $M \setminus O \subseteq N$.}
		M &\subseteq N \cup O
			& \\
		M \setminus O &\subseteq (N \cup O) \setminus O
			& \text{$\cdot \setminus O$ monotone} \\
		M \setminus O &\subseteq (N \setminus O) \cup (O \setminus O)
			& \text{$\setminus$ distributes over $\cup$} \\
		M \setminus O &\subseteq N \setminus O 
			& \text{simpl.} \\
		M \setminus O &\subseteq N
			& \text{$\cdot \setminus O$ decreasing; trans.}
	\end{salign}
\end{proof}

\begin{lemma}
	If $N \cap O = \emptyset$ and $M \subseteq N \cup O$ then $M \setminus O = N \cap M$.
	\label{lem:subset_disjoint}
\end{lemma}

\begin{lemma}
	If $M \subseteq M'$ then $(M \cup N) \setminus M' = N \setminus M'$.
	\label{lem:absorb_setminus}
\end{lemma}
\begin{proof}
	\begin{salign}
		(M \cup N) \setminus M' &= (M \setminus M') \cup (N \setminus M')
		& \text{ $\setminus$/$\cup$ distr.} \\
		&= \emptyset \cup (N \setminus M')
		& \text{ $M \subseteq M'$ } \\
		&= N \setminus M'
	\end{salign}
\end{proof}

\begin{lemma}
	If $M \cap N = \emptyset$ then $M \setminus (M' \cup N) = M \setminus M'$.
	\label{lem:discard_disjoint}
\end{lemma}
\begin{proof}
	\begin{salign}
		M \setminus (M' \cup N) &= M \setminus (N \cup M')
			& \text{ $\cup$ comm. } \\
		&= (M \setminus N) \setminus M
			& \text{ \reflemma{minus_minus_cup} } \\
		&= M \setminus M'
			& \text{ $M \cap N = \emptyset$ }
	\end{salign}
\end{proof}

\subsection{Well-known facts about relational operators}

\begin{center}
\begin{tabular}{ll}
	$\Join$ monotone
	&
	if $M \subseteq M'$ and $N \subseteq N'$ then $M \Join N \subseteq M' \Join N'$
	\\
	$\proj{\cdot}{U}$ monotone
	&
	if $M \subseteq M'$ then $\proj{M}{U} \subseteq \proj{M'}{U}$
\end{tabular}
\end{center}

\subsection{Other basic properties of relational operators}

\begin{lemma}[$\Join$ after $\pi_U \times \pi_V$ increasing]
   Suppose $M: U \cup V$. Then $M \subseteq \proj{M}{U} \Join \proj{M}{V}$.
\end{lemma}

\begin{proof}
   \begin{salign}
      m &\in M
         & \text{suppose ($m$)} \\
      \recproj{m}{U} &\in \proj{M}{U}
         & \text{def.~$\proj{\cdot}{U}$} \\
      \recproj{m}{V} &\in \proj{M}{V}
         & \text{def.~$\proj{\cdot}{V}$} \\
      m &\in \proj{M}{U} \Join \proj{M}{V}
         & \text{def.~$\Join$}
   \end{salign}
\end{proof}

\begin{lemma}[$\pi_U$ after $\Join$ decreasing]
   Suppose $M: U$. Then $\proj{M \Join N}{U} \subseteq M$.
\end{lemma}

\begin{proof}
   \begin{salign}
      m &\in \proj{M \Join N}{U}
         & \text{suppose ($m$)} \\
      m' &\in M \Join U \wedge \recproj{m'}{U} = m
         & \text{def.~$\proj{\cdot}{U}$; exists $(m')$} \\
      m = \recproj{m'}{U} &\in M
         & \text{def.~$\Join$}
   \end{salign}
\end{proof}

\section{Proofs from \Secref{background}}
\subsection{Equivalence of project lens definitions}
\label{appendix:background}

\begin{theorem}
	The definition of the project lens given in \Secref{background:relational-lenses:project}:
	\[\begin{array}{lrcl}
	&	\getOp(M) &=& \proj{M}{U \setminus A}
		\\
	&	\putOp(M, N) &= &
	\letexpr{M_1}{N \Join \set{\{ A = a \}})}{}
	\\
	&&& \relrevise{M_1}{X \to A}{M}
	\end{array}\]
 	is equivalent to the definition given in \citet{bohannon2006relational}:
		\[\begin{array}{lrcl}
		&	\getOp(M) &=& \proj{M}{U \setminus A}
			\\
		&	\putOp(M, N) &= &
		\letexpr{N_{\opname{new}}}{N \setminus \proj{M}{U \setminus A}}{}
			\\
		&&&
		\letexpr{M_0}{(M \Join N) \cup (N_{\opname{new}} \Join \set{\{ A =
		a \}})}{}
		\\
		&&& \relrevise{M_0}{X \to A}{M}
		\end{array}\]
	\label{lem:project_equivalence}
\end{theorem}

\begin{proof}
\def\currentprefix{proof:minimality_proj_special_lem}
\begin{sflalign}
&
M: U\text{ and }N: U \setminus A\text{ with }A \in U\text{ and }X \subseteq U \setminus A
&
\text{
	suppose $(M, U, A, X)$
}
\notag
\\
&
M_0 = (M \Join N) \cup ((N \setminus \proj{M}{U \setminus A}) \Join \set{\record{A = a}})
&
\text{
	define ($M_0$)
}
\locallabel{M_0}
\\
&
M_1 = N \Join \set{\record{A = a}}
&
\text{
	define ($M_1$)
}
\locallabel{M_1}
\\
&
\text{Now show $\relrevise{M_0}{X \to A}{M} = \relrevise{M_1}{X \to A}{M}$.}
\notag
\\
\interSubcaseZero{\implies\text{ direction} \notag}
	&
	\indent
	m \in M_0
	&
	\text{
		suppose ($m$)
	}
	\locallabel{case-one-exists-m}
	\\
	\interSubcaseOne{\text{if }m \in M \Join N \notag}
		&
		\indent\indent
		m \in M
		&
		\text{
			def.~$\Join$
		}
		\notag
		\\
		&
		\indent\indent
		\recproj{m}{U \setminus A} \in N
		&
		\text{
			def.~$\Join$
		}
		\notag
		\\
		&
		\indent\indent
		m' = \recupdate{\recproj{m}{U \setminus A}}{\record{A = a}}
		&
		\text{
			define ($m'$)
		}
		\locallabel{m-prime}
		\\
		&
		\indent\indent
		\recproj{m'}{U \setminus A} = \recproj{m}{U \setminus A}
		&
		\locallabel{m-prime-eq-m}
		\\
		&
		\indent\qedLocal
		m' \in M_1
		&
		\text{
			(\localref{M_0}); def.~$\Join$
		}
		\notag
		\\
		&
		\indent\qedLocal
		\newrecrevise{m'}{X \to A}{M} = \recupdate{m'}{\recproj{m}{A}}
		&
		\text{def.~$\newrecrevise{\cdot}{X \to A}{\cdot}$; $m \in M$}
		\notag
		\\
		&
		\indent\indent
		\phantom{\newrecrevise{m'}{X \to A}{M}} = \recupdate{m}{\recproj{m}{A}}
		&
		\text{(\localref{m-prime-eq-m})}
		\notag
		\\
		&
		\indent\indent
		\phantom{\newrecrevise{m'}{X \to A}{M}} = \newrecrevise{m}{X \to A}{M}
		&
		\text{
			$m \in M$; def.~$\newrecrevise{\cdot}{X \to A}{\cdot}$
		}
		\notag
		\\
	\interSubcaseOne{\text{if }m \in (N \setminus \proj{M}{U \setminus A}) \Join \set{\record{A = a}} \notag}
		&
		\indent\qedLocal
		m \in M_1
		&
		\text{
			$\Join$ monotone; (\localref{case-one-exists-m})
		}
		\notag
		\\
\interSubcaseZero{\impliedby\text{ direction} \notag}
	&
	\indent
	m \in M_1
	&
	\text{
		suppose ($m$)
	}
	\notag
	\\
	&
	\indent
	\recproj{m}{U \setminus A} \in N
	&
	\notag
	\\
	\interSubcaseOne{
		\text{if }\exists m' \in M.\recproj{m'}{U \setminus A} = \recproj{m}{U \setminus A}
		\locallabel{case-two-m-prime-eq-m}
	}
		&
		\indent\indent
		m' \in M \Join N
		&
		\notag
		\\
		&
		\indent\qedLocal
		m' \in M_0
		&
		\text{
			(\localref{M_0})
		}
		\notag
		\\
		&
		\indent\qedLocal
		\newrecrevise{m'}{X \to A}{M} = \recupdate{m'}{\recproj{m'}{A}}
		&
		\text{
			def.~$\newrecrevise{\cdot}{X \to A}{\cdot}$; $m' \in M$
		}
		\notag
		\\
		&
		\indent\indent
		\phantom{\newrecrevise{m'}{X \to A}{M}} = \recupdate{m}{\recproj{m'}{A}}
		&
		\text{
			(\localref{case-two-m-prime-eq-m})
		}
		\notag
		\\
		&
		\indent\indent
		\phantom{\newrecrevise{m'}{X \to A}{M}} = \newrecrevise{m}{X \to A}{M}
		&
		\text{
			$m' \in M$; def.~$\newrecrevise{\cdot}{X \to A}{\cdot}$
		}
		\notag
		\\
	\interSubcaseOne{\text{if }\nexists m' \in M.\recproj{m'}{U \setminus A} = \recproj{m}{U \setminus A} \notag}
		&
		\indent\indent
		\recproj{m}{U \setminus A} \in N \setminus \proj{M}{U \setminus A}
		&
		\notag
		\\
		&
		\indent\indent
		m \in (N \setminus \proj{M}{U \setminus A}) \Join \set{\record{A = a}}
		&
		\notag
		\\
		&
		\indent\qedLocal
		m \in M_0
		&
		\text{
			(\localref{M_0})
		}
		\notag
\end{sflalign}
\end{proof}

\section{Proofs from \Secref{framework}}
\subsection*{Proof of Lemma~\ref{lem:expand_minimal}}

\restatelemexpandminimal*

\begin{proof}
	\def\currentprefix{proof:expandsetdeltaplus}
	\begin{sflalign}
		&
		\Delta M \text{ minimal for }M
		&
		\text{
			suppose ($M$, $\Delta M$)
		}
		\notag
		\\
		&
		\negative{\Delta M} \subseteq M
		&
		\notag
		\\[2mm]
		&
		M \deltaplus \Delta M
		=
		(M, \emptyset) \deltaplus (\positive{\Delta M}, \negative{\Delta M})
		&
		\text{coerce}
		\notag
		\\
		&
		\phantom{M \deltaplus \Delta M}
		=
		( (M \setminus \negative{\Delta M}) \cup (\positive{\Delta M} \setminus \emptyset),
		(\emptyset \setminus \positive{\Delta M}) \cup (\negative{\Delta M} \setminus M) )
		&
		\text{
			def.~$\deltaplus$
		}
		\notag
		\\
		&
		\phantom{M \deltaplus \Delta M}
		=
		( (M \setminus \negative{\Delta M}) \cup \positive{\Delta M},
		 \negative{\Delta M} \setminus M )
		&
		\text{
			simpl.~$\emptyset$
		}
		\notag
		\\
		&
		\phantom{M \deltaplus \Delta M}
		=
		( (M \setminus \negative{\Delta M}) \cup \positive{\Delta M}, \emptyset )
		&
		\text{$\nDelta{M} \subseteq M$}
		\notag
		\\
		&
		\phantom{M \deltaplus \Delta M}
		=
		(M \setminus \negative{\Delta M}) \cup \positive{\Delta M}
		&
		\text{
			uncoerce
		}
		\notag
      \\
		&
		\phantom{M \deltaplus \Delta M}
		=
		(M \cup \positive{\Delta M}) \setminus \negative{\Delta M}
		&
		\text{
			$\positive{\Delta M} \cap \negative{\Delta M} = \emptyset$
		}
		\notag
	\end{sflalign}
\end{proof}

\subsection*{Proof of Lemma~\ref{lem:diff}}

\restatelemdiff*

\begin{proof}
\begin{salign}
M \deltaminus N &= (M,\emptyset) \deltaminus
  (N,\emptyset)
  & \text{coerce} \\
&= (M,\emptyset) \deltaplus
  (\emptyset,N)
  & \text{def.~$\deltaminus$} \\
&= ((M-N)\cup(\emptyset-\emptyset), (\emptyset-\emptyset)\cup(N-M))
  & \text{def.~$\deltaplus$} \\
&= (M-N, N-M)
\end{salign}

Moreover:
	\begin{sflalign}
		& M \setminus N \subseteq M
		& \text{ $\cdot \setminus N$ decreasing }
      \notag
      \\
		& (N \setminus M) \cap M = \emptyset
      & \text{ \reflemma{setminus_cap_empty} }
      \notag
      \\
		& (M \setminus N, N \setminus M) \text{ minimal for } M
		& \text{def.~minimal}
      \notag
	\end{sflalign}
\end{proof}

\subsection*{Proof of Lemma~\ref{lem:delta-basics-positive}}
\restatelemdeltabasicspositive*

\begin{proof}
	\begin{salign}
		(M \deltaapp \Delta M) \setminus M
		&= ((M \setminus \nDelta M) \cup \pDelta M) \setminus M
		& \text{\reflemma{expand_minimal}}
		\\
		&= ((M \setminus \nDelta M) \setminus M) \cup (\pDelta M \setminus M)
		& \text{$\setminus$/$\cup$ distr.}
		\\
		&= \pDelta M
		& \text{simpl.; $\pDelta M \cap M = \emptyset$}
	\end{salign}
\end{proof}

\subsection*{Proof of Lemma~\ref{lem:delta-basics-negative}}
\restatelemdeltabasicsnegative*

\begin{proof}
	\begin{salign}
		(M \cap (M \deltaapp \Delta M)) \deltaminus M
		&= (M \cap ((M \setminus \nDelta M) \cup \pDelta M)) \deltaminus M
		& \text{\reflemma{expand_minimal}}
		\\
		&= (M \cap (M \setminus \nDelta M)) \deltaminus M
		& \text{$M \cap \pDelta M = \emptyset$}
		\\
		&= (M \setminus \nDelta M) \deltaminus M
		& \text{$\nDelta M \subseteq M$}
		\\
		&= ((M \setminus \nDelta M) \setminus M, M \setminus (M \setminus \nDelta M))
		& \text{\reflemma{diff}}
		\\
		&= (\emptyset, \nDelta M)
		& \text{simpl.; $\nDelta M \subseteq M$}
		\\
		&= \deltaminus(\nDelta M, \emptyset)
		& \text{def.~$\deltaminus$}
		\\
		&= \deltaminus\nDelta M
		& \text{uncoerce}
	\end{salign}
In particular this implies that $\nDelta{M} = M \setminus ((M \cap (M \deltaapp
\Delta M))) = (M \setminus M) \cup (M \setminus (M \deltaapp \Delta
M)) =M \setminus (M \deltaapp \Delta M)$.
\end{proof}

\subsection*{Proof of Lemma~\ref{lem:correctness-delta-ops}}
\restatelemcorrectnessdeltaops*

To simplify notation, we abbreviate $M' = M \deltaplus \Delta M$ and $N' = N
\deltaplus \Delta N$. Recall that the positive and negative parts of a delta are
always disjoint; we freely use the fact that $(X \cup Y) - Z = (X - Z) \cup Y$
when $Y$ and $Z$ are disjoint.

The proof of parts (1)-(4) all follow a similar pattern: we first find
expressions $\pDelta{N},\nDelta{N}$ for the components of a minimal relational
delta such that $\opname{op}(M \deltaapp \Delta M) = \opname{op}(M) \deltaapp
\Delta N$.  By Lemma~\ref{lem:delta-correct-unique} it then follows that
$\deltaOp{op}(M,\Delta M) = \opname{op}(M \deltaplus \Delta M) \deltaminus
\opname{op}(M) = (\opname{op}(M) \deltaapp \Delta N) \deltaminus  \opname{op}(M)
=\Delta N$.

\subsubsection*{Proof of part (1)}

\begin{salign}
    \select{P}{M'}
&= \select{P}{(M \cup \pDelta{M}) \setminus \nDelta{M}}
\\
&= \select{P}{M \cup \pDelta{M}} \setminus \select{P}{\nDelta{M}}
\\
&= (\select{P}{M}   \cup \select{P}{\pDelta{M}})\setminus \select{P}{\nDelta{M}}
\\
&= \select{P}{M} \deltaplus (\select{P}{\pDelta{M}},\select{P}{\nDelta{M}})
\end{salign}

\noindent where clearly $\select{P}{\pDelta{M}} \cap \select{P}{\nDelta{M}} =
\emptyset$.  Moreover, $\select{P}{\pDelta{M}} \cap \select{P}{M} =
\emptyset$ by minimality of $\Delta M$, and likewise
$\select{P}{\nDelta{M}} \subseteq \select{P}{M}$ by minimality and
monotonicity of selection.

\subsubsection*{Proof of part (2)}
First we observe that
\begin{salign}
\proj{M\cup N}{U} &= \proj{M}{U} \cup \proj{N}{U}
\\
\proj{M-N}{U} &= \proj{M}{U} - (\proj{N}{U} - \proj{M-N}{U})
\end{salign}
Now we proceed as follows:
\begin{salign}
     \proj{M'}{U}
&=
   \proj{(M \cup \pDelta{M}) \setminus \nDelta{M}}{U}
\\
   &=
   (\proj{M \cup \pDelta{M}}{U})
   \setminus
   (\proj{\nDelta{M}}{U} \setminus \proj{M'}{U})
  \\
   &=
   (\proj{\pDelta{M}}{U} \cup \proj{M}{U})
   \setminus
   (\proj{\nDelta{M}}{U} \setminus \proj{M'}{U})
   \\
   &=
   ((\proj{\pDelta{M}}{U} \cup \proj{M}{U})\;\setminus\;(\proj{M}{U} \setminus \proj{M}{U}))
   \setminus
   (\proj{\nDelta{M}}{U} \setminus \proj{M'}{U})
   \\
   &=
   (\proj{M}{U} \cup (\proj{\pDelta{M}}{U} \setminus \proj{M}{U}))
   \setminus
   (\proj{\nDelta{M}}{U} \setminus \proj{M'}{U})
\\
   &=
   \proj{M}{U} \deltaapp
   (\proj{\pDelta{M}}{U} \setminus \proj{M}{U},
    \proj{\nDelta{M}}{U} \setminus \proj{M'}{U})
\end{salign}
\noindent where line 4 follows from the identity $X \cup (Y \setminus Z) = (X
\cup Y) \setminus (Z \setminus X)$, line 2 from the first observation above and line
1 by the second observation.

To establish minimality, clearly $(\proj{\pDelta{M}}{U} \setminus \proj{M}{U})
\cap \proj{M}{U} = \emptyset$, while $\proj{\nDelta{M}}{U} \setminus
\proj{M'}{U} \subseteq \proj{M}{U}$ by monotonicity of projection
since $\nDelta{M} \subseteq M$.

\subsubsection*{Proof of part (3)}
We need to prove:
\[M ' \Join N' = (M \Join N   \cup
    (M' \Join \pDelta{N} \cup
    \pDelta{M} \Join N')) \setminus (M \Join \nDelta{N} \cup \nDelta{M} \Join N)\]

We first consider the special cases where $\nDelta{M} = \nDelta{N} =
\emptyset$ and $\pDelta{M} = \pDelta{N} = \emptyset$.
In the first case we have:
\begin{salign}
 (M \cup \pDelta{M}) \Join (N \cup \pDelta{N})
&= M \Join N \cup (\pDelta{M} \Join N \cup M \Join \pDelta{N} \cup \pDelta{M}
    \Join \pDelta{N})\\
&= M \Join N \cup (\pDelta{M} \Join N \cup \pDelta{M}
    \Join \pDelta{N} \cup M \Join \pDelta{N} \cup \pDelta{M}
    \Join \pDelta{N})\\
&= M \Join N \cup ((M \cup \pDelta{M}) \Join \pDelta{N} \cup
    \pDelta{M} \Join (N \cup \pDelta{N}))
\end{salign}
In the second case we have:
\begin{salign}
  (M \setminus \nDelta{M}) \Join (N \setminus \nDelta{N})
&=
M \Join  (N \setminus \nDelta{N}) \setminus \nDelta{M} \Join  (N
    \setminus \nDelta{N})
\\
&=
M \Join  N \setminus M \Join \nDelta{N} \setminus \nDelta{M} \Join  (N
    \setminus \nDelta{N})
\\
&=
M \Join  N \setminus (M \Join \nDelta{N} \cup \nDelta{M} \Join  (N
    \setminus \nDelta{N}))
\\
&=
M \Join  N \setminus (M \Join \nDelta{N} \cup (\nDelta{M} \Join  N
    \setminus \nDelta{M} \Join \nDelta{N}))
\\
&=
M \Join  N \setminus ((M \Join \nDelta{N} \cup \nDelta{M} \Join  N)
    \setminus (\nDelta{M} \Join \nDelta{N} \setminus M \Join \nDelta{N}))
\\
&=
M \Join  N \setminus ((M \Join \nDelta{N} \cup \nDelta{M} \Join  N)
    \setminus (\nDelta{M} \setminus M) \Join \nDelta{N})
\\
&=
M \Join  N \setminus (M \Join \nDelta{N} \cup \nDelta{M} \Join  N)
    \cup (M \Join N \cap (\nDelta{M} \setminus M) \Join \nDelta{N})
\\
&=
M \Join  N \setminus (M \Join \nDelta{N} \cup \nDelta{M} \Join  N)
\end{salign}
where the final step follows because $M $ and $\nDelta{M} \setminus M$
are disjoint, and so $M \Join N \cap (\nDelta{M} \setminus M) \Join
\nDelta{N} = \emptyset$.

We now proceed as follows:
\begin{salign}
  M' \Join N'
&= ((M \setminus \nDelta{M}) \cup \pDelta{M}) \Join ((N
                  \setminus \nDelta{N}) \cup \pDelta{N})
\\
&= (M \setminus \nDelta{M}) \Join (N \setminus \nDelta{N})  \\
&{}\cup
    (((M\setminus \nDelta{M}) \cup \pDelta{M}) \Join \pDelta{N} \cup
    \pDelta{M} \Join ((N \setminus \nDelta{N}) \cup \pDelta{N}))
\\
&= (M \setminus \nDelta{M}) \Join (N \setminus \nDelta{N})  \cup
    (M' \Join \pDelta{N} \cup
    \pDelta{M} \Join N')
\\
&= (M \Join N  \setminus (M \Join \nDelta{N} \cup \nDelta{M} \Join N))  \cup
    (M' \Join \pDelta{N} \cup
    \pDelta{M} \Join N')
\\
&= (M \Join N   \cup
    (M' \Join \pDelta{N} \cup
    \pDelta{M} \Join N')) \setminus (M \Join \nDelta{N} \cup \nDelta{M} \Join N)
\end{salign}
where in the last step we use the disjointness of $\nDelta{N}$ with
$N'$ and $\pDelta{N}$, and likewise of $\nDelta{M}$ with $M'$ and
$\pDelta{M}$,  and the fact that $X,Y$ disjoint implies $X \Join Z$
and $Y \Join W$ disjoint.  Clearly, the positive and negative deltas
resulting in this case are disjoint.  Minimality follows since $M'
\Join \pDelta{N} \cap M \Join N = \emptyset = \pDelta{M} \Join N'
\cap M \Join N$
by minimality of $\Delta M, \Delta N$ and the fact that $X \cap Y =
\emptyset$ implies $X \Join Z \cap Y \Join Z = \emptyset$, and
likewise $M \Join \nDelta{N} \subseteq M \Join N \supseteq
\nDelta{M} \Join N$ by minimality and monotonicity of $\Join$.

\subsubsection*{Proof of part (4)}

\begin{salign}
  \rename{A/B}{M'} & = \rename{A/B}{(M \cup
                                            \pDelta{M}) \setminus
                                            \nDelta{ M}}
\\
 & = (\rename{A/B}{M} \cup
                                            \rename{A/B}{\pDelta{M}}) \setminus
                                            \rename{A/B}{\nDelta{ M}}
\\
&=     \rename{A/B}{M} \deltaplus(\rename{A/B}{\Delta M^+},
    \rename{A/B}{\Delta M^-})
\end{salign}
Clearly, $\rename{A/B}{\pDelta{M}}$ and $\rename{A/B}{\nDelta{ M}}$
are disjoint, and minimal since $\rename{A/B}{M} \cap
\rename{A/B}{\pDelta{M}} = \emptyset$ and $ \rename{A/B}{\nDelta{ M}}
\subseteq \rename{A/B}{M}$.

\subsubsection*{Proof of part (5)}

\begin{lemma}
	If $N \cap O = \emptyset$, $N' \subseteq N$ and $M \subseteq N$, then $M
	\subseteq N' \cup O$ implies $M \subseteq N'$. \label{lem:disjoint_subs}
\end{lemma}
\begin{proof}
	\begin{salign}
		M &\subseteq N' \cup O
			& \\
		M \cap N &\subseteq (N' \cup O) \cap N
			& \text{$\cap$ monotone} \\
		&\subseteq (N' \cap N) \cup (O \cap N)
			& \text{$\cup$ distributes over $\cap$} \\
		&\subseteq N' \cap N
			& \text{$N \cap O = \emptyset$} \\
		&\subseteq N'
			& \text{$N' \subseteq N$; $\cap$ induces $\subseteq$} \\
		M &\subseteq N'
			& \text{$M \subseteq N$; $\cap$ induces $\subseteq$}
	\end{salign}
\end{proof}

We want to show:
\begin{quote}
	If $N \subseteq M$ and $N' \subseteq M'$ then $(M, \Delta M) \deltasetminus (N, \Delta N) = \Delta M \deltaminus \Delta N$.
\end{quote}
\begin{proof}
	\def\currentprefix{proof:deltaminus_opt}
	\begin{sflalign}
		& \minimal{\Delta N}{N}					& \text{ suppose ($\Delta N, N$) } \notag
		\\ & \negDN \subseteq N					& \text{ def.~minimal }
			\locallabel{negnsubset}
		\\ & \posDN \cap N = \emptyset			& \text{ def.~minimal }
			\locallabel{disjointnplusn}
		\\[2mm] & \minimal{\Delta M}{M}		& \text{ suppose ($\Delta M, M$) } \notag
		\\ & \negDM \subseteq M					& \text{ def.~minimal }
			\locallabel{negmsubset}
		\\ & \posDM \cap M = \emptyset			& \text{ def.~minimal }
			\locallabel{disjointmplusm}
		\\[2mm] & N \subseteq M					& \text{ suppose } \locallabel{nsubm}
		\\ & \posDM \cap N = \emptyset
			& \text{ $\cap$ monotone (\localref{disjointmplusm}) } \locallabel{disjointmplusn}
		\\[2mm]
		& \Delta O = (M, \Delta M) \deltasetminus (N, \Delta N)
			& \text{ suppose ($\Delta O$) } \notag
		\\
		&
		\phantom{\Delta O}
		=
		(M' \setminus N') \deltaminus (M \setminus N)
		&
  		\text{delta-correctness}
		\notag
		\\
		&
		\phantom{\Delta O}
		=
		((M' \setminus N') \setminus (M \setminus N), (M \setminus N) \setminus (M' \setminus N'))
		&
		\text{
			\reflemma{diff}
		}
		\locallabel{defdo}
		\\[2mm] & N' \subseteq M'		& \text{suppose} \notag
		\\ & \uline{(N \setminus \negative{\Delta N})} \cup \positive{\Delta N} \subseteq (M \setminus \negative{\Delta M}) \cup \positive{\Delta M}
			& \text{ \reflemma{expand_minimal} } \notag
		\\ & \uline{(N \setminus \negative{\Delta N})} \subseteq \uline{(\uline{M} \setminus \negative{\Delta M})} \cup \uline{\positive{\Delta M}}
			& \text{ $\cdot \cup \positive{\Delta N}$ incr.; trans.} \locallabel{expupdsubset}
		\\ & \uline{N} \setminus \uline{\negative{\Delta N}} \subseteq \uline{M \setminus \negative{\Delta M}}
			& \text{\reflemma{disjoint_subs} (\localref{nsubm}, \localref{disjointmplusm})} \notag
		\\ & \uline{N} \subseteq \uline{(M \setminus \negative{\Delta M}) \cup \negative{\Delta N}}
			& \text{\reflemma{set_cup_to_minus}} \notag
		\\ & N \cap \negative{\Delta M} \subseteq ( \uline{(M \setminus \negative{\Delta M})} \cup \uline{\negative{\Delta N}}) \cap \uline{\negative{\Delta M}}
			& \text{$\cap$ monotone} \notag
		\\ & N \cap \negative{\Delta M} \subseteq ( (\uline{M} \setminus \uline{\negative{\Delta M}}) \cap \uline{\negative{\Delta M}}) \cup (\negative{\Delta N} \cap \negative{\Delta M})
			& \text{distr.} \notag
		\\ & N \cap \negative{\Delta M} \subseteq \negative{\Delta N} \cap \negative{\Delta M}
			& \text{\reflemma{setminus_cap_empty}; simpl.~$\emptyset$} \locallabel{ncapnegdm_subs}
		\\[2mm] & \negDN \cap \negDM \subseteq N \cap \negDM			& \text{$\cap$ monotone (\localref{negnsubset})} \notag
		\\ & \negDN \cap \negDM = N \cap \negDM				& \text{(\localref{ncapnegdm_subs}); antisym.} \locallabel{negdncapnegdm_equiv}
		\\[2mm] & \posDN \subseteq (M - \negDM) \cup \posDM				& \text{ (\localref{expupdsubset}); $\cup$ LUB; trans. } \locallabel{posdnsubs}
		\\[2mm] & N \setminus \negDN \subseteq M						& \text{ $\cdot \setminus \negDN$ decr.; (\localref{nsubm}); trans. } \notag
		\\ & M \cap (N \setminus \negDN) = N \setminus \negDN	& \text{ $\cap$ induces $\subseteq$ } \locallabel{mcapnetc}
		\\[2mm] & M \cap \negDM = \negDM								& \text{ $\cap$ induces $\subseteq$ (\localref{negmsubset}) } \locallabel{negdmcapm}
		\\[2mm] & \positive{\Delta O} =
			(\uline{M'} \setminus \uline{N'}) \setminus \uline{(M \setminus N)}
			& \text{(\localref{defdo})}\notag
		\\ & \phantom{\positive{\Delta O}} =
			\uline{M'} \setminus ( \uline{N'} \cup (M \setminus N))
			& \text{ \reflemma{minus_minus_cup} } \notag
		\\ & \phantom{\positive{\Delta O}} =
			(\uline{(M \setminus \negative{\Delta M})} \cup \uline{\positive{\Delta M}}) \setminus
			\uline{( (N \setminus \negative{\Delta N}) \cup \positive{\Delta N} \cup (M \setminus N) )}
			& \text{ \reflemma{expand_minimal}	} \notag
		\\ & \phantom{\positive{\Delta O}} =
			((M \setminus \negative{\Delta M}) \setminus ((N \setminus \negative{\Delta N}) \cup \posDN \cup (M \setminus N) ))\;\cup
			& \notag
		\\ & \phantom{\positive{\Delta O} =\;}
		   (\positive{\Delta M} \setminus (\uline{(N \setminus \negative{\Delta N})} \cup \posDN \cup \uline{(M \setminus N)} ))
			& \text{ $\setminus$/$\cup$ distr.} \notag
		\\ & \phantom{\positive{\Delta O}} =
			((\uline{M \setminus \negative{\Delta M}}) \setminus (\uline{(N \setminus \negative{\Delta N}) \cup \posDN} \cup \uline{(M \setminus N)} )) \cup
			(\positive{\Delta M} \setminus \posDN)
			& \text{ (\localref{disjointmplusn}, \localref{disjointmplusm}) } \notag
		\\ & \phantom{\positive{\Delta O}} =
			(((\uline{M} \setminus \uline{\negative{\Delta M}}) \setminus (M \setminus \uline{N})) \setminus ((N \setminus \negative{\Delta N}) \cup \posDN))\;\cup
			\notag
		\\ & \phantom{\positive{\Delta O} =\;}
			(\positive{\Delta M} \setminus \posDN)
			& \text{ $\cup$ comm.; \reflemma{minus_minus_cup} } \notag
		\\ & \phantom{\positive{\Delta O}} =
			(\uline{(N \setminus \negative{\Delta M})} \setminus (\uline{(N \setminus \negative{\Delta N})} \cup \uline{\posDN})) \cup
			(\positive{\Delta M} \setminus \posDN)
			& \text{ \reflemma{minus_minus}; $N \subseteq M$ } \notag
		\\ & \phantom{\positive{\Delta O}} =
			(((\uline{N} \setminus \uline{\negative{\Delta M}}) \setminus (N \setminus \uline{\negative{\Delta N}})) \setminus \posDN) \cup
			(\positive{\Delta M} \setminus \posDN)
			& \text{ \reflemma{minus_minus_cup} } \notag
		\\ & \phantom{\positive{\Delta O}} =
			(\uline{(\negative{\Delta N} \setminus \negative{\Delta M})} \setminus \uline{\posDN}) \cup (\positive{\Delta M} \setminus \posDN)
			& \text{ \reflemma{minus_minus}; $\negative{\Delta N} \subseteq N$ } \notag
		\\ & \phantom{\positive{\Delta O}} =
			(\negative{\Delta N} \setminus \negative{\Delta M}) \cup (\positive{\Delta M} \setminus \posDN)
			& \text{ $\nDelta N \cap \posDN = \emptyset$ } \locallabel{defdoplus}
		\\[2mm]	& \negative{\Delta O} =
			(M \setminus N) \setminus (\uline{M'} \setminus N')
		& \text{(\localref{defdo})} \notag
		\\ & \phantom{\negative{\Delta O}} =
 			(M \setminus N) \setminus ( ( \uline{(M \setminus \negative{\Delta M})} \cup \uline{\positive{\Delta M}}) \setminus \uline{N'} )
			& \text{\reflemma{expand_minimal}} \notag
		\\ & \phantom{\negative{\Delta O}} =
			(M \setminus N)\;\setminus (( (\uline{M} \setminus \uline{\negative{\Delta M}}) \setminus \uline{( N')} )
				\cup (\positive{\Delta M} \setminus N' ))
			& \text{$\setminus$/$\cup$ distr.} \notag
		\\ & \phantom{\negative{\Delta O}} =
			(M \setminus N) \setminus ((M \setminus (N' \cup \negative{\Delta M}))
				\cup (\positive{\Delta M} \setminus ( \uline{N} \deltaplus \uline{\Delta N})))
			& \text{\reflemma{minus_minus_cup}; $\cup$ comm.}\notag
		\\ & \phantom{\negative{\Delta O}} =
			(M \setminus N)\;\setminus \notag
		\\ & \phantom{\negative{\Delta O} =}
			\;\;((M \setminus (N' \cup \negative{\Delta M})) \cup
			(\uline{\positive{\Delta M}} \setminus ( \uline{(N \setminus \negative{\Delta N})} \cup \uline{\positive{\Delta N}})))
			& \text{ \reflemma{expand_minimal} }\notag
		\\ & \phantom{\negative{\Delta O}} =
			(\uline{M} \setminus \uline{N}) \setminus \uline{( ( M \setminus ( N' \cup \negative{\Delta M}))
			\cup (\positive{\Delta M} \setminus \positive{\Delta N}) )}
			& \text{(\localref{disjointmplusn}); \reflemma{discard_disjoint}}  \notag
		\\ & \phantom{\negative{\Delta O}} =
			\uline{M} \setminus ( \uline{( M \setminus (N' \cup \negative{\Delta M}))}
			\cup \uline{(\positive{\Delta M} \setminus \positive{\Delta N}) \cup N} )
			& \text{\reflemma{minus_minus_cup}} \notag
		\\ & \phantom{\negative{\Delta O}} =
			(\uline{M} \setminus ( M \setminus \uline{(N' \cup \negative{\Delta M})}))
			\setminus \uline{((\positive{\Delta M} \setminus \positive{\Delta N}) \cup N )}
			& \text{\reflemma{minus_minus_cup}} \notag
		\\ & \phantom{\negative{\Delta O}} =
			(\uline{M} \cap (N' \cup \uline{\negative{\Delta M}})) \setminus (
			(\positive{\Delta M} \setminus \positive{\Delta N}) \cup N )
			& \text{ $\cap$ in terms of $\setminus$ } \notag
		\\ & \phantom{\negative{\Delta O}} =
			(\uline{M} \cap (\uline{(N \setminus \negative{\Delta N})} \cup \uline{\positive{\Delta N}} \cup \uline{\negative{\Delta M}})) \setminus (
			(\positive{\Delta M} \setminus \positive{\Delta N}) \cup N )
			& \text{ \reflemma{expand_minimal}} \notag
		\\ & \phantom{\negative{\Delta O}} =
			(\uline{(N \setminus \negative{\Delta N})} \cup \uline{(M \cap \positive{\Delta N}) \cup \negative{\Delta M}}) \setminus
			\uline{((\positive{\Delta M} \setminus \positive{\Delta N}) \cup N )}
			& \text{distr.; (\localref{mcapnetc}, \localref{negdmcapm})} \notag
		\\ & \phantom{\negative{\Delta O}} =
			(\uline{(M \cap \positive{\Delta N})} \cup \uline{\negative{\Delta M}}) \setminus
			\uline{((\positive{\Delta M} \setminus \positive{\Delta N}) \cup N )}
			& \text{ \reflemma{absorb_setminus} } \notag
		\\ & \phantom{\negative{\Delta O}} =
			( \uline{(M \cap \posDN)} \setminus (\uline{(\posDM \setminus \posDN)} \cup \uline{N}) )\;\cup \notag
		\\ & \phantom{\negative{\Delta O} =}
			\;(\uline{\negative{\Delta M}} \setminus (\uline{(\posDM \setminus \posDN)} \cup \uline{N}))
		& \text{ $\setminus$/$\cup$ distr. } \notag
		\\ & \phantom{\negative{\Delta O}} =
		( \uline{M} \cap \uline{\posDN}) \cup (\negative{\Delta M} \setminus N)
			& \text{ (\localref{disjointnplusn}, \localref{disjointmplusm}); \reflemma{discard_disjoint} } \notag
		\\ & \phantom{\negative{\Delta O}} =
			(\posDN \setminus \posDM) \cup \uline{(\negative{\Delta M} \setminus N)}
			& \text{ \reflemma{subset_disjoint}; (\localref{posdnsubs})} \notag
		\\ & \phantom{\negative{\Delta O}} =
			(\posDN \setminus \posDM) \cup (\negDM \setminus \negDN)
			& \text{\reflemma{minus_similar} (\localref{negdncapnegdm_equiv})}
			\locallabel{defdominus}
		\\[2mm] & \Delta O = (\posDM, \negDM) \deltaplus (\negDN, \posDN)		& \text{ def.~$\cdot \deltaplus \cdot$ (\localref{defdoplus}, \localref{defdominus}) } \notag
		\\ & \phantom{\Delta O} =\Delta M \deltaminus \Delta N							& \text{ def.~$\cdot \deltaminus \cdot$ } \notag
	\end{sflalign}

\end{proof}

\subsection*{Proof of Theorem~\ref{thm:compositional}}

\restatethmcompositional*
\begin{proof}
  By induction on the structure of $q$.  First observe that in any case
  $\Dagger{q}$ is delta-correct with respect to $q$ if and only if
  $\Derive{q}$ is. Thus, we show that $\Derive{q}$ is delta-correct by
  induction, and the reasoning for $\Dagger{q}$ is similar.
  \begin{itemize}
  \item If $q = M$, a constant relation, then $\Derive{q} = \emptyset$
    which is delta-correct with respect to $q$
    since $q(R_1\deltaapp \Delta R_1,\ldots,R_n \deltaplus \Delta R_n) = M
    = q(R_1,\ldots,R_n) \deltaapp \emptyset$.  Minimality is obviously preserved.
\item If $q = R_i$, a relation reference, then $\Derive{R_i} = \Delta R_i$ is delta-correct with respect to $q$ since
  $q(R_1\deltaapp \Delta R_1,\ldots,R_n \deltaplus \Delta R_n) = R_i
  \deltaapp \Delta R_i
    = q(R_1,\ldots,R_n) \deltaapp \Delta R_i$.  Minimality is
    obviously preserved.
\item If $q = \opname{op}(q_1,\ldots,q_n)$ then the desired result
  follows from the definition of $\deltaOp{op}$ (which are
  delta-correct by construction), the induction hypothesis applied to
  the subexpressions $q_i$, and finally
  Lemma~\ref{lem:delta-correct-composition}.
\item If $q = \letexpr{R}{q_1}{q_2}$, then the translation is
\[\letexpr{(R,\Delta
                                   R)}{\Dagger{q_1}}{\Derive{q_2(R)}}\]
where $q_2$ has an additional parameter $R$, so $\Dagger{q_2}$
will have an additional pair of parameters $(R,\Delta R)$.  By induction both
$\Dagger{q_1}$ and
$\Dagger{q_2}(R,\Delta R) $ are delta-correct with respect to $q_1$ and $q_2$
respectively.  We reason as follows:
\begin{eqnarray*}
&&\letexpr{R'}{q_1(R_1 \deltaapp \Delta
  R_1,\ldots,R_n\deltaplus \Delta R_n)}{q_2(R',R_1 \deltaapp \Delta
  R_1,\ldots,R_n \deltaapp \Delta R_n)}
\\
&=&
\letexpr{R'}{q_1(R_1,\ldots,R_n)\deltaplus \Derive{q}((R_1,\Delta
    R_1),\ldots,(R_n,\Delta R_n))}{q_2(R',R_1 \deltaapp \Delta
  R_1,\ldots,R_n \deltaapp \Delta R_n)}
\\
&=&
\letexpr{(R',\Delta R')}{\Dagger{q_1}((R_1,\Delta
    R_1),\ldots,(R_n,\Delta R_n))}{q_2(R' \deltaapp\Delta R',R_1 \deltaapp \Delta
  R_1,\ldots,R_n \deltaapp \Delta R_n)}
\\
&=&
\letexpr{(R',\Delta R')}{\Dagger{q}_1((R_1,\Delta
    R_1),\ldots,(R_n,\Delta R_n))}{}
\\
&&q_2(R',R_1,\ldots,R_n) \deltaapp \Derive{q_2}((R' ,\Delta R'),(R_1 , \Delta
  R_1),\ldots,(R_n , \Delta R_n))
\\
&=&
(\letexpr{(R',\Delta R')}{\Dagger{q_1}((R_1,\Delta
    R_1),\ldots,(R_n,\Delta R_n))}{q_2(R',R_1,\ldots,R_n)})
\\
&&{}\deltaapp (\letexpr{(R',\Delta R')}{\Dagger{q_1}((R_1,\Delta
    R_1),\ldots,(R_n,\Delta R_n))}{ \Derive{q_2}((R' ,\Delta R'),(R_1 , \Delta
  R_1),\ldots,(R_n , \Delta R_n))})
\\
&=&
(\letexpr{R'}{q_1(R_1,\ldots,R_n)}{q_2(R',R_1,\ldots,R_n)}) \deltaapp \Derive{\letexpr{R'}{q_1(R_1,\ldots,R_n)}{q_2(R',R_1,\ldots,R_n)}}
\end{eqnarray*}
Moreover, minimality is preserved by $\Derive{q_1}$, so $\Delta R'$ is
minimal, which together with the minimality of other deltas implies
that the delta computed by $\Derive{\letexpr{R'}{q_1}{q_2(R')}}$  is
  also minimal.  This shows that $\Derive{\letexpr{R'}{q_1}{q_2(R')}}$ is
  delta-correct.
  \end{itemize}
\end{proof}

\subsection*{Proof of Lemma~\ref{lem:delta_rel_rev}}

To prove this property, we first observe that $\relrevise{M}{F}{N}$
can be written as $\opmap{f}{M} = \{f(m) \mid m \in M\}$ where $f =
\newrecrevise{\cdot}{F}{N}$.

\begin{lemma}
  \label{lem:delta_map}
Suppose $\Delta M$ is minimal with respect to $M$.  Then
\[\dopmap{f}{M,\Delta M} = (\opmap{f}{\pDelta M}
\setminus \opmap{f}{M},
\opmap{f}{\nDelta M} \setminus \opmap{f}{M \deltaplus \Delta M})\]
\end{lemma}
\begin{proof}
We first observe that
\begin{eqnarray*}
  \opmap{f}{M \cup N}& =& \opmap{f}{M} \cup (\opmap{f}{N} \setminus
  \opmap{f}{M})\\
\opmap{f}{M \setminus N} &=& \opmap{f}{M} \setminus (\opmap{f}{N}
                             \setminus \opmap{f}{M \setminus N})
\end{eqnarray*}
These equations are easy to show by calculation.  Combining them we
have
\begin{eqnarray*}
\opmap{f}{M \deltaplus \Delta M} &=&
  \opmap{f}{(M \cup \pDelta M) \setminus \nDelta M}
\\
& =&\opmap{f}{M \cup \pDelta M} \setminus (\opmap{f}{\nDelta M}
                             \setminus \opmap{f}{(M \cup \pDelta M)
     \setminus \nDelta M})
\\
& =&(\opmap{f}{M} \cup (\opmap{f}{\pDelta M} \setminus
  \opmap{f}{M})) \setminus (\opmap{f}{\nDelta M}
                             \setminus \opmap{f}{M \deltaplus \Delta M})
\\
& =&\opmap{f}{M} \deltaplus (\opmap{f}{\pDelta M} \setminus
  \opmap{f}{M},\opmap{f}{\nDelta M}
                             \setminus \opmap{f}{M \deltaplus \Delta M})
\end{eqnarray*}
Moreover, it is easy to see that
$\opmap{f}{\pDelta M} \setminus
\opmap{f}{M}$
is disjoint from $\opmap{f}{M}$ and
$\opmap{f}{\nDelta M} \setminus \opmap{f}{M \deltaplus \Delta
  M}\subseteq \opmap{f}{M}$,
assuming $\Delta M$ is minimal with respect to $M$.  Hence by
uniqueness of minimal deltas,
$\dopmap{f}{M,\Delta M} = (\opmap{f}{\pDelta M}
\setminus \opmap{f}{M},
\opmap{f}{\nDelta M} \setminus \opmap{f}{M \deltaplus \Delta M})$ as required.
\end{proof}
\begin{lemma}
  \label{lem:map_inj}
If $f$ is injective on $M \cup N$ and
$M, N$ are disjoint then so are $\opmap{f}{M}$ and $\opmap{f}{N}$.
\end{lemma}
\begin{proof}
Suppose $x \in \opmap{f}{M} \cap \opmap{f}{N}$.  Then $x = f(y) =
f(z)$ for some $y\in M$ and $z\in N$.  By injectivity, $f(y) = f(z)$
implies $y=z$, which is impossible since $M$ and $N$ are disjoint.
\end{proof}

\begin{lemma}
  \label{lem:delta_map_inj}
If $f$ is injective on $M$ and on $M \deltaplus \Delta M$, where
$\Delta M$ is minimal for $M$, then
\[\dopmap{f}{M,\Delta M} =  (\opmap{f}{\pDelta M},
\opmap{f}{\nDelta M})\]
\end{lemma}
\begin{proof}
  By Lem. \ref{lem:delta_map} we have
\[
\dopmap{f}{M,\Delta M} = (\opmap{f}{\pDelta M}
\setminus \opmap{f}{M},
\opmap{f}{\nDelta M} \setminus \opmap{f}{M \deltaplus \Delta M})\]
Using Lem. \ref{lem:map_inj} since $f$ is injective on $M$, and
$\Delta M$ is minimal, we have
\[\opmap{f}{\pDelta M}
\cap \opmap{f}{M} = \emptyset\]
hence, $\opmap{f}{\pDelta {M}} \setminus \opmap{f}{M} = \opmap{f}{\pDelta M}$.
Using the same lemma since $f$ is injective on $M\deltaplus \Delta
M$, since minimality implies $(M \deltaplus \Delta M) \cap \nDelta M = \emptyset$, we have
\[\opmap{f}{\nDelta M} \cap \opmap{f}{M \deltaplus \Delta M} =
\emptyset\]
and so $\opmap{f}{\nDelta M} \setminus \opmap{f}{M \deltaplus \Delta
  M} = \opmap{f}{\nDelta M}$.
The desired result follows.
\end{proof}

\begin{lemma}
  \label{lem:recrevise_antirestrict}
 $\domsubtract{\newrecrevise{m}{X \to A}{N} }{A} = \domsubtract{m}{A}$.
\end{lemma}
\begin{proof}
  There are two cases according to the definition of $\newrecrevise{m}{X\to A}{N}$.  If
  $m$ is unchanged the result is immediate, otherwise the result
  is $\recupdate{m}{n[A]}$ for some $n\in N$ where $n[X] = m[X]$, and only the $A$ field is
  changed.
\end{proof}

\begin{lemma}
  \label{lem:recrevise_inj}
If $M \vDash X \to A$ then $\newrecrevise{\cdot}{X \to A}{N}$ is
injective on $M$; that is, for $m,m' \in M$ we have
\[\newrecrevise{m}{X\to A}{N} = \newrecrevise{m'}{X\to A}{N}
\Longrightarrow m = m'\]
\end{lemma}
\begin{proof}
Clearly $A \notin X$ since $\{X \to A\}$ is required to be in tree
form.  Let $m,m' \in M$ be given and assume that $\newrecrevise{m}{X\to
  A}{N} = \newrecrevise{m'}{X \to A}{N}$.  By
Lem.~\ref{lem:recrevise_antirestrict}, we know also that
\[\domsubtract{m}{A} = \domsubtract{\newrecrevise{m}{X\to
  A}{N}}{A} = \domsubtract{\newrecrevise{m'}{X \to A}{N}}{A} = \domsubtract{m'}{A}
\]
  Since
$A \notin X$, it follows that $m[X] = m'[X]$, so $m[A] = m'[A]$
because $M \vDash X \to A$.
Thus,
\[m = \recupdate{\domsubtract{m}{A}}{m[A]} =
\recupdate{\domsubtract{m'}{A}}{m'[A]} = m'\]
Hence, $\newrecrevise{\cdot}{X \to A}{N}$ is injective.
\end{proof}
\restatelemdeltarelrev*
\begin{proof}
Let $ f = \newrecrevise{\cdot}{X \to A}{N}$.  Then by
Lem.~\ref{lem:recrevise_inj}, $f$ is injective on $M$ and $M
\deltaplus \Delta M$, so we have:
\begin{eqnarray*}
\relrevise{M\deltaapp \Delta M}{X \to A}{N} &=& \opmap{f}{M
    \deltaapp \Delta M} \\
&=& \opmap{f}{M} \deltaapp \dopmap{f}{M,\Delta
    M}
\\
&=& \relrevise{M}{X \to A}{N} \deltaapp \dopmap{f}{M,\Delta M}
\end{eqnarray*}
which implies
\begin{eqnarray*}
\deltarelrevise{(M,\Delta M)}{X \to A}{(N,\emptyset)}
&=& \dopmap{f}{M,\Delta M}\\
& =& (\opmap{f}{\pDelta M},
\opmap{f}{\nDelta M})\\
& =& (\relrevise{\pDelta M}{X\to A}{N},
\relrevise{\nDelta M}{X \to A}{N})
\end{eqnarray*}
by uniqueness of minimal deltas.
\end{proof}

\subsection*{Proof of Lemma~\ref{lem:delta_rel_merge}}

\begin{lemma}
\label{lem:fc_same_cases}
  Given $m,N,N':U$ with $X \subseteq U$, then exactly one of the following holds:
  \begin{itemize}
  \item There is no $n \in N'$ with $m[X] = n[X]$
  \item There exists $n \in N'-N$ with $m[X]=n[X]$.
\item There exists $n \in N'\cap N$ with $m[X] = n[X]$, but no $n'\in
  N'-N$ with $m[X] = n'[X]$.
  \end{itemize}
\end{lemma}
\begin{proof}
  It is easy to see that the three cases are mutually exclusive.  To
  see that they are exhaustive, suppose the first and second cases do
  not hold.  The failure of the first case implies that there is an
  $n\in N'$ with $m[X] = n[X]$ and the failure of the second case
  implies that such an $n$ must fall in $N'\cap N$ and there can be no
  other $n' \in N'-N$ satisfying $m[X]=n'[X]$, as required.
\end{proof}

\begin{definition}
  \label{def:fc_same_inv}
Suppose $m,N,N',F:U$ are given. For the sake of an inductive invariant for \reflemma{fc_same_gen}, we define
\[Inv_F(m,N,N') \iff \forall X \to Y \in F. \text{ if $\not\exists n\in
  (N'-N). m[X]=n[X]$ then  $m,(N'\cap N) \vDash X
  \to Y$}\]
\end{definition}

\begin{lemma}
  \label{lem:fc_same_inv_properties}
Given $m,N,N',F : U$, we have:
\begin{enumerate}
\item $m \in N \models F$ implies $Inv_F(m,N,N')$
\item $Inv_F(m,N,N')$ implies $Inv_{F'}(m,N,N')$ whenever $F'\subseteq F$.
\item If $Inv_F(m,N,N')$ and $F = \{X \to Y\}\cdot F'$ and there exists $n \in N'-N$ with $n[X] =
  m[X]$ then $Inv_{F'}(\recupdate{m}{n[Y]},N,N')$.
\end{enumerate}
\end{lemma}
\begin{proof}
  \begin{enumerate}
  \item Since $m,N \vDash F$ we know that $m,N
    \vDash X \to Y$ for all $X\to Y\in F$.  Clearly also $m,(N\cap N')
    \vDash X \to Y$.
\item For the second part, assume $Inv_F(m,N,N')$.  Suppose $X \to Y
  \in F' \subseteq F$ is given. Then $Inv_F(m,N,N')$  implies that if $\not\exists n\in
  (N'-N). m[X]=n[X]$ then  $m,(N'\cap N) \vDash X
  \to Y$, as required.
\item Suppose $Inv_F(m,N,N')$ and $F = \{X \to Y\}\cdot F'$ and there exists $n \in N'-N$ with $n[X] =
  m[X]$.  To show $Inv_{F'}(\recupdate{m}{n[Y]},N,N')$,
let $X' \to Y'
  \in F'$ be given.  We must show that if $\not\exists n' \in N'-N$ with $n'[X'] = (\recupdate{m}{n[Y]})[X']$ then
  $\recupdate{m}{n[Y]},(N\cap N') \models X'\to Y'$.

Because $F$ is in tree form, we know that
  $X',Y',X,Y$ are all either disjoint or identical, and $X \neq Y, X'
  \neq Y'$.  We consider the following cases:
  \begin{itemize}
  \item  $Y' = X$ is impossible since $X$ was a root and $X' \to X \in
    F'$ implies that $X$ could not be a root.
  \item  $Y' = Y$ is impossible since there cannot be two FDs $X \to
    Y, X'\to Y$ in $F$ since it is in tree form.
\item
  $X' = X$ is impossible since by assumption $n \in N'-N$ satisfies
  $n[X]=m[X] = (\recupdate{m}{n[Y]})[X]$, since $X,Y$ are disjoint,
  which contradicts the assumption that $\not\exists n' \in N'-N$ with $n'[X] = (\recupdate{m}{n[Y]})[X]$.
\item $X' = Y$ is impossible since $n[Y] =
  (\recupdate{m}{n[Y]})[Y]$, which contradicts the assumption that
  $\not\exists n' \in N'-N$ with $n'[Y] = (\recupdate{m}{n[Y]})[Y]$.
\item If $X'$ and $Y'$ are both disjoint from $X$ and $Y$ then using
  $Inv_F(m,N,N')$ we can conclude $m,(N\cap N')
  \models X'\to Y'$, and since $X',Y'$ are disjoint from $Y$ we can
  also conclude $\recupdate{m}{n[Y]},(N\cap N')
  \models X'\to Y'$ since updating $m[Y]$ to $n[Y]$ has no effect on the values of
  $m[X']$ or $m[Y']$.

  \end{itemize}

  \end{enumerate}
\end{proof}
\begin{lemma}
  \label{lem:fc_same_gen}
  Suppose $Inv_F(m,N,N')$.  Then
  $\newrecrevise{m}{F}{N'} = \newrecrevise{m}{F}{N'-N}$.
\end{lemma}
\begin{proof}
  Let $N,N'$ be given.  We prove by induction on $F$ that for any $m$,
  if $Inv_F(m,N,N')$ then
  $\newrecrevise{m}{F}{N'} = \newrecrevise{m}{F}{N'-N}$.
  \begin{itemize}
  \item If $F = \emptyset$ the desired conclusion is immediate.
  \item If $F = \{X \to Y\}\cdot F'$, then suppose $Inv_F(m,N,N')$,
    and note that
    $Inv_{F'}(m,N,N')$ holds by
    Lem.~\ref{lem:fc_same_inv_properties}(2).  We consider the
    following three subcases:
  \begin{itemize}
  \item If there is no $n \in N'$ such that $n[X] = m[X]$, then
    \begin{eqnarray*}
      \newrecrevise{m}{F}{N'}
&=&
\newrecrevise{m}{F'}{N'}\\
&=&
\newrecrevise{m}{F'}{N'\setminus N}\\
&=&
\newrecrevise{m}{F}{N'\setminus N}
    \end{eqnarray*}
using the induction hypothesis on $m$ and since  $N' \supseteq N'\setminus N$.
\item If there exists $n \in N' \setminus N$ such that $n[X] = m[X]$,
  then observe that by Lemma~\ref{lem:fc_same_inv_properties}(3) we
  have $Inv_{F'}(\recupdate{m}{n[Y]},N,N')$, so the induction
  hypothesis is available for $\recupdate{m}{n[Y]}$.  We reason as follows:
    \begin{eqnarray*}
      \newrecrevise{m}{F}{N'}
&=&
\newrecrevise{\recupdate{m}{n[Y]}}{F'}{N'}\\
&=&
\newrecrevise{\recupdate{m}{n[Y]} }{F'}{N'\setminus N}\\
&=&
\newrecrevise{m}{F}{N'\setminus N}
    \end{eqnarray*}
where we use the induction hypothesis on $\recupdate{m}{n[Y]}$  and the fact that $n \in N'
\setminus N$.
\item If there exists $n \in N \cap N'$ such that $n[X] = m[X]$ but no
  $n' \in N'\setminus N $ with $n'[X] = m[X]$ then $Inv_F(m,N,N')$
  implies that $m,(N \cap N') \vDash X \to Y$, so $m[X] = n[X]$ implies $m[Y] = n[Y]$.  Moreover, $\recupdate{m}{n[Y]}
  = \recupdate{m}{m[Y]} = m$. So, we can reason as follows:
    \begin{eqnarray*}
      \newrecrevise{m}{F}{N'}
&=&
\newrecrevise{\recupdate{m}{n[Y]}}{F'}{N'}\\
&=&
\newrecrevise{m }{F'}{N'}\\
&=&
\newrecrevise{m}{F'}{N'\setminus N}\\
&=&
\newrecrevise{m}{F}{N'\setminus N}
    \end{eqnarray*}
using the induction hypothesis on $m$ and the fact that there was no $n' \in
N'\setminus N$ matching $m$ on $X$.
  \end{itemize}
These three cases are exhaustive by lemma~\ref{lem:fc_same_cases}, so
the proof of the induction step is complete.
  \end{itemize}
\end{proof}

\begin{lemma}
	\label{lem:fc_same}
	Suppose $m \in N$
        and $N \vDash F$. Then $\newrecrevise{m}{F}{N \deltaapp \Delta N} = \newrecrevise{m}{F}{\pDelta N}$.
\end{lemma}
\begin{proof}
  Note that $(N \deltaapp \Delta N) \setminus N = \pDelta N$, and $m
  \in N$ implies $Inv_F(m,N,N \deltaapp \Delta N)$, so by
  Lemma~\ref{lem:fc_same_gen} we can conclude $\newrecrevise{m}{F}{N
    \deltaapp \Delta N} = \newrecrevise{m}{F}{\pDelta N}$.
\end{proof}

\begin{lemma}
	\label{lem:relmere_contains_bsubs}
	If $\relmerge{M}{F}{N} = M$ and $N \deltaapp \Delta N \vDash F$, then $N
	\setminus \negative{\Delta N} \subseteq \relmerge{M}{F}{\positive{\Delta
	N}}$.
\end{lemma}

\begin{proof}
	\def\currentprefix{proof:relmere_contains_bsubs}
	\begin{sflalign}
		& \relmerge{M}{F}{N} = M & \text{assumption} \notag
      \\ & N \subseteq M & \text{def. $\relmerge{\cdot}{\cdot}{\cdot}$} \notag
		\\ & N \setminus \nDelta N \subseteq N \subseteq M & \text{$\cdot \setminus \nDelta N$ decr.; transitivity} \locallabel{n_subdn_subset_m}
		\\ & N \deltaapp \Delta N \vDash F  & \text{assumption} \notag
      \\ & \indent m \in N \setminus \nDelta N & \text{suppose $(m)$} \locallabel{intro_m}
		\\ & \indent m \in N & \text{(\localref{n_subdn_subset_m})}
		\intertext{\indent Now show $\newrecrevise{m}{F}{\pDelta N} = m$ by induction on cardinality of $F$.}
		\interCaseOne{F = \emptyset \notag}
		& \indent\qedLocal \newrecrevise{m}{F}{\pDelta N} = m & \text{def. $\newrecrevise{\cdot}{F}{\cdot}$} \notag
		\\
		\interCaseOne{F = X \to Y \cdot F' \notag}
		& \indent\indent N \deltaapp \Delta N \vDash F' & \text{$N \deltaapp \Delta N \vDash F$ and $F' \subseteq F$} \notag
		\\
		\interCaseTwo{\exists n \in \pDelta N. \recproj{n}{X} = \recproj{m}{X} \notag}
		& \indent\indent \indent n[Y] = m[Y] & \text{$N \deltaapp \Delta N \vDash F$} \notag
		\\ & \indent\indent \indent \newrecrevise{m}{F}{\pDelta N} = \newrecrevise{\recupdate{m}{\recproj{n}{Y}}}{F'}{\pDelta N}
		& \text{def. $\newrecrevise{\cdot}{F}{\cdot}$} \notag
		\\ & \indent\indent \qedLocal \newrecrevise{m}{F}{\pDelta N} = \newrecrevise{m}{F'}{\pDelta N} = m
		& \text{$n[Y] = m[Y]$; IH} \notag
		\\
		\interCaseTwo{\nexists n \in \pDelta N. \recproj{n}{X} = \recproj{m}{X} \notag}
		& \indent\indent \qedLocal \newrecrevise{m}{F}{\pDelta N} = \newrecrevise{m}{F'}{\pDelta N} = m
		& \text{def. $\newrecrevise{\cdot}{F}{\cdot}$; IH}
		\\[2mm] & N \setminus \nDelta N = \relrevise{N \setminus \nDelta N}{F}{\pDelta N}
		& \text{$\forall$ intro (\localref{intro_m}); def. $\newrecrevise{\cdot}{F}{\cdot}$} \notag
      \\ & \phantom{N \setminus \nDelta N}\;\subseteq \relrevise{M}{F}{\pDelta N}
		& \text{$\newrecrevise{\cdot}{F}{\pDelta N}$ monotone} \notag
      \\ & \phantom{N \setminus \nDelta N}\;\subseteq \relmerge{M}{F}{\pDelta N}
		& \text{def.~$\relmerge{\cdot}{F}{\cdot}$} \notag
	\end{sflalign}
\end{proof}

\restatelemdeltarelmerge*

\begin{proof}
	\def\currentprefix{proof:delta_merge_opt}
	\begin{sflalign}
		& \relrevise{M}{F}{N \deltaapp \Delta N} =
			\left\{ \Genrecrevise_F(m, N \deltaapp \Delta N) \;\middle|\; m \in M \right\}
			& \text{def.~$\relrevise{\cdot}{}{\cdot}$} \notag
		\\ & \phantom{\relrevise{M}{F}{N \deltaapp \Delta N}} =
			\left\{ \Genrecrevise_F(m, \positive{\Delta N}) \;\middle|\; m \in M \right\}
			& \text{lemma \ref{lem:fc_same}} \notag
		\\ & \phantom{\relrevise{M}{F}{N \deltaapp \Delta N}} =
			\relrevise{M}{F}{\positive{\Delta N}}
			& \text{def.~$\relrevise{\cdot}{}{\cdot}$} \locallabel{revise_opt}
		\\[2mm] & \relmerge{M}{F}{N \deltaapp \Delta N} =
			\relrevise{M}{F}{N \deltaapp \Delta N} \cup (N \deltaapp \Delta N)
			& \text{def.~$\relmerge{\cdot}{}{\cdot}$} \notag
		\\ & \phantom{\relmerge{M}{F}{N \deltaapp \Delta N}} =
			\relrevise{M}{F}{\pDelta N} \cup (N \deltaapp \Delta N)
			& \text{(\localref{revise_opt})} \notag
		\\ & \phantom{\relmerge{M}{F}{N \deltaapp \Delta N}} =
			\relrevise{M}{F}{\pDelta N} \cup ((N \setminus \nDelta N) \cup \pDelta N)
			& \text{def.~$\deltaapp$} \notag
		\\ & \phantom{\relmerge{M}{F}{N \deltaapp \Delta N}} =
			\relrevise{M}{F}{\pDelta N} \cup \pDelta N \cup (N \setminus \nDelta N)
			& \text{comm. $\cup$; assoc. $\cup$} \notag
		\\ & \phantom{\relmerge{M}{F}{N \deltaapp \Delta N}} =
			\relmerge{M}{F}{\pDelta N} \cup (N \setminus \nDelta N)
			& \text{def.~$\relmerge{\cdot}{}{\cdot}$} \notag
		\\ & \phantom{\relmerge{M}{F}{N \deltaapp \Delta N}} =
			\relmerge{M}{F}{\pDelta N}
			& \text{lemma \ref{lem:relmere_contains_bsubs}} \locallabel{merge_opt}
		\\[2mm] & \deltarelmerge{(M,\emptyset)}{F}{(N,\Delta N)} =
			\relmerge{M}{F}{N \deltaapp \Delta N} \deltaminus \relmerge{M}{F}{N}
			& \notag
		\\ & \phantom{\deltarelmerge{(M,\emptyset)}{F}{(N,\Delta N)}} =
			\relmerge{M}{F}{N \deltaapp \Delta N} \deltaminus M
			& \text{assumption} \notag
		\\ & \phantom{\deltarelmerge{(M,\emptyset)}{F}{(N,\Delta N)}} =
			\relmerge{M}{F}{\pDelta{N}} \deltaminus M
			& \text{(\localref{merge_opt})} \notag
	\end{sflalign}
\end{proof}

\subsection*{Proof of Lemma~\ref{lem:relmerge_efficient}}

\begin{lemma}
	\label{lemma:gen_record_revision}
	If $P=\affected_{G}(N)$, $F \subseteq G$ and $m \in \select{\neg P}{M}$ for any $M$, then $\newrecrevise{m}{F}{N} = m$.
\end{lemma}

\begin{proof}
	\def\currentprefix{proof:gen_record_revision}
	\begin{sflalign}
		& P = \affected_{G}(N) & \text{suppose ($P$)} \notag
		\\ & m \in \select{\neg P}{M} & \text{suppose ($m$)} \notag
		\intertext{Now show $\newrecrevise{m}{F}{N} = m$ by induction on the cardinality of $F$.}
		\interCaseZero{F = \emptyset \notag}
		& \qedLocal \newrecrevise{m}{F}{N} = m & \text{def. $\newrecrevise{\cdot}{F}{\cdot}$} \notag
		\\
		\interCaseZero{F = X \to Y \cdot F' \notag}
		& \indent \nexists n \in N. \recproj{n}{X} = \recproj{m}{X} & \text{$m \in \select{\neg P}{M}$; def. $\affected_{G}(\cdot)$} \notag
		\\ & \qedLocal \newrecrevise{m}{F}{N} = \newrecrevise{m}{F'}{N} = m & \text{def. $\newrecrevise{\cdot}{F}{\cdot}$; IH} \notag
	\end{sflalign}
\end{proof}

\begin{lemma}
	\label{lemma:relational_revision}
	If $P = \affected_F(N)$, then $\relrevise{M}{F}{N} = \relrevise{\select{P}{M}}{F}{N} \cup \select{\neg P}{M}$.
\end{lemma}

\begin{proof}
	\def\currentprefix{proof:relational_revision}
	\begin{sflalign}
		& P = \affected_F(N) & \text{suppose ($P$)} \notag
		\\[2mm] & \indent m \in \select{\neg P}{M} & \text{suppose ($m$)} \notag
		\\ & \indent m = \newrecrevise{m}{F}{N} & \text{lemma \ref{lemma:gen_record_revision}} \notag
		\\ & \relrevise{\select{\neg P}{M}}{F}{N} = \select{\neg P}{M}
		& \text{$\forall$ intro ($m$); def. $\relrevise{\cdot}{F}{\cdot}$} \locallabel{revise_simpl}
		\\[2mm] & M = \select{P}{M} \cup \select{\neg P}{M} & \text{decompose $M$} \notag
		\\[2mm] & \relrevise{M}{F}{N} = \relrevise{\select{P}{M}}{F}{N} \cup \relrevise{\select{\neg P}{M}}{F}{N}
		& \text{def. $\relrevise{\cdot}{F}{\cdot}$} \notag
		\\ & \phantom{\relrevise{M}{F}{N}} = \relrevise{\select{P}{M}}{F}{N} \cup \select{\neg P}{M}
		& \text{(\localref{revise_simpl})} \notag
	\end{sflalign}
\end{proof}

\begin{lemma}
	\label{lemma:relational_merge}
	If $P = \affected_F(N)$, then $\relmerge{M}{F}{N} = \relmerge{\select{P}{M}}{F}{N} \cup \select{\neg P}{M}$.
\end{lemma}

\begin{proof}
	\def\currentprefix{proof:relational_merge}
	\begin{sflalign}
		& P = \affected_F(N) & \text{suppose ($P$)} \notag
		\\[2mm] & \relmerge{M}{F}{N} = \relrevise{M}{F}{N} \cup N & \text{def. $\relmerge{\cdot}{F}{\cdot}$} \notag
		\\ & \phantom{\relmerge{M}{F}{N} } =
		(\relrevise{\select{P}{M}}{F}{N} \cup \select{\neg P}{M}) \cup N & \text{lemma \ref{lemma:relational_revision}} \notag
		\\ & \phantom{\relmerge{M}{F}{N} } =
		(\relrevise{\select{P}{M}}{F}{N} \cup N) \cup \select{\neg P}{M} & \text{comm. $\cup$} \notag
		\\ & \phantom{\relmerge{M}{F}{N} } =
		\relmerge{\select{P}{M}}{F}{N} \cup \select{\neg P}{M} & \text{def. $\relmerge{\cdot}{F}{\cdot}$} \notag
	\end{sflalign}
\end{proof}

\begin{lemma}
	If $M \cap P = N \cap O = \emptyset$, then $(M \cup N) \deltaminus (O \cup P) = (M \deltaminus O) \deltaplus (N \deltaminus P)$.
	\label{lem:disjoint_deltaminus_deltaplus}
\end{lemma}

\begin{proof}
	\def\currentprefix{proof:lem_deltaminus}
	First observe that because $M \cap P = N \cap O = \emptyset$, we have
	\begin{salign}
		(M \cup N) \setminus (O \cup P) &= (M \setminus (O \cup P)) \cup (N \setminus (O \cup P))
		& \text{$\setminus$/$\cup$ distr.}
		\\
		\phantom{(M \cup N) \setminus (O \cup P)} &= (M \setminus O) \cup (N \setminus P)
	\end{salign}
	Now we proceed as follows:
	\begin{salign}
		& (M \cup N) \deltaminus (O \cup P) = ( (M \cup N) \setminus (O \cup P), (O \cup P) \setminus (M \cup N) )
		& \text{\reflemma{diff}} \notag
		\\
		& \phantom{(M \cup N) \deltaminus (O \cup P)} =
		( (M \setminus O) \cup (N \setminus P), (O \setminus M) \cup (P \setminus N) ) & \text{observation} \notag
		\\
		& \phantom{(M \cup N) \deltaminus (O \cup P)} =
		( ( (M \setminus O) \setminus (P \setminus N) ) \cup ( (N \setminus P) \setminus (O \setminus M) ), & \notag
		\\
		& \phantom{(M \cup N) \deltaminus (O \cup P)} \quad
		\;\;( (O \setminus M) \setminus (N \setminus P) ) \cup ( (P \setminus N) \setminus (M \setminus O) ))
		& \text{$M \cap P = \emptyset$; $N \cap O = \emptyset$} \notag
		\\
		& \phantom{(M \cup N) \deltaminus (O \cup P)} = (M \setminus O, O \setminus M) \deltaplus (N \setminus P, P \setminus N)
		& \text{def.~$\deltaplus$} \notag
		\\
		& \phantom{(M \cup N) \deltaminus (O \cup P)} = (M \deltaminus O) \deltaplus (N \deltaminus P)
		& \text{\reflemma{diff}} \notag
	\end{salign}
\end{proof}

\restaterelmergeefficient*

\begin{proof}
	\def\currentprefix{proof:relational_merge}
	\begin{sflalign}
		& P = \affected_F(\pDelta N) = \bigvee_{X \to Y \in F} X \in \proj{\pDelta N}{X} & \text{suppose ($P$); def. $\affected_F(\cdot)$} \locallabel{blah}
		\\ & \text{First show $\relmerge{\select{P}{M}}{F}{\pDelta N} \cap \select{\neg P}{M} = \emptyset$ by case analysis on F.}
      & \locallabel{disjoint}
      \\
		\interCaseZero{F = \emptyset \notag}
		& \indent \relmerge{\select{P}{M}}{F}{\pDelta N} = \relmerge{\select{\bot}{M}}{F}{\pDelta N}
		& \text{def. $\affected_F(\cdot)$} \notag
		\\ & \indent \phantom{\relmerge{\select{P}{M}}{F}{\pDelta N}} =
		\relmerge{\emptyset}{F}{\pDelta N} & \text{$\select{\bot}{\cdot }$} \notag
		\\ & \indent \phantom{\relmerge{\select{P}{M}}{F}{\pDelta N}} =
		\relrevise{\emptyset}{F}{\pDelta N} \cup \pDelta N = \pDelta N & \text{def.~$\relmerge{\cdot}{F}{\cdot }$; def.~$\relrevise{\cdot}{F}{\cdot}$} \notag
		\\ & \qedLocal \relmerge{\select{P}{M}}{F}{\pDelta N} \cap \select{\neg P}{M} = \pDelta N \cap \select{\neg P}{M} = \emptyset
		& \text{def. $\affected_F(\cdot)$} \notag
		\\
		\interCaseZero{F = X \to Y \cdot F' \notag}
		& \indent P = X \in \proj{\pDelta N}{X} \vee \bigvee_{X' \to Y' \in F'} X' \in \proj{\pDelta N}{X'}
		& \text{expand $F$ in $P$} \notag
		\\ & \indent \neg P = X \notin \proj{\pDelta N}{X} \wedge \bigwedge_{X' \to Y' \in F'} X' \notin \proj{\pDelta N}{X'}
		& \text{de Morgan} \notag
		\\ & \indent\indent m \in \relmerge{\select{P}{M}}{F}{\pDelta N} = \relrevise{\select{P}{M}}{F}{\pDelta N} \cup \pDelta N
		& \text{suppose ($m$); def. $\relmerge{\cdot}{F}{\cdot}$} \notag
		\\
		\interCaseTwo{m \in \relrevise{\select{P}{M}}{F}{\pDelta N} \notag}
		& \indent\indent\indent \exists m' \in \select{P}{M}.m = \newrecrevise{m'}{F}{\pDelta N}\text{ and }\recproj{m'}{X} = \recproj{m}{X} & \text{def. $\relrevise{\cdot}{F}{\cdot}$} \notag
		\\ & \indent\indent\indent \recproj{m'}{X} \in \proj{\pDelta N}{X} & \notag
      \\ & \indent\indent\indent \recproj{m}{X} \in \proj{\pDelta N}{X} & \text{$\recproj{m}{X} = \recproj{m'}{X}$} \notag
		\\ & \indent\indent\indent \forall n \in \select{\neg P}{M}. \recproj{n}{X} \neq \recproj{m}{X} & \text{$\recproj{n}{X} \notin \proj{\pDelta N}{X}$} \notag
		\\ & \indent\indent\qedLocal m \notin \select{\neg P}{M} & \notag
		\\
		\interCaseTwo{m \in \pDelta N \notag}
		& \indent\indent\qedLocal m \notin \select{\neg P}{M} & \text{def. $\affected_F(\cdot)$} \notag
		\\ & \qedLocal \relmerge{\select{P}{M}}{F}{\pDelta N} \cap \select{\neg P}{M} = \emptyset & \text{$\forall$ intro ($m$)} \notag
		\\[2mm] & \relmerge{M}{F}{\pDelta N} \deltaminus M & \notag
		\\ & \qquad = (\relmerge{\select{P}{M}}{F}{\pDelta N} \cup \select{\neg P}{M}) \deltaminus (\select{P}{M} \cup \select{\neg P}{M})
		& \text{lemma \ref{lemma:relational_merge}; decompose $M$} \notag
		\\ & \qquad = (\relmerge{\select{P}{M}}{F}{\pDelta N} \deltaminus \select{P}{M}) \deltaplus (\select{\neg P}{M} \deltaminus \select{\neg P}{M})
		& \text{(\localref{disjoint}); lemma \ref{lem:disjoint_deltaminus_deltaplus}} \notag
		\\ & \qquad = \relmerge{\select{P}{M}}{F}{\pDelta N} \deltaminus \select{P}{M}
		& \text{simpl.} \notag
	\end{sflalign}
\end{proof}

\section{Proofs from \Secref{incremental}}
\subsection*{Proof of Theorem~\ref{thm:select:optimised}}

\begin{lemma}
	\label{lem:select_relrevise_rewrite_disjoint}
	Suppose $M \vDash F$. Then
	$\relrevise{\select{\neg P}{M}}{F}{N} \cap \select{P}{M} = \emptyset$.
\end{lemma}

\begin{proof}
	\def\currentprefix{proof:select_relrevise_rewrite_disjoint}
	\begin{sflalign}
		& Z = U \setminus outputs(F) & \text{suppose ($Z$)} \notag
		\\ & m \in \select{\neg P}{M} & \text{suppose ($m$)} \notag
		\\
		\interCaseZero{\newrecrevise{m}{F}{N} = m \notag}
		& \qedLocal m \notin \select{P}{M} & \text{$m \in \select{\neg P}{M}$} \notag
		\\
		\interCaseZero{\newrecrevise{m}{F}{N} = m' \neq m \notag}
		& \indent m'[Z] = m[Z] & \text{def.~$\newrecrevise{\cdot}{F}{\cdot}$} \notag
		\\ & \indent \forall m'' \in M \text{ and } \recproj{m}{Z} = \recproj{m''}{Z}. m = m'' \neq m' & \text{$M \models F$} \notag
		\\ & \qedLocal m' \notin \select{P}{M} & \text{$m' \notin M$} \notag
	\end{sflalign}
\end{proof}

\begin{lemma}
	\label{lem:select_relrevise_rewrite}
	Suppose $M \vDash F$. Then $\relrevise{\select{\neg P}{M}}{F}{N} =
	\relrevise{\select{\neg P}{M}}{F}{N \setminus \select{P}{M}}$.
\end{lemma}

\begin{proof}
It suffices to show that if $m \in \select{\neg P}{M}$, then $\Genrecrevise_F(m,
N)=\Genrecrevise_F(m, N \setminus \select{P}{M})$, by the definition of
relational revision.  This follows from Lem.~\ref{lem:fc_same_gen}, provided
that  $Inv_F(m,\select{P}{M},N)$ holds. To show this, let $X \to Y \in F $ be
given and assume that there is no $n\in N \setminus \select{P}{M}$ such that
$n[X] = m[X]$. We need to show that $m,(\select{P}{M} \cap N) \vDash X \to Y$.
This is immediate since $\{m\} \cup (\select{P}{M} \cap N) \subseteq M \vDash F$.
\end{proof}

\begin{lemma}
	\label{lem:select_relmerge_rewrite}
	If $M \vDash F$ then $
		\relmerge{\select{\neg P}{M}}{F}{N}=
		\relmerge{\select{\neg P}{M}}{F}{N \setminus \select{P}{M}} \cup (\select{P}{M} \cap N).
	$
\end{lemma}

\begin{proof}
\begin{salign}
&\phantom{=}\;
\relmerge{\select{\neg P}{M}}{F}{N}
\\
&=
\relrevise{\select{\neg P}{M}}{F}{N} \cup N
&
\text{def.~$\relmerge{\param}{F}{\param}$}
\\
&=
\relrevise{\select{\neg P}{M}}{F}{N \setminus \select{P}{M}} \cup N
&
\text{\reflemma{select_relrevise_rewrite}}
\\
&=
\relrevise{\select{\neg P}{M}}{F}{N \setminus \select{P}{M}} \cup (N \setminus \select{P}{M}) \cup (\select{P}{M} \cap N)
&
\text{decompose $N$}
\\
&=
\relmerge{\select{\neg P}{M}}{F}{N \setminus \select{P}{M}} \cup (\select{P}{M} \cap N)
&
\text{def.~$\relmerge{\param}{F}{\param}$}
\end{salign}
\end{proof}

\begin{lemma}
	\label{lem:select-merge-opt}
	Suppose $M \models F$ and $\Delta N$ minimal for $\select{P}{M}$. Then $\relmerge{\select{\neg
	P}{M}}{F}{\select{P}{M} \deltaapp \Delta N}\;\deltaminus$\\ $\relmerge{\select{\neg P}{M}}{F}{\select{P}{M}} = (\relmerge{\select{\neg
	P}{M}}{F}{\positive{\Delta N}} \deltaminus \select{\neg P}{M}) \deltaminus
	\negative{\Delta N}$.
\end{lemma}

\begin{proof}
	\def\currentprefix{proof:select_merge_opt}
	\begin{sflalign}
		& (\select{P}{M} \deltaapp \Delta N) \setminus \select{P}{M}
		  = \pDelta N & \text{\reflemma{delta-basics-positive}} \locallabel{n_dn_sub_selpm}
		\\[2mm]
		& \relrevise{\select{\neg P}{M}}{F}{\pDelta N} \cap \select{P}{M} = \emptyset
			& \text{(\localref{n_dn_sub_selpm}); \reflemma{select_relrevise_rewrite_disjoint}} \notag
		\\ & \pDelta N \cap \select{P}{M} = \emptyset
			& \text{(\localref{n_dn_sub_selpm})} \notag
		\\ & (\relrevise{\select{\neg P}{M}}{F}{\pDelta N} \cup \pDelta N) \cap \select{P}{M} = \emptyset
			& \notag
		\\ & \relmerge{\select{\neg P}{M}}{F}{\pDelta N} \cap \select{P}{M} = \emptyset
			& \text{def.~$\Relmerge$} \locallabel{merge_disjoint}
		\\[2mm]
		& \select{P}{M} \cap \select{\neg P}{M} = \emptyset & \notag
		\\ & \select{P}{M} \cap (\select{P}{M} \deltaapp \Delta N) \cap \select{\neg P}{M} = \emptyset
		  & \text{$\cap$ monotone} \locallabel{cap_disjoint}
		\\[2mm]
		& (\select{P}{M} \cap (\select{P}{M} \deltaapp \Delta N)) \deltaminus \select{P}{M}
		  = \deltaminus \nDelta N & \text{\reflemma{delta-basics-negative}} \locallabel{derive_ndelta_n}
		\\[2mm]
		& \relmerge{\select{\neg P}{M}}{F}{\select{P}{M} \deltaapp \Delta N} \deltaminus \relmerge{\select{\neg P}{M}}{F}{\select{P}{M}} & \notag
		\\ & = (\relmerge{\select{\neg P}{M}}{F}{(\select{P}{M} \deltaapp \Delta N) \setminus \select{P}{M}} \cup (\select{P}{M} \cap (\select{P}{M} \deltaapp \Delta N))) & \notag
		\\ & \qquad \deltaminus (\relmerge{\select{\neg P}{M}}{F}{\select{P}{M} \setminus \select{P}{M}} \cup (\select{P}{M} \cap \select{P}{M}))
			& \text{\reflemma{select_relmerge_rewrite}; \reflemma{select_relmerge_rewrite}} \notag
		\\ & = (\relmerge{\select{\neg P}{M}}{F}{\pDelta N} \cup (\select{P}{M} \cap (\select{P}{M} \deltaapp \Delta N))) & \notag
		\\ & \qquad \deltaminus (\relmerge{\select{\neg P}{M}}{F}{\emptyset} \cup \select{P}{M}) & \text{(\localref{n_dn_sub_selpm}); simpl.} \notag
		\\ & = (\relmerge{\select{\neg P}{M}}{F}{\pDelta N} \cup (\select{P}{M} \cap (\select{P}{M} \deltaapp \Delta N)))
			\deltaminus (\select{\neg P}{M} \cup \select{P}{M}) & \text{$\relmerge{\cdot}{F}{\emptyset} = \mathop{id}$} \notag
		\\ & = (\relmerge{\select{\neg P}{M}}{F}{\pDelta N} \deltaminus \select{\neg P}{M}) \deltaapp
		((\select{P}{M} \cap (\select{P}{M} \deltaapp \Delta N)) \deltaminus \select{P}{M}) & \text{(\localref{merge_disjoint}, \localref{cap_disjoint}); \reflemma{disjoint_deltaminus_deltaplus}} \notag
		\\ & = (\relmerge{\select{\neg P}{M}}{F}{\pDelta N} \deltaminus \select{\neg P}{M}) \deltaminus \nDelta N
		  & \text{(\localref{derive_ndelta_n}); def.~$\deltaminus$} \notag
	\end{sflalign}
\end{proof}

\restateselectoptimised*

\setcounter{equation}{0}
\proofContext{thm:select:optimised}
\begin{proof}
\small
\begin{flalign}
&
N
=
\getOp_{\ell}(M)
&
\text{%
   suppose ($N, M$)
}
\notag
\\
&
(M_0, \Delta M_0)
=
\Deltarelmerge{\Deltaselect{\neg P}{M, \emptyset}}{F}{(N, \Delta N)}
&
\text{%
   suppose ($M_0, \Delta M_0, \Delta N$)
}
\notag
\\
&
(N_\#,\Delta N_\#)
=
\Deltaselect{P}{M_0, \Delta M_0} \Deltasetminus (N, \Delta N)
&
\text{%
   suppose ($N_\#,\Delta N_\#)$
}
\locallabel{def-N_hash}
\\[2mm]
&
\deltaOp{put}_{\ell}(M, \Delta N) = (M_0, \Delta M_0) \deltasetminus (N_\#, \Delta N_\#)
&
\text{%
   def.~$\deltaOp{put}_{\ell}$
}
\locallabel{def-put-v}
\\[2mm]
&
N
=
\select{P}{M}
&
\text{%
   def.~$\getOp_{\ell}$
}
\notag
\\[2mm]
&
M_0 = \relmerge{\select{\neg P}{M}}{F}{N}
&
\text{%
	def.~$\lift{\cdot}$; $(\cdot, \cdot)$ reflects $=$
}
\notag
\\
&
\phantom{M_0 }= \relmerge{\select{\neg P}{M}}{F}{\select{P}{M}}
&
\text{%
	$N = \select{P}{M}$
}
\notag
\\
&
\phantom{M_0 }= \select{\neg P}{M} \cup \select{P}{M}
&
\text{%
	$M \models F$; defs.~$\relrevise{\cdot}{F}{\cdot}$, $\relmerge{\cdot}{F}{\cdot}$
}
\notag
\\
&
\phantom{M_0 }= M
&
\text{%
}
\notag
\\[2mm]
&
N \subseteq N = \select{P}{M} = \select{P}{M_0}
&
\text{%
}
\locallabel{N-subseteq-select-P-M_0}
\\[2mm]
&
M_0 \deltaapp \Delta M_0 = \relmerge{\select{\neg P}{M}}{F}{N
  \deltaapp \Delta N}
&
\text{delta-correctness}
&
\notag
\\
&
N \deltaapp \Delta N
\subseteq
M_0 \deltaapp \Delta M_0
&
\text{%
  def.~$\Relmerge$
}
\notag
\\
&
\select{P}{N \deltaapp \Delta N}
\subseteq
\select{P}{M_0 \deltaapp \Delta M_0}
&
\text{%
   monotone $\select{P}{\cdot}$
}
\notag
\\
&
N \deltaapp \Delta N
\subseteq
\select{P}{M_0 \deltaapp \Delta M_0}
&
\text{%
   idemp.~$\select{P}{\cdot}$
}
\notag
\\
&
\phantom{N \deltaapp \Delta N}
\subseteq
\select{P}{M_0} \deltaapp \deltaselect{P}{M_0, \Delta M_0}
&
\text{%
   delta-correctness
}
\locallabel{N-plus-delta-N-subseteq}
\\[2mm]
&
(N_\#,\Delta N_\#)
=
(\select{P}{M_0} \setminus N,
 (\select{P}{M_0}, \deltaselect{P}{M_0, \Delta M_0}) \deltasetminus (N, \Delta N))
&
\text{%
   def.~$\lift{\cdot}$
}
\notag
\\
&
\Delta N_\#
=
(\select{P}{M_0}, \deltaselect{P}{M_0, \Delta M_0}) \deltasetminus (N, \Delta N)
&
\text{%
   $(\cdot, \cdot)$ reflects $=$
}
\notag
\\
&
\phantom{\Delta N_\#}
=
\deltaselect{P}{M_0, \Delta M_0} \deltaminus \Delta N
&
\text{%
   (\localref{N-subseteq-select-P-M_0}, \localref{N-plus-delta-N-subseteq}); \reflemma{correctness-delta-ops} part 5
}
\notag
\\
&
\phantom{\Delta N_\#}
=
(\select{P}{\pDelta{M_0}},\select{P}{\nDelta{M_0}}) \deltaminus \Delta N
&
\text{%
	\reflemma{correctness-delta-ops} part 1
}
\notag
\\[2mm]
&
Q = \affected_F(\positive{\Delta N})
&
\text{%
	suppose ($Q$)
}
\locallabel{affected-query}
\\[2mm]
&
\Delta M_0
=
\deltarelmerge{(\select{\neg P}{M}, \deltaselect{\neg P}{M, \emptyset})}{F}{(N, \Delta N)}
&
\text{%
   def.~$\lift{\cdot}$; $(\cdot, \cdot)$ reflects $=$
}
\notag
\\
&
\phantom{\Delta M_0}
=
\deltarelmerge{(\select{\neg P}{M}, \emptyset)}{F}{(N, \Delta N)}
&
\text{%
   delta-correctness
}
\notag
\\
&
\phantom{\Delta M_0}
=
\relmerge{\select{\neg P}{M}}{F}{N \deltaplus \Delta N} \deltaminus \relmerge{\select{\neg P}{M}}{F}{N}
&
\text{%
	def. $\deltaOp{op}$
}
\notag
\\
&
\phantom{\Delta M_0}
=
(\relmerge{\select{\neg P}{M}}{F}{\positive{\Delta N}} \deltaminus \select{\neg P}{M}) \deltaminus \negative{\Delta N}
&
\text{%
   lemma~\ref{lem:select-merge-opt}
}
\notag
\\
&
\phantom{\Delta M_0}
=
(\relmerge{\select{Q \wedge \neg P}{M}}{F}{\positive{\Delta N}} \deltaminus \select{Q \wedge \neg P}{M}) \deltaminus \negative{\Delta N}
&
\text{%
   (\localref{affected-query}); lemma~\ref{lem:relmerge_efficient}
}
\notag
\\[2mm]
&
\deltaOp{put}_{\ell'}(M, \Delta N) = \Delta M_0 \deltaminus \Delta N_\#
&
\text{%
   def.~$\deltaOp{put}_{\ell'}$
}
\locallabel{def-put-v-prime}
\\[2mm]
&
N_\# \subseteq N_\# = \select{P}{M_0} \setminus N \subseteq \select{P}{M_0} \subseteq M_0
&
\text{%
	def.~$N_\#$
}
\locallabel{N_hash-below-M_0}
\\
&
N_\# \deltaapp \Delta N_\#
=
\select{P}{M_0 \deltaapp \Delta M_0} \setminus (N \deltaapp \Delta N)
\subseteq
\select{P}{M_0 \deltaapp \Delta M_0}
&
\text{%
   (\localref{def-N_hash}); delta-correctness
}
\notag
\\
&
\phantom{
N_\# \deltaapp \Delta N_\#
=
\select{P}{M_0 \deltaapp \Delta M_0} \setminus (N \deltaapp \Delta N)
}
\subseteq
M_0 \deltaapp \Delta M_0
&
\locallabel{N_hash-oplus-below-M_0-oplus}
\\[2mm]
&
(M_0, \Delta M_0) \deltasetminus (N_\#, \Delta N_\#)
=
\Delta M_0 \deltaminus \Delta N_\#
&
\text{%
   (\localref{N_hash-below-M_0}, \localref{N_hash-oplus-below-M_0-oplus}); \reflemma{correctness-delta-ops} part 5)
}
\locallabel{ext-equiv}
\\[2mm]
&
\deltaOp{put}_{\ell'}(M, \Delta N)
=
\deltaOp{put}_{\ell}(M, \Delta N)
&
\text{%
   (\localref{def-put-v-prime}, \localref{ext-equiv}, \localref{def-put-v}); transitivity
}
\notag
\end{flalign}
\end{proof}

\subsection*{Proof of Theorem~\ref{thm:project:optimised}}

\restateprojectoptimised*

\setcounter{equation}{0}
\proofContext{thm:project:optimised}
\begin{proof}
\small
\begin{flalign}
&
N
=
\getOp_{\ell}(M)
&
\text{%
   suppose ($N, M$)
}
\notag
\\
&
(M', \Delta M')
=
(N, \Delta N) \DeltaJoin (\set{\set{A = a}}, \emptyset)
&
\text{%
   suppose ($M', \Delta M', \Delta N$)
}
\notag
\\[2mm]
&
\deltaOp{put}_{\ell}(M, \Delta N)
=
\deltarelrevise{(M', \Delta M')}{X \to A}{(M, \emptyset)}
&
\text{%
   def.~$\deltaOp{put}_{\ell}$
}
\locallabel{def-put-v}
\end{flalign}
\begin{flalign}
&
\Delta M'
=
(N, \Delta N) \deltaJoin (\set{\{A = a\}}, \emptyset)
&
\text{%
   def.~$\lift{\cdot}$; $(\cdot,\cdot)$ reflects $=$
}
\notag
\\
&
\phantom{\Delta M'}
=
(\posDN \Join \set{\{A = a\}}, \negDN \Join \set{\{A = a\}})
&
\text{%
   lemma \ref{lem:correctness-delta-ops} part 3
}
\notag
\\[2mm]
&
\deltaOp{put}_{\ell'}(M, \Delta N)
=
(\relrevise{\pDelta{M'}}{X \to A}{M}, \relrevise{\nDelta{M'}}{X \to A}{M})
&
\text{%
   def.~$\deltaOp{put}_{\ell'}$
}
\locallabel{def-put-v-prime}
\\[2mm]
&
(\relrevise{\pDelta{M'}}{X \to A}{M}, \relrevise{\nDelta{M'}}{X \to A}{M})
& \notag
\\ & \quad 
=
\deltarelrevise{(M', \Delta M')}{X \to A}{(M, \emptyset)}
&
\text{%
   lemma \ref{lem:delta_rel_rev}
}
\locallabel{ext-equiv}
\\[2mm]
&
\deltaOp{put}_{\ell'}(M, \Delta N)
=
\deltaOp{put}_{\ell}(M, \Delta N)
&
\text{%
   (\localref{def-put-v-prime}, \localref{ext-equiv}, \localref{def-put-v}); transitivity
}
\notag
\end{flalign}
\end{proof}

\subsection*{Proof of Theorem~\ref{thm:join:optimised}}

\restatejoinoptimised*

\setcounter{equation}{0}
\proofContext{thm:join:optimised}
\begin{proof}
\small
\begin{flalign}
&
O
=
M \Join N
&
\text{%
   suppose ($O, M, N$); def.~$\getOp_{\ell}$
}
\notag
\\
&
(M_0, \Delta M_0)
= \Deltarelmerge{(M, \emptyset)}{F}{\Deltaproj{O, \Delta O}{U}}
&
\text{%
   suppose ($M_0, \Delta M_0, \Delta O$)
}
\locallabel{def-M_0}
\\
&
(N', \Delta N')
=
\Deltarelmerge{(N, \emptyset)}{G}{\Deltaproj{O, \Delta O}{V}}
&
\text{%
   suppose ($N', \Delta N'$)
}
\locallabel{def-N-prime}
\\
&
(L, \Delta L)
=
((M_0, \Delta M_0) \DeltaJoin (N', \Delta N')) \Deltasetminus (O, \Delta O)
&
\text{%
suppose ($L, \Delta L$)
}
\locallabel{def-delta-L}
\\
&
\Delta M'
=
(M_0, M_0') \deltasetminus \Deltaproj{L, \Delta L}{U}
&
\text{%
   suppose ($\Delta M'$)
}
\notag
\\[2mm]
&
\deltaOp{put}_{\ell}((M, N), \Delta O)
=
(\Delta M', \Delta N')
&
\text{%
   def.~$\deltaOp{put}_{\ell}$
}
\locallabel{def-put-v}
\\[2mm]
&
P_M
=
\affected_F(\proj{\positive{\Delta O}}{U})
&
\text{%
	suppose ($P_M$)
}
\locallabel{affected-left-query}
\\
&
P_N
=
\affected_G(\proj{\positive{\Delta O}}{V})
&
\text{%
	suppose ($P_N$)
}
\locallabel{affected-right-query}
\\[2mm]
&
M_0
=
\relmerge{M}{F}{\proj{O}{U}}
=
M
&
\text{%
   def.~$\lift{\cdot}$; $\proj{O}{U} \subseteq M$
}
\notag
\\
&
\Delta M_0
=
\deltarelmerge{(M, \emptyset)}{F}{(\proj{O}{U}, \deltaproj{O, \Delta O}{U})}
&
\text{%
   def.~$\lift{\cdot}$
}
\notag
\\
&
\phantom{\Delta M_0}
=
\relmerge{M}{F}{\positive{\deltaproj{O, \Delta O}{U}}} \deltaminus M
&
\text{%
   lemma \ref{lem:delta_rel_merge}
}
\notag
\\
&
\phantom{\Delta M_0}
=
\relmerge{M}{F}{\proj{\positive{\Delta O}}{U} \setminus \proj{O}{U}} \deltaminus M
&
\text{%
	lemma \ref{lem:correctness-delta-ops} part 2
}
\notag
\\
&
\phantom{\Delta M_0}
=
\relmerge{M}{F}{\proj{\positive{\Delta O}}{U}} \deltaminus M
&
\text{%
   $\proj{O}{U} \subseteq M$
}
\notag
\\
&
\phantom{\Delta M_0}
=
\relmerge{\select{P_M}{M}}{F}{\proj{\positive{\Delta O}}{U}} \deltaminus \select{P_M}{M}
&
\text{%
	(\localref{affected-left-query}); lemma~\ref{lem:relmerge_efficient}
}
\notag
\\[2mm]
&
N'
=
\relmerge{N}{G}{\proj{O}{V}}
=
N
&
\text{%
   def.~$\lift{\cdot}$; $\proj{O}{V} \subseteq N$
}
\notag
\\
&
\Delta N'
=
\deltarelmerge{(N, \emptyset)}{G}{(\proj{O}{V}, \deltaproj{O, \Delta O}{V})}
&
\text{%
   def.~$\lift{\cdot}$
}
\locallabel{def-delta-N-prime}
\\
&
\phantom{\Delta N'}
=
\relmerge{N}{G}{\positive{\deltaproj{O, \Delta O}{V}}} \deltaminus N
&
\text{%
   lemma \ref{lem:delta_rel_merge}
}
\notag
\\
&
\phantom{\Delta N'}
=
\relmerge{N}{G}{\proj{\positive{\Delta O}}{V} \setminus \proj{O}{V}} \deltaminus N
&
\text{%
	\reflemma{correctness-delta-ops} part 2
}
\notag
\\
&
\phantom{\Delta N'}
=
\relmerge{N}{G}{\proj{\positive{\Delta O}}{V}} \deltaminus N
&
\text{%
   $\proj{O}{V} \subseteq N$
}
\notag
\\
&
\phantom{\Delta N'}
=
\relmerge{\select{P_N}{N}}{G}{\proj{\positive{\Delta O}}{V}} \deltaminus \select{P_N}{N}
&
\text{%
   (\localref{affected-right-query}); lemma~\ref{lem:relmerge_efficient}
}
\notag
\\[2mm]
&
L = (M_0 \Join N') \setminus O = (M \Join N) \setminus O = O \setminus O = \emptyset
&
\notag
\end{flalign}
\begin{flalign}
&
\Delta O'
=
(M, \Delta M_0) \deltaJoin (N, \Delta N')
&
\text{%
   define ($\Delta O'$)
}
\locallabel{def-delta-O-prime}
\\
&
\positive{\Delta O'}
=
((M \deltaapp \Delta M_0) \Join \positive{\Delta N'})
\cup
(\positive{\Delta M_0} \Join (N \deltaapp \Delta N'))
&
\text{%
	\reflemma{correctness-delta-ops} part 3
}
\notag
\\
&
\negative{\Delta O'}
=
(\negative{\Delta M_0} \Join N) \cup (M \Join \negative{\Delta N'})
&
\text{%
	\reflemma{correctness-delta-ops} part 3
}
\notag
\\[2mm]
&
M_0 \deltaapp \Delta M_0 = \relmerge{M}{F}{\proj{O \deltaapp \Delta O}{U}}
&
\text{%
   (\localref{def-M_0}); delta-correctness
}
\notag
\\
&
N' \deltaapp \Delta N' = \relmerge{N}{G}{\proj{O \deltaapp \Delta O}{V}}
&
\text{%
   (\localref{def-N-prime}); delta-correctness
}
\notag
\\
&
\proj{O \deltaapp \Delta O}{U} \subseteq M_0 \deltaapp \Delta M_0
&
\text{%
   def.~$\relmerge{\cdot}{F}{\cdot}$
}
\notag
\\
&
\proj{O \deltaapp \Delta O}{V} \subseteq N' \deltaapp \Delta N'
&
\text{%
   def.~$\relmerge{\cdot}{G}{\cdot}$
}
\notag
\\
&
\proj{O \deltaapp \Delta O}{U} \Join \proj{O \deltaapp \Delta O}{V}
\subseteq
(M_0 \deltaapp \Delta M_0) \Join (N' \deltaapp \Delta N')
&
\text{%
   $\Join$ monotone
}
\notag
\\
&
O \deltaapp \Delta O
\subseteq
(M_0 \deltaapp \Delta M_0) \Join (N' \deltaapp \Delta N')
&
\text{%
   $\Join$ after $\pi_U \times \pi_V$ incr.; trans.
}
\notag
\\
&
O \deltaapp \Delta O \subseteq O \deltaapp \Delta O'
&
\text{%
   $(\localref{def-delta-O-prime})$; delta-correctness
}
\locallabel{O-subseteq}
\end{flalign}
\begin{flalign}
&
\proj{L}{U} \subseteq M = M_0
&
\text{%
   $\pi_U$ after $\Join$ decr.
}
\locallabel{proj-L-U-subset-M}
\\
&
\proj{L}{U} \deltaapp \deltaproj{L, \Delta L}{U}
=
\proj{L \deltaapp \Delta L}{U}
&
\text{%
   delta-correctness
}
\notag
\\
&
\phantom{\proj{L}{U} \deltaapp \deltaproj{L, \Delta L}{U}}
=
\proj{((M_0 \deltaapp \Delta M_0) \Join (N' \deltaapp \Delta N')) \setminus (O \deltaapp \Delta O)}{U}
&
\text{%
   \text{(\localref{def-delta-L}); delta-correctness}
}
\notag
\\
&
\phantom{\proj{L}{U} \deltaapp \deltaproj{L, \Delta L}{U}}
\subseteq
\proj{(M_0 \deltaapp \Delta M_0) \Join (N' \deltaapp \Delta N')}{U}
&
\text{%
   \text{$\proj{\cdot}{U}$ monotone}
}
\notag
\\
&
\phantom{\proj{L}{U} \deltaapp \deltaproj{L, \Delta L}{U}}
\subseteq
M_0 \deltaapp \Delta M_0
&
\text{%
   \text{$\pi_U$ after $\Join$ decr.}
}
\locallabel{delta-proj-L-U-subset-M}
\\[2mm]
&
\Delta M'
=
(M_0, \Delta M_0) \deltasetminus (\proj{L}{U}, \deltaproj{L, \Delta L}{U})
&
\text{%
   def.~$\lift{\cdot}$
}
\notag
\\
&
\phantom{\Delta M'} = \Delta M_0 \deltaminus \deltaproj{L, \Delta L}{U}
&
\text{%
   (\localref{proj-L-U-subset-M}, \localref{delta-proj-L-U-subset-M});
   \reflemma{correctness-delta-ops} part 5
}
\notag
\\
&
\phantom{\Delta M'} = \Delta M_0
\deltaminus
(\proj{\positive{\Delta L}}{U} \setminus \proj{L}{U}, \proj{\negative{\Delta L}}{U} \setminus \proj{L \deltaapp \Delta L}{U})
&
\text{%
	\reflemma{correctness-delta-ops} part 2
}
\notag
\\
&
\phantom{\Delta M'}
=
\Delta M_0
\deltaminus
(\proj{\positive{\Delta L}}{U}, \proj{\negative{\Delta L}}{U} \setminus \proj{\positive{\Delta L}}{U})
&
\text{%
   $L = \emptyset$
}
\notag
\\
&
\phantom{\Delta M'}
=
\Delta M_0
\deltaminus
(\proj{\positive{\Delta L}}{U}, \proj{\negative{\Delta L}}{U})
&
\text{%
   $\positive{\Delta L} \cap \negative{\Delta L} = \emptyset$
}
\notag
\\
&
\phantom{\Delta M'}
=
\Delta M_0
\deltaminus
\proj{\Delta L}{U}
&
\text{%
   def.~$\proj{\cdot}{U}$ on deltas
}
\notag
\\[2mm]
&
\deltaOp{put}_{\ell'}((M, N), \Delta O)
=
(\Delta M', \Delta N')
&
\text{%
   def.~$\deltaOp{put}_{\ell'}$
}
\locallabel{def-put-v-prime}
\\[2mm]
&
\deltaOp{put}_{\ell'}((M, N), \Delta O)
=
\deltaOp{put}_{\ell}((M, N), \Delta O)
&
\text{%
   (\localref{def-put-v-prime}, \localref{def-put-v}); transitivity
}
\notag
\end{flalign}
\end{proof}

\subsection*{Proof of Theorem~\ref{thm:rename:optimised}}

\restaterenameoptimised*

\setcounter{equation}{0}
\proofContext{thm:rename:optimised}
\begin{proof}
\small
\begin{flalign}
&
N
=
\getOp_{\ell}(M)
&
\text{%
   suppose ($N, M$)
}
\notag
\\
&
(M', \Delta M')
=
\Deltarename{B/A}{N, \Delta N}
&
\text{%
   suppose ($N_{\opname{new}}, \Delta N_{\opname{new}}, \Delta N$)
}
\notag
\\
&
\phantom{(M', \Delta M')}
=
(\rename{B/A}{N},\deltarename{B/A}{N, \Delta N})
&
\text{
	def.~$\lift{\cdot}$
}
\locallabel{def-m-prime}
\\[2mm]
&
\deltaOp{put}_{\ell}
=
\Delta M'
&
\text{%
	def. $\deltaOp{put}_{\ell}$
}
\locallabel{def-put-v}
\\[2mm]
&
\Delta M'
=
\deltarename{B/A}{N, \Delta N}
&
\text{%
	(\localref{def-m-prime})
}
\notag
\\
&
\phantom{\Delta M'}
=
(\rename{B/A}{\positive{\Delta N}},\rename{B/A}{\negative{\Delta N}})
&
\text{%
	\reflemma{correctness-delta-ops} part 4
}
\locallabel{def-delta-m-prime}
\\[2mm]
&
\deltaOp{put}_{\ell'}
=
\Delta M'
&
\text{%
	def. $\deltaOp{put}_{\ell'}$
}
\notag
\\
&
\phantom{\deltaOp{put}_{\ell'}}
=
(\rename{B/A}{\positive{\Delta N}},\rename{B/A}{\negative{\Delta N}})
&
\text{%
	(\localref{def-delta-m-prime})
}
\locallabel{def-put-v-prime}
\\[2mm]
&
\deltaOp{put}_{\ell}
=
\deltaOp{put}_{\ell'}
&
\text{%
	(\localref{def-put-v-prime}, \localref{def-put-v}); transitivity
}
\notag
\end{flalign}
\end{proof}

\section{DBLP Example}
\label{appendix:dblp_example}

We use a parser to convert the given XML file into a set of PostgreSQL tables.
In particular we are interested in the conference proceedings as well as their inproceedings.

\begin{figure}[h]
	\centering
	\begin{subfigure}{1.0\textwidth}
		\begin{equation*}
			\tblfd{
				\begin{array}{c|cccc}
					& \hd{inproceedings} & \hd{title} & \hd{year} & \hd{proceedings} \\
					\hline
					& \tx{conf/pods/BohannonPV06} & \tx{Relational lenses: \dots} & 2006 & \tx{conf/pods/2006} \\
					& \vdots
				\end{array}
			}{ \hd{inproceedings} \to \hd{title}\ \hd{year}\ \hd{proceedings} }
		\end{equation*}
		\caption{The \tblname{inproceedings} table.}
	\end{subfigure}
	\begin{subfigure}{1.0\textwidth}
		\begin{equation*}
			\tblfd{
				\begin{array}{c|cc}
					& \hd{inproceedings} & \hd{author} \\
					\hline
					& \tx{conf/pods/BohannonPV06} & \tx{Aaron Bohannon} \\
					& \tx{conf/pods/BohannonPV06} & \tx{Benjamin C. Pierce} \\
					& \tx{conf/pods/BohannonPV06} & \tx{Jeffrey A. Vaughan} \\
					& \vdots
				\end{array}
			}{ }
		\end{equation*}
		\caption{The \tblname{inproceedings\_author} table.}
	\end{subfigure}
	\caption{The four tables used in our DBLP example.}
	\label{fig:dblp_tables}
\end{figure}

Example data of the tables used is shown in Figure \ref{fig:dblp_tables}.
The join of the \tblname{inproceedings} and \tblname{inproceedings\_author} tables produces a view with the columns \colname{author}, \colname{inproceedings}, \colname{title}, \colname{year}, \colname{proceedings}.
The functional dependencies of the joined table are shown in Figure \ref{fig:dblp_fun_deps}.

\begin{figure}
	\centering
	\begin{tikzpicture}[sibling distance=5em,
		node distance = 0.5cm and 1cm,
		every node/.style = {font={\small\itshape}}]
		\node (a) {author};
		\node[right=3em of a]{inproceedings}
			child { node {title} }
			child { node {year} }
			child { node {proceedings}
			};
	\end{tikzpicture}
	\caption{The functional dependencies in tree form of the DBLP example.}
	\label{fig:dblp_fun_deps}
\end{figure}
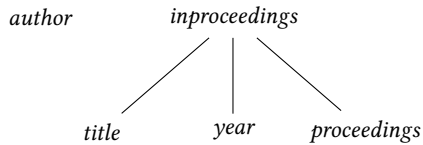

Retrieving the view of the lens using $\getOp$ results in table containing the entries as shown below.
We retrieve the view in links and make a change in the title to all entries with attribute \colname{inproceedings} = \tx{conf/pods/BohannonPV06}.
The updated view is applied to the database using the $\putOp$ operation.

\begin{equation*}
	\begin{array}{c|cccccc}
		& \hd{author} & \hd{inproceedings} & \hd{title} & \hd{year} & \hd{proceedings} \\
		\hline
		& \tx{Aaron Bohannon} & \tx{conf/pods/BohannonPV06} & \tx{Rel\dots} & 2006 & \tx{conf/pods/2006} \\
		& \tx{Benjamin C. Pierce} & \tx{conf/pods/BohannonPV06} & \tx{Rel\dots} & 2006 & \tx{conf/pods/2006} \\
		& \tx{Jeffrey A. Vaughan} & \tx{conf/pods/BohannonPV06} & \tx{Rel\dots} & 2006 & \tx{conf/pods/2006} \\
		& \vdots
	\end{array}
\end{equation*}

\end{techreport}

\end{document}